%% file: Main.tex
\documentclass[reqno]{amsart}
\usepackage[margin=1.5in,bottom=1.25in]{geometry}		

\usepackage{amsmath}		
\usepackage{amssymb}		
\usepackage{amsfonts}		
\usepackage{amsthm}		
\usepackage[foot]{amsaddr}		

\usepackage{mathtools}		

\mathtoolsset{%
}

\usepackage[utf8]{inputenc}		
\usepackage[T1]{fontenc}		

\usepackage[
cal=cm,
]
{mathalfa}

\usepackage{dsfont}		

\usepackage[proportional,tabular,lining,sf,mono=false]{libertine}

\usepackage{acronym}		
\renewcommand*{\aclabelfont}[1]{\acsfont{#1}}		
\newcommand{\acli}[1]{\emph{\acl{#1}}}		
\newcommand{\acdef}[1]{\emph{\acl{#1}} \textup{(\acs{#1})}\acused{#1}}		

\usepackage[labelfont={bf,small},labelsep=colon,font=small]{caption}	
\usepackage{subcaption}		
\usepackage[dvipsnames,svgnames]{xcolor}		
\colorlet{MyRed}{FireBrick!50!Crimson}
\colorlet{MyBlue}{DodgerBlue!75!black}
\colorlet{MyGreen}{DarkGreen!85!black}
\colorlet{MyViolet}{DarkMagenta}

\colorlet{MyLightBlue}{DodgerBlue!20}
\colorlet{MyLightGreen}{MyGreen!20}

\colorlet{PrimalColor}{MyBlue}
\colorlet{PrimalFill}{MyLightBlue}
\colorlet{DualColor}{MyRed}

\colorlet{RevColor}{blue}
\colorlet{LinkColor}{MediumBlue}

\newcommand{\afterhead}{.\;}		
\newcommand{\para}[1]{\medskip\paragraph{\textbf{#1\afterhead}}}
\newcommand{\itempar}[1]{\medskip\paragraph{\raisebox{1pt}{\small$\blacktriangleright$}\;\;\bfseries#1\afterhead}}

\usepackage{cancel}		
\usepackage{latexsym}		

\usepackage{pifont}		

\usepackage{tikz}		
\usepackage{tikz-cd}		
\usetikzlibrary{calc,patterns,decorations.markings}		

\usepackage{array}		
\usepackage{booktabs}		
\usepackage[inline,shortlabels]{enumitem}		
\setenumerate{itemsep=\medskipamount,topsep=\smallskipamount,left=\parindent}
\setitemize{itemsep=\medskipamount,topsep=\smallskipamount,left=\parindent}

\usepackage[kerning=true]{microtype}		

\usepackage{tabto}		
\usepackage{xspace}		

\usepackage[sort&compress]{natbib}		

\bibpunct[, ]{[}{]}{,}{n}{,}{,}

\newcommand{\citef}[2][]{\citeauthor{#2} \cite[#1]{#2}}

\usepackage{hyperref}
\hypersetup{
final,
colorlinks=true,
linktocpage=true,
pdfstartview=FitH,
breaklinks=true,
pdfpagemode=UseNone,
pageanchor=true,
pdfpagemode=UseOutlines,
plainpages=false,
bookmarksnumbered,
bookmarksopen=false,
bookmarksopenlevel=1,
hypertexnames=false,
pdfhighlight=/O,
urlcolor=LinkColor,linkcolor=LinkColor,citecolor=LinkColor,	
pdftitle={},
pdfauthor={},
pdfsubject={},
pdfkeywords={},
pdfcreator={pdfLaTeX},
pdfproducer={LaTeX with hyperref}
}

\newcommand{\EMAIL}[1]{\email{\href{mailto:#1}{#1}}}

\usepackage[sort&compress,capitalize,nameinlink]{cleveref}		
\crefname{algo}{Algorithm}{Algorithms}
\crefname{assumption}{Assumption}{Assumptions}
\crefname{figure}{Fig.}{Figs.}

\usepackage{algorithm}		
\usepackage{algpseudocode}		

\usepackage{thmtools}		
\usepackage{thm-restate}		

\theoremstyle{plain}
\newtheorem{theorem}{Theorem}		
\newtheorem{corollary}{Corollary}		
\newtheorem{lemma}{Lemma}		
\newtheorem{proposition}{Proposition}		

\newtheorem*{corollary*}{Corollary}		

\theoremstyle{definition}
\newtheorem{definition}{Definition}		
\newtheorem{assumption}{Assumption}		
\newtheorem{example}{Example}		
\newtheorem{algo}{Algorithm}		

\newtheorem*{definition*}{Definition}		
\newtheorem*{assumption*}{Assumption}		
\newtheorem*{blanket*}{Blanket assumption}		
\newtheorem*{example*}{Example}		

\theoremstyle{remark}
\newtheorem{remark}{Remark}		

\newtheorem*{remark*}{Remark}		

\def\endenv{\hfill{\small$\lozenge$}\smallskip}

\newcounter{proofpart}

\numberwithin{example}{section}		

\usepackage[showdeletions]{color-edits}		
\usepackage[normalem]{ulem}		

\newcommand{\debug}[1]{#1}		

\newcommand{\explain}[1]{\tag*{\#\:#1}}

\newcommand{\newmacro}[2]{\newcommand{#1}{\debug{#2}}}		
\newcommand{\newop}[2]{\DeclareMathOperator{#1}{\debug{#2}}}		

\DeclarePairedDelimiter{\braces}{\{}{\}}		
\DeclarePairedDelimiter{\bracks}{[}{]}		
\DeclarePairedDelimiter{\parens}{(}{)}		

\DeclarePairedDelimiter{\abs}{\lvert}{\rvert}		
\DeclarePairedDelimiter{\clip}{[}{]}		

\DeclarePairedDelimiterX{\setdef}[2]{\{}{\}}{#1:#2}		
\DeclarePairedDelimiterXPP{\exclude}[1]{\mathopen{}\setminus}{\{}{\}}{}{#1}

\newcommand{\N}{\mathbb{N}}		
\newcommand{\R}{\mathbb{R}}		

\DeclareMathOperator*{\argmax}{arg\,max}		
\DeclareMathOperator*{\argmin}{arg\,min}		
\DeclareMathOperator*{\intersect}{\bigcap}		
\DeclareMathOperator*{\union}{\bigcup}		

\DeclareMathOperator{\bigoh}{\mathcal{O}}		
\DeclareMathOperator{\diam}{diam}		
\DeclareMathOperator{\dist}{dist}		
\DeclareMathOperator{\dom}{dom}		
\DeclareMathOperator{\im}{im}		
\DeclareMathOperator{\Jac}{Jac}		
\DeclareMathOperator{\one}{\mathds{1}}		
\DeclareMathOperator{\proj}{pr}		
\DeclareMathOperator{\relint}{ri}		
\DeclareMathOperator{\unif}{unif}		

\newcommand{\cf}{cf.\xspace}		
\newcommand{\eg}{e.g.,\xspace}		
\newcommand{\ie}{i.e.,\xspace}		

\newcommand{\textpar}[1]{\textup(#1\textup)}		

\newcommand{\txs}{\textstyle}		

\newcommand{\alt}[1]{#1'}		
\newcommand{\altalt}[1]{#1''}		

\newmacro{\dd}{\:d}		
\newcommand{\ddt}{\frac{d}{dt}}		
\newcommand{\eps}{\varepsilon}		

\newmacro{\const}{c}		
\newmacro{\Const}{C}		
\newmacro{\coef}{\lambda}		

\newmacro{\param}{\theta}		
\newmacro{\params}{\Theta}		

\newmacro{\pexp}{p}		
\newmacro{\qexp}{q}		
\newmacro{\rexp}{r}		

\newmacro{\stepexp}{\ell_{\step}}		
\newmacro{\mixexp}{\ell_{\mix}}		
\newmacro{\biasexp}{\ell_{\bias}}		
\newmacro{\noisexp}{\ell_{\sdev}}		

\newmacro{\beforestart}{0}		
\newmacro{\start}{1}		
\newmacro{\afterstart}{2}		
\newmacro{\running}{\start,\afterstart,\dotsc}		
\newmacro{\halfrunning}{1/2,1,3/2,\dotsc}		

\newmacro{\run}{n}		
\newmacro{\runalt}{k}		
\newmacro{\runaltalt}{m}		
\newmacro{\nRuns}{T}		
\newmacro{\runs}{\mathcal{\nRuns}}		

\newmacro{\state}{X}		
\newmacro{\statealt}{Y}		
\newmacro{\statealtalt}{Z}		

\newcommand{\init}[1][\state]{\debug{#1}_{\start}}		

\newcommand{\iter}[1][\state]{\debug{#1}_{\runalt}}		

\newcommand{\prev}[1][\state]{\debug{#1}_{\run-1}}		
\newcommand{\curr}[1][\state]{\debug{#1}_{\run}}		
\renewcommand{\next}[1][\state]{\debug{#1}_{\run+1}}		

\newcommand{\beforelead}[1][\state]{\debug{#1}_{\run-1/2}}		
\newcommand{\lead}[1][\state]{\debug{#1}_{\run+1/2}}		

\newmacro{\tstart}{0}		
\renewcommand{\time}{\debug{t}}		
\newmacro{\timealt}{s}		
\newmacro{\horizon}{T}		

\newcommand{\trajof}[2][]{\point_{#1}(#2)}
\newcommand{\dtrajof}[2][]{\dpoint_{#1}(#2)}

\newmacro{\flow}{\phi}		
\DeclarePairedDelimiterXPP{\flowof}[2]{\flow_{#1}}{(}{)}{}{#2}		

\newmacro{\basin}{\mathcal{B}}		
\newmacro{\dbasin}{\mathcal{D}}		
\newmacro{\limset}{\mathcal{L}}		

\newop{\Nash}{NE}		
\newop{\CE}{CE}		
\newop{\CCE}{CCE}		
\newop{\NI}{NI}		

\newop{\brep}{br}		
\newop{\reg}{Reg}		
\newop{\preg}{\overline{Reg}}		
\newop{\val}{val}		

\newmacro{\play}{i}		
\newmacro{\playalt}{j}		
\newmacro{\playaltalt}{k}		
\newmacro{\nPlayers}{N}		
\newmacro{\players}{\mathcal{\nPlayers}}		

\newmacro{\pure}{\alpha}		
\newmacro{\purealt}{\beta}		
\newmacro{\purealtalt}{\gamma}		
\newmacro{\nPures}{A}		
\newmacro{\pures}{\mathcal{\nPures}}		

\newmacro{\strat}{x}		
\newmacro{\stratalt}{\alt\strat}		
\newmacro{\strataltalt}{\altalt\strat}		
\newmacro{\strats}{\mathcal{X}}		
\newmacro{\intstrats}{\strats^{\circle}}		

\newcommand{\eq}{\sol}		
\newcommand{\eqs}{\sol[\points]}		

\newmacro{\loss}{\ell}		
\newmacro{\pay}{u}		
\newmacro{\payv}{v}		
\newmacro{\pot}{\Phi}		

\newmacro{\game}{\mathcal{G}}		
\newmacro{\gamefull}{\game(\players,\points,\pay)}		

\newmacro{\fingame}{\Gamma}		
\newmacro{\fingamefull}{\Gamma(\players,\pures,\pay)}		
\newmacro{\mixgame}{\simplex(\fingame)}

\newmacro{\gmat}{g}		
\newmacro{\gdist}{\dist_{\gmat}}
\newmacro{\mfld}{\mathcal{M}}		
\newmacro{\form}{\omega}		

\newmacro{\tvec}{z}		
\newmacro{\uvec}{u}		

\newmacro{\ball}{\mathbb{B}}		
\newmacro{\sphere}{\mathbb{S}}		

\newmacro{\vertex}{v}		
\newmacro{\vertexalt}{\alt\vertex}		
\newmacro{\vertexaltalt}{\altalt\vertex}		
\newmacro{\nVertices}{V}		
\newmacro{\vertices}{\mathcal{\nVertices}}		

\newmacro{\edge}{e}		
\newmacro{\edgealt}{\alt\edge}		
\newmacro{\edgealtalt}{\altalt\edge}		
\newmacro{\nEdges}{E}		
\newmacro{\edges}{\mathcal{\nEdges}}		

\newmacro{\graph}{\mathcal{G}}		
\newmacro{\graphall}{\graph(\vertices,\edges)}		

\newmacro{\vecspace}{\mathcal{Z}}		
\newmacro{\subspace}{\mathcal{W}}		

\newmacro{\bvec}{e}		
\newmacro{\bvecs}{\mathcal{E}}		

\newmacro{\coord}{k}		
\newmacro{\coordalt}{j}		
\newmacro{\coordaltalt}{l}		
\newmacro{\nCoords}{d}		
\newmacro{\dims}{\nCoords}		
\newmacro{\vdim}{\nCoords}		

\newmacro{\pvec}{z}		
\newmacro{\dvec}{w}		

\newmacro{\pspace}{\mathcal{V}}		
\newmacro{\dspace}{\mathcal{Y}}		

\newmacro{\ppoint}{x}		
\newmacro{\ppointalt}{\alt\ppoint}		
\newmacro{\ppointaltalt}{\altalt\ppoint}		
\newmacro{\ppoints}{\mathcal{X}}		
\newmacro{\pstate}{X}		

\newmacro{\dpoint}{y}		
\newmacro{\dpointalt}{\alt\dpoint}		
\newmacro{\dpointaltalt}{\altalt\dpoint}		
\newmacro{\dpoints}{\mathcal{Y}}		
\newmacro{\dstate}{Y}		

\newmacro{\dset}{\mathcal{D}}		
\newmacro{\dnhd}{\mathcal{W}}		
\newmacro{\dlyap}{F}		

\newmacro{\mat}{M}		
\newmacro{\hmat}{H}		

\newmacro{\ones}{\mathbf{1}}		
\newmacro{\eye}{I}		
\newmacro{\zer}{\mathbf{0}}		

\DeclarePairedDelimiter{\norm}{\lVert}{\rVert}		
\DeclarePairedDelimiterXPP{\dnorm}[1]{}{\lVert}{\rVert}{_{\ast}}{#1}		

\DeclarePairedDelimiterXPP{\onenorm}[1]{}{\lVert}{\rVert}{_{1}}{#1}		
\DeclarePairedDelimiterXPP{\twonorm}[1]{}{\lVert}{\rVert}{_{2}}{#1}		
\DeclarePairedDelimiterXPP{\supnorm}[1]{}{\lVert}{\rVert}{_{\infty}}{#1}		

\DeclarePairedDelimiterX{\braket}[2]{\langle}{\rangle}{#1,#2}		

\DeclarePairedDelimiterX{\inner}[2]{\langle}{\rangle}{#1,#2}		

\newcommand{\defeq}{\coloneqq}		
\newcommand{\eqdef}{\eqqcolon}		

\newcommand{\from}{\colon}		

\newop{\Opt}{Opt}		
\newop{\Sol}{Sol}		
\newop{\gap}{Gap}		

\newmacro{\orcl}{V}		
\newop{\err}{err}		
\newop{\IWE}{IWE}		

\newmacro{\obj}{f}		
\newmacro{\objalt}{g}		
\newmacro{\sobj}{F}		

\newmacro{\gvec}{g}		
\newmacro{\gbound}{G}		

\newcommand{\sol}[1][\point]{#1^{\ast}}		
\newcommand{\sols}{\sol[\points]}		

\newmacro{\vecfield}{v}		
\newmacro{\oper}{A}		
\newmacro{\vbound}{\gbound}		

\newmacro{\lips}{L}		
\newmacro{\strong}{\alpha}		
\newmacro{\smooth}{\beta}		

\newop{\tspace}{T}		
\newop{\tcone}{TC}		
\newop{\dcone}{\tcone^{\ast}}		
\newop{\ncone}{NC}		
\newop{\pcone}{PC}		
\newop{\hull}{\Delta}		

\newmacro{\cvx}{\mathcal{C}}		
\newmacro{\subd}{\partial}		

\newmacro{\minmax}{\mathcal{L}}		

\newmacro{\minvar}{\point_{1}}		
\newmacro{\minvaralt}{\alt\minvar}		
\newmacro{\minvars}{\points_{1}}		

\newmacro{\maxvar}{\point_{2}}		
\newmacro{\maxvaralt}{\alt\maxvar}		
\newmacro{\maxvars}{\points_{2}}		

\newop{\Eucl}{\Pi}		
\newop{\logit}{\Lambda}		
\newop{\dkl}{KL}		

\newmacro{\hreg}{h}		
\newmacro{\hconj}{\hreg^{\ast}}		
\newmacro{\breg}{D}		
\newmacro{\mprox}{P}		
\newmacro{\mirror}{Q}		
\newmacro{\fench}{F}		
\newmacro{\hstr}{K}		

\newmacro{\proxdom}{\points_{\hreg}}		
\newmacro{\proxdomi}{\points_{\hreg_{\play}}}		

\DeclarePairedDelimiterXPP{\proxof}[2]{\mprox_{#1}}{(}{)}{}{#2}		

\newmacro{\zone}{\mathbb{D}}		

\newmacro{\point}{x}		
\newmacro{\pointalt}{\alt\point}		
\newmacro{\pointaltalt}{\altalt\point}		
\newmacro{\points}{\mathcal{X}}		
\newmacro{\intpoints}{\relint\points}		

\newmacro{\base}{p}		
\newmacro{\basealt}{q}		
\newmacro{\basealtalt}{u}		

\newmacro{\open}{\mathcal{U}}		
\newmacro{\closed}{\mathcal{C}}		
\newmacro{\cpt}{\mathcal{K}}		

\newmacro{\set}{\mathcal{S}}		
\newmacro{\nhd}{\mathcal{U}}		

\newop{\ex}{\mathbb{E}}		
\newop{\prob}{\mathbb{P}}		
\newop{\Var}{Var}		
\newop{\simplex}{\hull}		

\providecommand\given{}		

\DeclarePairedDelimiterXPP{\exof}[1]{\ex}{[}{]}{}{
\renewcommand\given{\nonscript\,\delimsize\vert\nonscript\,\mathopen{}} #1}

\DeclarePairedDelimiterXPP{\probof}[1]{\prob}{(}{)}{}{
\renewcommand\given{\nonscript\:\delimsize\vert\nonscript\:\mathopen{}} #1}

\DeclarePairedDelimiterXPP{\oneof}[1]{\one}{(}{)}{}{
\renewcommand\given{\nonscript\,\delimsize\vert\nonscript\,\mathopen{}} #1}

\newmacro{\seed}{\theta}		
\newmacro{\seeds}{\Theta}		
\newmacro{\pdist}{P}		
\newmacro{\history}{\mathcal{H}}		

\newmacro{\sample}{\omega}		
\newmacro{\samples}{\Omega}		

\newmacro{\filter}{\mathcal{F}}		
\newmacro{\probspace}{(\samples,\filter,\prob)}		

\newcommand{\as}{\debug{\textpar{a.s.}}\xspace}		
\newmacro{\event}{\mathcal{E}}       
\newmacro{\eventalt}{\mathcal{H}}       
\newmacro{\mean}{\mu}		
\newmacro{\sdev}{\sigma}		
\newmacro{\variance}{\sdev^{2}}		

\newcommand{\est}[1]{\hat #1}		

\newmacro{\signal}{\hat\vecfield}		
\newmacro{\step}{\gamma}		
\newmacro{\learn}{\eta}		

\newmacro{\efftime}{\tau}		
\newcommand{\apt}[2][]{\dstate_{#1}(#2)}		
\newcommand{\papt}[2][]{\state_{#1}(#2)}		

\newmacro{\noise}{U}		
\newmacro{\snoise}{\xi}		
\newmacro{\noisepar}{\sdev}		
\newmacro{\noisevar}{\variance}		
\newmacro{\aggnoise}{\mathrm{\uppercase\expandafter{\romannumeral1}}}		
\newmacro{\supnoise}{\aggnoise_{\infty}}		
\newmacro{\maxnoise}{\aggnoise^{\ast}}		

\newmacro{\bias}{b}		
\newmacro{\bbound}{B}		
\newmacro{\sbias}{\chi}		
\newmacro{\aggbias}{\mathrm{\uppercase\expandafter{\romannumeral2}}}		
\newmacro{\supbias}{\aggbias_{\infty}}		
\newmacro{\maxbias}{\aggbias^{\ast}}		

\newmacro{\second}{\psi}		
\newmacro{\sbound}{M}		
\newmacro{\aggsecond}{\mathrm{\uppercase\expandafter{\romannumeral3}}}		
\newmacro{\supsecond}{\aggsecond_{\infty}}		
\newmacro{\maxsecond}{\aggsecond^{\ast}}		

\newmacro{\mix}{\delta}		
\newmacro{\unitvec}{w}		
\newmacro{\unitvar}{W}		
\newmacro{\perturb}{z}		

\newmacro{\purequery}{\est\pure}		
\newmacro{\query}{\est\state}		
\newmacro{\pivot}{\point}		
\newmacro{\querypoint}{\est\point}		

\addauthor[Pan]{PM}{Crimson}

\DeclareMathOperator{\sech}{sech}

\newmacro{\kerspace}{\mathcal{W}}

\newmacro{\pureq}{\eq[\pure]}

\newop{\probalt}{\mathbb{Q}}		
\newmacro{\good}{\event}
\newmacro{\bad}{\mathcal{N}}
\newmacro{\lyap}{\pot}
\newmacro{\gauge}{\varphi}

\newmacro{\energy}{E}
\newmacro{\elvl}{a}
\newmacro{\ebound}{H}
\newmacro{\esmooth}{\beta}
\newmacro{\hien}{\energy_{+}}
\newmacro{\loen}{\energy_{-}}
\newmacro{\emax}{\hien}
\newmacro{\ground}{0}
\newmacro{\ediff}{\Delta\elvl}

\newmacro{\lvl}{c}
\newcommand{\sublvl}[2]{\debug{L}_{#1}^{-}(#2)}		
\newcommand{\suplvl}[2]{\debug{L}_{#1}^{+}(#2)}		

\newmacro{\conf}{\eta}		
\newmacro{\stoptime}{N}		

\newmacro{\texp}{\mu}		
\newmacro{\bexp}{b}		
\newmacro{\sexp}{s}		

\newmacro{\dev}{z}		
\newmacro{\iDev}{k}		
\newmacro{\nDevs}{m}		
\newmacro{\devs}{\mathcal{Z}}		

\newmacro{\pflow}{\chi}		
\DeclarePairedDelimiterXPP{\pflowof}[2]{\pflow_{#1}}{(}{)}{}{#2}		

\newmacro{\dflow}{\psi}		
\DeclarePairedDelimiterXPP{\dflowof}[2]{\dflow_{#1}}{(}{)}{}{#2}		

\newmacro{\score}{\dpoint}		
\newmacro{\cone}{\mathcal{C}}		
\newmacro{\poly}{\mathcal{P}}		

\addauthor[Ya-Ping]{YPH}{DarkGreen}

\addauthor[Volkan]{VC}{Purple}

\begin{document}


\title
[A unified stochastic approximation framework for learning in games]
{A unified stochastic approximation\\framework for learning in games}		

\author
[P.~Mertikopoulos]
{Panayotis Mertikopoulos$^{\ast}$}
\address{$^{\ast}$\,%
Univ. Grenoble Alpes, CNRS, Inria, Grenoble INP, LIG, 38000 Grenoble, France.}
\EMAIL{panayotis.mertikopoulos@imag.fr}
\author
[Y.~P.~Hsieh]
{Ya-Ping Hsieh$^{\S}$}
\address{$^{\S}$\,%
Institute for Machine Learning, Universitaetstrasse 6, 8092 Zurich, Switzerland.}
\EMAIL{yaping.hsieh@inf.ethz.ch}
\author
[V.~Cevher]
{Volkan Cevher$^{\sharp}$}
\address{$^{\sharp}$\,%
Laboratory for Information and Inference Systems, IEL STI EPFL, 1015 Lausanne, Switzerland.}
\EMAIL{volkan.cevher@epfl.ch}



\thanks{The authors are grateful to Victor Boone, Pierre-Louis Cauvin, Angeliki Giannou, Kyriakos Lotidis, Sylvain Sorin, and Manolis Vlatakis for many fruitful discussions.
Part of this work was done while P.~Mertikopoulos was visiting the Simons Institute for the Theory of Computing.}

\newacro{LHS}{left-hand side}
\newacro{RHS}{right-hand side}
\newacro{iid}[i.i.d.]{independent and identically distributed}
\newacro{lsc}[l.s.c.]{lower semi-continuous}
\newacro{GFO}{generalized first-order oracle}
\newacro{SFO}{stochastic first-order oracle}
\newacro{whp}[w.h.p]{with high probability}
\newacro{wp1}[w.p.$1$]{with probability $1$}
\newacro{MSE}{mean squared error}
\newacro{ODE}{ordinary differential equation}

\newacro{GAN}{generative adversarial network}

\newacro{KKT}{Karush\textendash Kuhn\textendash Tucker}
\newacroplural{KKT}[KKT]{Karush\textendash Kuhn\textendash Tucker}
\newacro{FOS}{first-order stationarity}
\newacro{SOS}{second-order stationarity}
\newacro{SOSC}{second-order sufficient condition}
\newacro{NE}{Nash equilibrium}
\newacroplural{NE}[NE]{Nash equilibria}
\newacro{LNE}{local Nash equilibrium}
\newacroplural{LNE}[LNE]{local Nash equilibria}
\newacro{VI}{variational inequality}
\newacroplural{VI}{variational inequalities}
\newacro{SVI}{Stampacchia variational inequality}
\newacroplural{SVI}{Stampacchia variational inequalities}
\newacro{MVI}{Minty variational inequality}
\newacroplural{MVI}{Minty variational inequalities}

\newacro{ES}{evolutionarily stable}
\newacro{VS}{variationally stable}
\newacro{NS}{neutrally stable}
\newacro{GVS}{globally variationally stable}
\newacro{GNS}{globally neutrally stable}
\newacro{DSC}{diagonal strict concavity}

\newacro{APT}{asymptotic pseudotrajectory}
\newacroplural{APT}{asymptotic pseudotrajectories}
\newacro{ICT}{internally chain transitive}
\newacroplural{ICT}{internally chain transitive sets}

\newacro{DA}{dual averaging}
\newacro{FP}{fictitious play}
\newacro{FTRL}{``follow-the-regularized-leader''}
\newacro{FTGL}{``follow-the-generalized-leader''}
\newacro{EW}[\textsc{Hedge}]{exponential\,/\,multiplicative weights}
\newacro{MWU}{multiplicative weights update}
\newacro{IWE}{importance-weighted estimator}
\newacro{EXP3}{exponential weights for exploration and exploitation}
\newacro{EGD}{entropic gradient descent}
\newacro{GA}{gradient ascent}
\newacro{EG}{extra-gradient}
\newacro{OG}{optimistic gradient}
\newacro{SEG}{stocahsatic extra-gradient}
\newacro{BMD}{bandit mirror descent}
\newacro{DE}{dual extrapolation}
\newacro{MD}{mirror descent}
\newacro{MP}{mirror-prox}
\newacro{OGD}{online gradient descent}
\newacro{OMD}{optimistic mirror descent}
\newacro{OMW}{optimistic multiplicative weights}
\newacro{GDA}{gradient descent\,/\,ascent}
\newacro{SGDA}{stochastic gradient descent/ascent}
\newacro{IGD}[iGD]{individual gradient descent}
\newacro{seqGA}{sequential gradient ascent}
\newacro{SG}{stochastic gradient}
\newacro{SGD}{stochastic gradient descent}
\newacro{SGA}{stochastic gradient ascent}
\newacro{SPSA}{single-point stochastic approximation}
\newacro{DGA}{dampened gradient approximation}

\newacro{method}[RRM]{regularized Robbins\textendash Monro}
\newacro{MRM}{mirrored Robbins\textendash Monro}
\newacro{MAD}{mirror ascent dynamics}

\newacro{HR}{Hessian\textendash Riemannian}
\newacro{RM}{Robbins\textendash Monro}
\newacro{BDG}{Burkholder\textendash Davis\textendash Gundy}

\begin{abstract}
\input{Abstract}
\end{abstract}
\maketitle

\allowdisplaybreaks		
\acresetall		

\section{Introduction}
\label{sec:introduction}
\input{Body/Introduction}
\acresetall

\section{Preliminaries}
\label{sec:prelims}
\input{Body/Prelims}

\section{The learning framework}
\label{sec:framework}
\input{Body/Framework}

\section{Stochastic approximation and first results}
\label{sec:general}
\input{Body/General}
\section{Robust convergence and stable limit sets}
\label{sec:primal}
\input{Body/Primal}

\section{Fast convergence to coherent sets}
\label{sec:sharp}
\input{Body/Sharp}

\section{Concluding remarks}
\label{sec:discussion}
\input{Body/Discussion}

\numberwithin{lemma}{section}		
\numberwithin{corollary}{section}		
\numberwithin{proposition}{section}		
\numberwithin{equation}{section}		
\appendix

\section{Regularizers and mirror maps}
\label{app:mirror}
\input{Appendix/App-Mirror.tex}

\section{Omitted proofs and calculations}
\label{app:proofs}
\input{Appendix/App-Proofs.tex}

\section*{Acknowledgments}
\begingroup
\small
\input{Thanks.tex}
\endgroup

\bibliographystyle{icml}
\bibliography{bibtex/IEEEabrv,Bibliography,bibtex/Bibliography-YPH}

\end{document}

%% file: Abstract.tex
%
%
We develop a flexible stochastic approximation framework for analyzing the long-run behavior of learning in games (both continuous and finite).
The proposed analysis template incorporates a wide array of popular learning algorithms, including gradient-based methods, the \acl{EW} algorithm for learning in finite games, optimistic and bandit variants of the above, etc.
In addition to providing an integrated view of these algorithms, our framework further allows us to obtain several new convergence results, both asymptotic and in finite time, in both continuous and finite games.
Specifically, we provide a range of criteria for identifying classes of \aclp{NE} and sets of action profiles that are attracting with high probability, and we also introduce the notion of \emph{coherence}, a game-theoretic property that includes strict and sharp equilibria, and which leads to convergence in finite time.
Importantly, our analysis applies to both oracle-based and bandit, payoff-based methods \textendash\ that is, when players only observe their realized payoffs.

%% file: Body/Introduction.tex

The prototypical setting of online learning in games can be summarized as follows:
\begin{enumerate}
\item
At each stage of a repeated decision process, every player selects an action.
\item
The players receive a reward determined by their chosen actions and their individual payoff functions \textendash\ assumed a priori unknown.
\item
Based on these payoffs and any other observed information, the players update their actions and the process repeats.
\end{enumerate}
\smallskip

A key question that arises in this general setting is whether the players eventually settle down to a stable profile from which no player has an incentive to deviate.
Put differently:
\begin{quote}
\centering
\emph{Does the players' learning process converge to a \acl{NE}?}
\end{quote}
This question has been at the forefront of game-theoretic research ever since the field's earliest steps, and it has recently received renewed attention owing to its connection to data science, multi-agent reinforcement learning, networks, and many other applications where agents are called to make decisions under uncertainty.
The first positive answer here was given by \citet{Bro51} and \citet{Rob51} who introduced the so-called \acf{FP} process and established its convergence in $2$-player zero-sum finite games.
Since then, a vast number of works have examined the convergence of a diverse array of learning procedures in different classes of games:
smoothed versions of \acl{FP} in potential, zero-sum, supermodular and $2\times m$ games \citep{HS02,LC06},
gradient methods in continuous min-max games \cite{AHU58,Kor76,Pop80},
the numerous variants of \acf{MD} and other regularized learning schemes in monotone \citep{Nes09,JNT11,MZ19,MLZF+19,TK19,TK19-CDC}, smooth \citep{SALS15}, and potential games \citep{LC05,CGM15,BBF20},
etc.

At the same time, the well-known impossibility results of \citet{HMC03,HMC06} rule out the prospect of an unconditionally positive answer:
there is no uncoupled learning rule \textendash\ deterministic \emph{or} stochastic \textendash\ that converges to \acl{NE} in all games.
As a result, contemporary research on the subject has focused on
extending the classes of games in which positive results can be obtained,
relaxing the feedback requirements of the players' learning process,
and
understanding the convergence failures of popular learning algorithms.
This has in turn revealed a very fragile convergence landscape:
for example, standard gradient methods are known to converge in \emph{strictly} monotone games \citep{MZ19}, but they may diverge in bilinear min-max games (which are monotone but not strictly so) \citep{DISZ18};
this failure can be overcome by means of an extra-gradient\,/\,optimistic correction term \citep{Kor76,Pop80}, but it re-emerges in the presence of randomness and uncertainty \citep{HIMM20};
and if the game is perturbed even slightly, all these methods \textendash\ gradient, extra-gradient and optimistic \textendash\ may end up converging to a spurious limit cycle containing no critical / equilibrium points whatsoever \citep{MPP18,HMC21}.

Since these negative results are all pointwise, it is natural to turn to \emph{sets} and instead ask:
\begin{center}
\itshape
Which sets of actions are stable and attracting under a given learning process?\\
Are these sets robust to the choice of method, initialization, or available information?
\end{center}
From a dynamic standpoint, the established notion of stability is that of an \emph{attractor}, which characterizes outcomes that are resilient to small perturbations in the dynamics' initialization.
However, the questions above call for much more:
ideally, the sets under consideration should be stable in a class of learning procedures which, other than a few broad unifying features, may have radically different update structures, feedback requirements, etc.
Our aim in this paper is to identify such sets and to quantify their stability and convergence properties.

\para{Our contributions}

The basis of our analysis is a flexible stochastic approximation framework which we call the \acdef{method} template in reference to the seminal method of \citet{RM51} and the \ac{FTRL} family of algorithms of \citet{SSS06}.
This framework hinges on an implicit regularization mechanism in the spirit of \citet{Nes09} and encompasses as special cases many popular learning algorithms:
gradient-based methods for continuous games \cite{AHU58,Nes09,MZ19},
the \acl{EW} family of algorithms for finite games \cite{Vov90,LW94,ACBFS95},
optimistic \cite{Kor76,Pop80,RS13-NIPS,MLZF+19,DISZ18,GBVV+19} and bandit, payoff-based variants of the above \cite{BLM18,TK19,HCM17,CGM15},
etc.
We then seek to analyze the long-run behavior of this ``parent scheme'' via a suitable dynamical system which captures its mean, continuous-time limit, and which is sufficiently rich to accommodate different types of feedback and update structures.

In this general context, our main results can be summarized along the following axes:
\begin{enumerate}
[\bfseries 1.]
\item
\textbf{Characterization of limit sets:}
First, we show that the limit sets of \ac{method} methods are \ac{ICT} in the associated mean dynamics, \ie they are invariant and contain no smaller attractors.
This property applies to all games satisfying a certain coercivity condition \textendash\ which we call ``\emph{subcoercivity}'' \textendash\ and it allows us to deduce a series of almost sure equilibrium convergence results for min-max and potential games.
\item
\textbf{Characterization of attractors:}
We further show that sets that admit a local energy function (relative to the mean dynamics mentioned above) are attracting with high probability \textendash\ or globally attracting \acl{wp1}, depending on the energy function.
As a corollary of this result, we readily infer convergence to \acl{NE} in all strictly monotone games, and we likewise derive a series of high probability convergence results to equilibria that satisfy a certain variational stability requirement.
\item
\textbf{Fast convergence to coherent sets:}
Finally, we introduce the notion of \emph{coherence} \textendash\ an algorithm-agnostic concept which covers strict \aclp{NE} in finite games, sharp equilibria in continuous games, linear programs, etc. \textendash\ and we show that \ac{method} methods converge to such sets under significantly weaker conditions for their runtime parameters (step-size, sampling radius, etc.).
In addition, we show that projection-based methods (as opposed to interior-valued ones) converge to coherent sets in a \emph{finite} number of iterations.
\end{enumerate}
\smallskip

An appealing feature of our analysis is that it applies to both \emph{first-order} (``oracle-based'') and \emph{zeroth-order} (``payoff-based'' or ``bandit'') methods.
More to the point, our results
can be easily adapted to many other learning algorithms in the literature, reducing in this way the number of ad hoc elements required to analyze a given method.
Of course, given the breadth of the relevant literature, it is impossible to include here \emph{all} methods covered by the proposed \ac{method} template \textendash\ or that \emph{could} be covered modulo minor modifications.
Our choice of examples is only meant to illustrate different trends in the literature, and to show how some algorithms that initially seem unrelated \textendash\ like the \ac{DGA} method of \cite{BBF20} \textendash\ can be included in our framework.

\para{Paper outline}

In \cref{sec:prelims}, we introduce the game-theoretic background of our work, including the various solution concepts that we use throughout our paper (critical points, \aclp{NE}, \acl{VS} states, etc.).
Subsequently, in \cref{sec:framework}, we introduce a range of well-known algorithms for learning in games, and we show how they can be seen as special instances of the \ac{method} blueprint.
Our analysis proper begins in \cref{sec:general}, where we introduce the notion of subcoercivity and present our \ac{ICT} convergence results.
Subsequently, in \cref{sec:primal,sec:sharp}, we state and prove our main convergence results for stochastically attracting and coherent sets respectively.

%% file: Body/Prelims.tex

\subsection*{Notation}

In what follows,
$\pspace$ will denote a $\vdim$-dimensional real space with norm $\norm{\cdot}$.
We will also write
$\dspace \defeq \pspace^{\ast}$ for the dual space of $\pspace$,
$\braket{\dpoint}{\point}$ for the canonical pairing between $\dpoint\in\dspace$ and $\point\in\pspace$,
and
$\dnorm{\dpoint} \defeq \max \setdef{\braket{\dpoint}{\point}}{\norm{\point}\leq 1}$ for the induced dual norm on $\dspace$.
As is customary, if $\pspace$ is Euclidean, we will not distinguish between primal and dual vectors.
Finally, if $\obj\from\pspace\to\R\cup\{\infty\}$ is
a convex function on $\pspace$, we will write
$\dom\obj \defeq \setdef{\point\in\pspace}{\obj(\point) < \infty}$ for its effective domain,
$\subd\obj(\point) \defeq \setdef{\dpoint\in\dspace}{\obj(\pointalt) \geq \obj(\point) + \braket{\dpoint}{\pointalt - \point} \; \text{for all $\pointalt\in\pspace$}}$ for the subdifferential of $\obj$ at $\point$,
and
$\dom\subd\obj \defeq \setdef{\point\in\pspace}{\subd\obj(\point) \neq \varnothing}$ for the domain of subdifferentiability of $\obj$.
\acused{whp}
\acused{wp1}

\subsection{Games in normal form}

Throughout the sequel, we will focus on games with a finite number of players $\play \in \players=\{1,\dotsc,\nPlayers\}$, each selecting an \emph{action} $\point_{\play}$ from some closed convex subset $\points_{\play}$ of a $\vdim_{\play}$-dimensional normed space $\pspace_{\play}$.
Gathering all players together, we will write $\points = \prod_{\play}\points_{\play}$ for the space of all \emph{action profiles} $\point = (\point_{\play})_{\play\in\players}$
and $\vdim = \sum_{\play} \vdim_{\play}$ for the dimension of the ambient space $\pspace = \prod_{\play} \pspace_{\play}$.
Finally, when we want to distinguish between the action of the $\play$-th player and that of all other players, we will employ the shorthand $(\point_{\play};\point_{-\play})$.

Given an action profile $\point\in\points$, each player $\play\in\players$ is assumed to receive a reward $\pay_{\play}(\point) \equiv \pay_{\play}(\point_{\play};\point_{-\play})$ based on an associated \emph{payoff function} $\pay_{\play}\from\points\to\R$.
In terms of regularity, we will tacitly assume that $\pay_{\play}$ is differentiable and we will write
\begin{equation}
\label{eq:vecfield}
\vecfield_{\play}(\point)
	= \nabla_{\point_{\play}} \pay_{\play}(\point_{\play};\point_{-\play})
	\quad
	\text{and}
	\quad
\vecfield(\point)
	= (\vecfield_{\play}(\point))_{\play\in\players}
\end{equation}
for the players' \emph{individual payoff gradients} and the ensemble thereof.
Finally, unless explicitly mentioned otherwise, we will treat each $\vecfield_{\play}(\point)$ as an element of the corresponding dual space $\dspace_{\play} = \pspace_{\play}^{\ast}$ of $\pspace_{\play}$,
and we will make the following blanket assumption:

\begin{assumption}
\label{asm:game}
The players' payoff functions are \emph{Lipschitz continuous and smooth}, \ie
there exist constants $\vbound_{\play},\lips_{\play}\geq0$, $\play\in\players$, such that
\begin{equation}
\label{eq:Lips}
\abs{\pay_{\play}(\pointalt) - \pay_{\play}(\point)}
	\leq \vbound_{\play} \norm{\pointalt - \point}
	\quad
	\text{and}
	\quad
\dnorm{\nabla\pay_{\play}(\pointalt) - \nabla\pay_{\play}(\point)}
	\leq \lips_{\play} \norm{\pointalt - \point}.
\end{equation}
for all $\point,\pointalt\in\points$, $\play\in\players$.
For concision, we will also write $\vbound \defeq \max_{\play}\vbound_{\play}$ and $\lips \defeq \max_{\play}\lips_{\play}$.
\end{assumption}

A \emph{continuous game in normal form} is then defined as a tuple $\game \equiv \gamefull$ with players, actions and payoff functions as above.
For concreteness, we provide some examples below:

\begin{example}[Min-max games]
\label{ex:minmax}

Consider two players $\play\in\{1,2\}$ with action spaces $\minvars$ and $\maxvars$, and payoff functions $\pay_{1} = -\minmax = -\pay_{2}$ for some smooth function $\minmax\from\points_{1}\times\points_{2}\to\R$.
Player $1$ (the ``min'' player) seeks to minimize $\minmax = -\pay_{1}$ whereas Player $2$ (the ``max'' player) seeks to maximize $\minmax = \pay_{2}$.
In many applications, $\minmax$ is (strictly) convex-concave, in which case von Neumann's theorem asserts that the game $\min_{\minvar\in\minvars}\max_{\maxvar\in\maxvars} \minmax(\minvar,\maxvar)$ always admits a solution if $\points_{1}\times\points_{2}$ is compact.
\endenv
\end{example}

\begin{example}[Cournot oligopolies]
\label{ex:Cournot}

Consider $\nPlayers$ firms supplying the market with a quantity $\point_{\play}\in[0,C_{\play}]$ of some good up to each firm's capacity $C_{\play}$.
The good is priced as a function $P(\point) = a - b \sum_{\play=1}^{\nPlayers}\point_{\play}$ of the total quantity of the good in the market, so the net utility of the $\play$-th firm is $\pay_{\play}(\point) = \point_{\play} P(\point) - c_{\play} \point_{\play}$ where $a$, $b$ and $c_{i}$ are market-related positive constants.
The resulting game $\gamefull$ is known as a Cournot competition game and it plays a central role in economic theory.
\endenv
\end{example}


\begin{example}[Finite games]
\label{ex:mixed}
In a finite game $\fingame \equiv \fingamefull$, each player $\play\in\players$ chooses an action $\pure_{\play}$ from some finite set $\pures_{\play}$;
the players' payoffs are then determined by the action profile $\pure = (\pure_{1},\dotsc,\pure_{\nPlayers}) \in \pures \defeq \prod_{\play} \pures_{\play}$ and an ensemble of payoff functions $\pay_{\play}\from\pures\to\R$, $\play=1,\dotsc,\nPlayers$.
In the \emph{mixed extension} of $\fingame$, a player may pick an action according to a probability distribution $\point_{\play} \in \simplex(\pures_{\play})$:
this is known as a \emph{mixed strategy}, and the corresponding mixed payoff to the $\play$-th player is $\pay_{\play}(\point) = \sum_{\pure\in\pures} \point_{\pure} \pay_{\play}(\pure)$ where $\point_{\pure} = \prod_{\play} \point_{\play\pure_{\play}}$ is the probability of the action profile $\pure = (\pure_{1},\dotsc,\pure_{\nPlayers})$.

Letting $\points_{\play} = \simplex(\pures_{\play})$, the mixed extension of $\fingame$ is defined as the continuous game $\mixgame = \gamefull$.
For posterity, we note here that the ``payoff gradient'' of each player $\play\in\players$ is simply their mixed payoff vector, \ie $\vecfield_{\play}(\point) = \nabla_{\point_{\play}} \pay_{\play}(\point) = (\pay_{\play}(\pure_{\play};\point_{-\play}))_{\pure_{\play}\in\pures_{\play}}$.
%
\endenv
\end{example}

\subsection{Solution concepts}
\label{sec:solutions}
\acused{NE}
\acused{LNE}
\acused{FOS}
\acused{VS}
\acused{NS}
\acused{GVS}
\acused{GNS}
\acused{SVI}
\acused{MVI}

The standard solution concept in game theory is that of a \acli{NE}, \ie an action profile that is resilient to unilateral deviations.
Formally, $\eq\in\points$ is a \acl{NE} of a game $\game \equiv \gamefull$ if
\begin{equation}
\label{eq:NE}
\tag{NE}
\pay_{\play}(\eq)
	\geq \pay_{\play}(\point_{\play};\eq_{-\play})
	\quad
	\text{for all $\point_{\play}\in\points_{\play}$, $\play\in\players$}.
\end{equation}
\aclp{NE} always exist if $\points$ is compact and each $\pay_{\play}$ is individually concave in $\point_{\play}$ \citep{Deb52}.
Otherwise, equilibria may not exist, in which case the following relaxations become relevant:
\begin{enumerate}
\item
\emph{Local \aclp{NE}}, \ie profiles $\eq\in\points$ for which \eqref{eq:NE} holds locally:
\begin{equation}
\label{eq:LNE}
\tag{LNE}
\pay_{\play}(\eq)
	\geq \pay_{\play}(\point_{\play};\eq_{-\play})
	\quad
	\text{for all $\point$ in a neighborhood $\nhd$ of $\eq$ in $\points$}.
\end{equation}

\item
\emph{Critical points}, \ie profiles $\eq\in\points$ that satisfy the \acl{FOS} condition:
\begin{equation}
\label{eq:FOS}
\tag{FOS}
\txs
\left.\frac{d}{dt}\right|_{t=0^{+}} \pay_{\play}(\eq_{\play} + t(\point_{\play} - \eq_{\play});\eq_{-\play})
	\leq 0
	\quad
	\text{for all $\point_{\play}\in\points_{\play}$, $\play\in\players$}.
\end{equation}
\end{enumerate}
Equivalently, \eqref{eq:FOS} can be reformulated as a \acl{SVI} of the form
\begin{equation}
\label{eq:SVI}
\tag{SVI}
\braket{\payv(\eq)}{\point - \eq}
	\leq 0
	\quad 
	\text{for all $\point \in \points$}.
\end{equation}
The solutions of \eqref{eq:SVI} are precisely the fixed points of the ``linearized'' best-response correspondence
\(
\point
	\mapsto \argmax\nolimits_{\pointalt\in\points} \braket{\vecfield(\point)}{\pointalt}
\)
so, by standard fixed point arguments, the set of critical points of $\game$ is always nonempty if $\points$ is compact.

Dually to the above, the \acli{MVI} associated to $\game$ is
\begin{equation}
\label{eq:MVI}
\tag{MVI}
\braket{\payv(\point)}{\point - \eq}
	\leq 0
	\quad 
	\text{for all $\point \in \points$}.
\end{equation}
It is straightforward to verify that the solutions of \eqref{eq:MVI} comprise a convex set of \aclp{NE} of $\game$, so \eqref{eq:MVI} can be seen as an equilibrium refinement criterion for $\game$.
Taking this a step further, a state $\eq\in\points$ is said to be \acli{VS} if
\begin{align}
\label{eq:VS}
\tag{VS}
\braket{\vecfield(\point)}{\point - \eq}
	&< 0
	\quad
	\text{for all $\point\neq\eq$ in a neighborhood $\nhd$ of $\eq$ in $\points$}
\shortintertext{and $\eq$ is called \acli{NS} if the strict inequality ``$<$'' in \eqref{eq:VS} is relaxed to ``$\leq$'', \ie if}
\label{eq:NS}
\tag{NS}
\braket{\vecfield(\point)}{\point - \eq}
	&\leq 0
	\quad
	\text{for all $\point$ in a neighborhood $\nhd$ of $\eq$ in $\points$}.
\end{align}
Finally, we say that $\eq$ is \acli{GVS} [resp.~\acli{GNS}] if \eqref{eq:VS} [resp.~\eqref{eq:NS}] holds
with $\nhd=\points$ (\ie for all $\point\in\points$).

In general, the solution concepts discussed above are related as follows:
\begin{equation}
\label{eq:taxonomy}
\begin{aligned}
\begin{tikzcd}[column sep=small,row sep=small]
\hyperref[eq:VS]{\textup{(GVS)}}
	\arrow[draw=none]{r}[sloped,auto=false]{\implies}
	\arrow[draw=none]{d}[sloped,auto=false]{\implies}
	&\hyperref[eq:NS]{\textup{(GNS)}} \equiv \eqref{eq:MVI}
	\arrow[draw=none]{r}[sloped,auto=false]{\implies}
	\arrow[draw=none]{d}[sloped,auto=false]{\implies}
	&\eqref{eq:NE}
	\arrow[draw=none]{d}[sloped,auto=false]{\implies}
	&
	\\
\eqref{eq:VS}
	\arrow[draw=none]{r}[sloped,auto=false]{\implies}
	&\eqref{eq:NS}
	\arrow[draw=none]{r}[sloped,auto=false]{\implies}
	&\eqref{eq:LNE}
	\arrow[draw=none]{r}[sloped,auto=false]{\implies}
	&\eqref{eq:FOS} \equiv \eqref{eq:SVI}
\end{tikzcd}
\end{aligned}
\end{equation}
Without further assumptions, the implications in \eqref{eq:taxonomy} are all one-way;
in the next section, we discuss a number of cases where some (or all) of these implications become equivalences.

\begin{remark}
In the optimization literature, the direction of \eqref{eq:SVI}\,/\,\eqref{eq:MVI} is reversed because optimization problems are usually stated in terms of cost minimization.
The maximization viewpoint is more common in games, so we will maintain the ``$\leq$'' direction throughout.
\endenv
\end{remark}
\smallskip

\begin{remark}
The notion of variational stability was introduced in \cite{MZ19} and echoes the seminal concept of \emph{evolutionary stability} as introduced by \citet{MP73} in the context of population games.
Informally, ``variational stability'' is to games with a finite number of players and a continuum of actions what ``evolutionary stability'' is to games with a continuum of players and a finite number of actions;
for an in-depth discussion, \cf \cite{MZ19}.
\endenv
\end{remark}

\subsection{Special cases of interest}

We close this section with a discussion of some special cases and examples of the above definitions that will play a major role in the sequel.

\itempar{Monotone games}
A game is \emph{monotone} if it satisfies the monotonicity condition
\begin{equation}
\label{eq:monotone}
\tag{Mon}
\braket{\payv(\pointalt) - \payv(\point)}{\pointalt - \point}
	\leq 0
	\quad
	\text{for all $\point,\pointalt\in\points$}.
\end{equation}
The strict version of this requirement (\ie that equality holds if and only if $\point = \pointalt$) is sometimes referred to as \acdef{DSC}, a terminology due to \citet{Ros65}.
In monotone games, the solutions of \eqref{eq:MVI} and \eqref{eq:SVI} coincide, leading to the string of equivalences
\begin{equation}
\eqref{eq:MVI}
	\iff \eqref{eq:NE}
	\iff \eqref{eq:LNE}
	\iff \eqref{eq:FOS}
	\equiv \eqref{eq:SVI}
\end{equation}
By comparison, if a game is strictly monotone, every implication in \eqref{eq:taxonomy} becomes an equivalence, so the game admits a unique, \acl{GVS} \acl{NE}.
\Cref{ex:Cournot,ex:minmax} are both strictly monotone (assuming $\minmax$ is strictly convex-concave in \cref{ex:minmax});
other examples include
socially concave games \citep{EMN09},
Cournot oligopolies \citep{MS96},
Kelly auctions \citep{KMT98,Tul80},
congestion control \citep{EMN09},
and many other classes of problems.

\itempar{Potential games}
First formalized by \citet{MS96}, potential games admit a \emph{potential function} $\pot\from\points\to\R$ such that
\begin{equation}
\label{eq:pot}
\tag{Pot}
\pay_{\play}(\point_{\play};\point_{-\play}) - \pay_{\play}(\pointalt_{\play};\point_{-\play})
	= \pot(\point_{\play};\point_{-\play}) - \pot(\pointalt_{\play};\point_{-\play})
	\quad
	\text{for all $\point,\pointalt \in \points$ and all $\play\in\players$}.
\end{equation}
If $\game$ is a potential game, we have $\vecfield(\point) = \nabla\pot(\point)$ so
\begin{enumerate*}
[\itshape a\upshape)]
\item
any local maximum of $\pot$ is a \acl{LNE} of $\game$;
and
\item
any \emph{strict} local maximum of $\pot$ is \acl{VS}.
\end{enumerate*}
The Cournot oligopoly of \cref{ex:Cournot} is a textbook example of a potential game;
other examples include
finite congestion games \citep{Ros73},
power allocation in wireless networks \citep{SFPP10},
etc.

\itempar{\Acl{SOS}}
Our next example concerns critical points that satisfy a condition similar to \aclp{SOSC} in optimization, namely
\begin{equation}
\label{eq:SOS}
\tag{SOS}
\tvec^{\top} \Jac_{\vecfield}(\eq) \tvec
	< 0
	\quad
	\text{for all nonzero tangent vectors $\tvec$ to $\points$ at $\eq$}
\end{equation}
where $\Jac_{\vecfield}(\eq)$ denotes the Jacobian of $\vecfield$ at $\eq$.
In the context of saddle-point problems and continuous games, this condition has been studied extensively in the machine learning and control literatures, \cf \cite{RBS16,HIMM19,MHKC20,GLMV22,Ros65} and references therein.
Importantly, as we note below, \eqref{eq:SOS} is a special case of \eqref{eq:VS}.

\begin{proposition}[{\citeauthor{HIMM19}, \citeyear{HIMM19}, Lemma~A.4}]
\label{prop:regular}
Let $\eq$ be a critical point of $\game$ satisfying \eqref{eq:SOS}.
Then $\eq$ is variationally stable.
\end{proposition}


\itempar{Finite games}
As a last example, let $\game = \simplex(\fingame)$ be the mixed extension of a finite game $\fingame \equiv \fingamefull$.
Since each player's payoff function $\pay_{\play}(\point_{\play};\point_{-\play})$ is linear in $\point_{\play}$, we readily get
\begin{equation}
\eqref{eq:NE}
	\iff \eqref{eq:LNE}
	\iff \eqref{eq:FOS}
	\equiv \eqref{eq:SVI}
\end{equation}
In addition, we have the following characterization of stable states in finite games:

\begin{proposition}[{\citeauthor{MZ19}, \citeyear{MZ19}, Prop.~5.2}]
\label{prop:strictVS}
A mixed strategy profile $\eq$ is variationally stable if and only if it is a strict \acl{NE} of $\fingame$, \ie if and only if \eqref{eq:NE} holds as a strict inequality for all $\point\neq\eq$.
\end{proposition}


We will use all this freely in the sequel.

%% file: Body/Framework.tex

We now proceed to detail our online learning framework, beginning with the general model in \cref{sec:method} and continuing with a range of learning algorithms that can be seen as special cases thereof in \cref{sec:algorithms}.
The reader interested only in the general theory can skip \cref{sec:algorithms}.
\setcounter{remark}{0}

\subsection{\Acl{method} processes}
\label{sec:method}

Our basic learning framework will hinge on the \acli{method} template
\begin{equation}
\label{eq:method}
\tag{RRM}
\begin{aligned}
\next[\dstate]
	= \curr[\dstate] + \curr[\step] \curr[\signal]
	\qquad
\next
	= \mirror(\next[\dstate]),
\end{aligned}
\end{equation}
where:
\begin{enumerate}
\item
$\curr = (\state_{\play,\run})_{\play\in\players} \in \points$ denotes the players' action profile at each stage $\run=\running$
\item
$\curr[\signal] = (\signal_{\play,\run})_{\play\in\players} \in \dspace$ is a sequence of individual ``gradient-like'' signals.
\item
$\curr[\dstate] = (\dstate_{\play,\run})_{\play\in\players} \in \dspace$ is an auxiliary state variable aggregating individual gradient steps.
\item
$\curr[\step] > 0$ is a step-size sequence,
for which we will assume throughout that $\sum_{\run} \curr[\step] = \infty$
(typically the method is run with $\curr[\step] \propto 1/\run^{\stepexp}$ for some $\stepexp\geq0$).
\item
$\mirror\from \dspace \to \points$ is a ``generalized projection'' map that mirrors gradient steps in $\dspace$ to action updates in $\points$;
we will refer to $\mirror$ throughout as the players' \emph{mirror map}.
\end{enumerate}
As far as terminology is concerned,
the term ``\acli{RM}'' refers to the seminal stochastic approximation method of \citet{RM51},
while the adjective ``regularized'' alludes to the \acdef{FTRL} family of algorithms of \citet{SSS06} \textendash\ which, in turn, is intimately related to the \acdef{MD} framework of \citet{NY83}.
To streamline our presentation, we detail each of these elements below and defer a list of examples to \cref{sec:algorithms}.


\itempar{The sequence of gradient signals}

To keep track of the sequence of events in \eqref{eq:method}, we will view $\curr$ as a stochastic process on some complete probability space $\probspace$, and we will write $\curr[\filter] \defeq \filter(\init,\dotsc,\curr) \subseteq \filter$ for the history of play up to stage $\run$ (inclusive).
Since we tacitly assume that $\curr[\signal]$ is generated \emph{after} each player has selected an action at round $\run$ but \emph{before} the $(\run+1)$-th update has been triggered, we also posit that $\curr[\signal]$ is $\next[\filter]$-measurable but not necessarily $\curr[\filter]$-measurable.
In this way, we may decompose $\curr[\signal]$ as
\begin{equation}
\label{eq:signal}
\curr[\signal]
	= \vecfield(\curr) + \curr[\noise] + \curr[\bias]
\end{equation}
where
\begin{equation}
\label{eq:bias-noise}
\curr[\noise]
	= \curr[\signal] - \exof{\curr[\signal] \given \curr[\filter]}
	\quad
	\text{and}
	\quad
\curr[\bias]
	= \exof{\curr[\signal] \given \curr[\filter]} - \vecfield(\curr).
\end{equation}

Since $\exof{\curr[\noise] \given \curr[\filter]} = 0$ by construction, $\curr[\noise]$ can be intepreted as a random, zero-mean error relative to $\vecfield(\curr)$;
by contrast, $\curr[\bias]$ is $\curr[\filter]$-measurable, so it captures any systematic \textendash\ and possibly \emph{non-random} \textendash\ offset of $\curr[\signal]$ relative to $\vecfield(\curr)$.
We will quantify all this by assuming that $\curr[\bias]$, $\curr[\noise]$ and $\curr[\signal]$ are bounded for some $\qexp\geq2$ as
\begin{equation}
\label{eq:errorbounds}
\exof{\dnorm{\curr[\bias]} \given \curr[\filter]}
	\leq \curr[\bbound]
	\qquad
\exof{\dnorm{\curr[\noise]}^{\qexp} \given \curr[\filter]}
	\leq \curr[\sdev]^{\qexp}
	\qquad
	\text{and}
	\qquad
\exof{\dnorm{\curr[\signal]}^{\qexp} \given \curr[\filter]}
	\leq \curr[\sbound]^{\qexp}
\end{equation}
where the sequences $\curr[\bbound]$, $\curr[\sdev]$ and $\curr[\sbound]$, $\run=\running$, are to be construed as deterministic upper bounds on the bias, fluctuations, and magnitude of $\curr[\signal]$ respectively (with the case $\qexp=\infty$ taken to mean that the various quantites are bounded \acs{wp1}).
Accordingly, depending on these bounds, a gradient signal with $\curr[\bbound]=0$ will be called \emph{unbiased}, and an unbiased signal with $\curr[\sdev] = 0$  will be called \emph{perfect}.

\begin{remark}
We should stress here that $\curr[\signal]$ \emph{should not} be interpreted narrowly as the output of a black-box oracle for $\vecfield(\curr)$, but as a ``model-agnostic'' surrogate thereof.
In particular, the noise term $\curr[\noise]$ can model raw observational noise, but also inner randomizations of the algorithm;
analogously, the bias term $\curr[\bias]$ is intended to capture situations where $\curr[\signal]$ results from actions other than $\curr$, the inclusion of corrective terms in a learning algorithm, etc.
These modeling aspects are crucial to include in our analysis optimistic and extra-gradient methods;
we explain this issue in detail in \cref{sec:algorithms}.
\endenv
\end{remark}

\begin{remark}
By \cref{asm:game} and the inequality $\parens*{\sum_{i=1}^{m} a_{i}}^{\qexp} \leq m^{\qexp-1} \sum_{i=1}^{m} a_{i}^{\qexp}$, the decomposition \eqref{eq:signal} of $\curr[\signal]$ shows that we can always pick $\curr[\sbound]^{\qexp} = 3^{\qexp-1} \parens*{\vbound^{\qexp} + \curr[\bbound]^{\qexp} + \curr[\sdev]^{\qexp}}$ in \eqref{eq:errorbounds}.
This makes the last part of \eqref{eq:errorbounds} redundant, but we will maintain the explicit bound $\curr[\sbound]$ for $\curr[\signal]$ to simplify the presentation.
\endenv
\end{remark}

\itempar{The players' mirror map}

The second defining element of \eqref{eq:method} is the ``mirror map'' $\mirror_{\play}\from\dspace_{\play}\to\points_{\play}$ of each player
\textendash\ or, in aggregate form, the product map $\mirror = (\mirror_{\play})_{\play\in\players} \from \dspace \to \points$.
This is defined by means of a ``\emph{regularizer}'' on $\points$ as follows:%
\footnote{The authors thank S.~Sorin for proposing this definition.}

\begin{definition}
\label{def:hreg}
We say that $\hreg_{\play}\from\pspace_{\play}\to\R\cup\{\infty\}$ is a \emph{regularizer} on $\points_{\play}$ if:
\begin{enumerate}
\item
$\hreg_{\play}$ is \emph{supported} on $\points_{\play}$, \ie $\dom\hreg_{\play} = \setdef{\point_{\play}\in\pspace_{\play}}{\hreg_{\play}(\point_{\play}) < \infty} = \points_{\play}$.

\item
$\hreg_{\play}$ is continuous and \emph{strongly convex} on $\points_{\play}$, \ie there exists a constant $\hstr_{\play}>0$ such that
\begin{equation}
\label{eq:hstr}
\hreg_{\play}(\coef \point_{\play} + (1-\coef) \pointalt_{\play})
	\leq \coef\hreg_{\play}(\point_{\play}) + (1-\coef) \hreg_{\play}(\pointalt_{\play})
		- \tfrac{1}{2} \hstr_{\play} \coef(1-\coef) \norm{\pointalt_{\play} - \point_{\play}}^{2}
\end{equation}
for all $\point_{\play},\pointalt_{\play} \in \points_{\play}$ and all $\coef\in[0,1]$.
\end{enumerate}
The \emph{mirror map} associated to $\hreg_{\play}$ is defined
for all $\dpoint_{\play}\in\dspace_{\play}$ as
\begin{equation}
\label{eq:mirror}
\mirror_{\play}(\dpoint_{\play})
	= \argmax\nolimits_{\point_{\play}\in\points_{\play}}
		\{ \braket{\dpoint_{\play}}{\point_{\play}} - \hreg_{\play}(\point_{\play}) \}
\end{equation}
and the image $\proxdomi = \im\mirror_{\play}$ of $\mirror_{\play}
$ is called the \emph{prox-domain} of $\hreg_{\play}$.
Finally, we will also say that $\hreg_{\play}$ is \emph{steep} when $\dom\subd\hreg_{\play} = \relint\points_{\play}$.
\end{definition}

For concision, we will write $\hreg(\point) = \sum_{\play}\hreg_{\play}(\point_{\play})$ for the players' aggregate regularizer and $\mirror = (\mirror_{\play})_{\play\in\players}$ for the induced mirror map.
We provide three examples of this construction below:

\begin{example}[Euclidean projection]
\label{ex:mirror-Eucl}
Consider the quadratic regularizer $\hreg(\point) = \twonorm{\point}^{2}/2$, $\point\in\points$.
Then the induced mirror map is the standard Euclidean projector
\begin{equation}
\label{eq:Eucl}
\mirror(\dpoint)
	= \Eucl_{\points}(\dpoint)
	\equiv \argmin_{\point\in\points} \twonorm{\dpoint-\point}.
\end{equation}
As a special case, in unconstrained settings (\ie when $\points = \pspace$), we have $\mirror(\dpoint) = \dpoint$.
\end{example}

\begin{example}[Entropic regularization on the simplex]
\label{ex:mirror-logit}
Let $\pures_{\play}$, $\play=1,\dotsc,\nPlayers$, be an ensemble of pure strategies,
set $\points_{\play} = \simplex(\pures_{\play})$,
and
let $\hreg_{\play}(\point_{\play}) = \sum_{\pure_{\play}\in\pures_{\play}} \point_{\play\pure_{\play}}\log\point_{\play\pure_{\play}}$ be the (negative) Gibbs\textendash Shannon entropy on $\points_{\play}$.
By standard arguments, the resulting mirror map of each player $\play\in\players$ is the logit choice map
\begin{equation}
\label{eq:logit}
\mirror_{\play}(\dpoint_{\play})
	= \logit_{\play}(\dpoint_{\play})
	\equiv \frac
		{(\exp(\dpoint_{\play\pure_{\play}}))_{\pure_{\play}\in\pures_{\play}}}
		{\sum_{\pure_{\play}\in\pures_{\play}} \exp(\dpoint_{\play\pure_{\play}})}
\end{equation}
This choice map plays a central role in finite games;
we will revisit it several times in \cref{sec:algorithms}.
\end{example}

\begin{example}[Regularization on the orthant]
\label{ex:mirror-orthant}
Let $\points_{\play} = [0,\infty)$ and set $\hreg_{\play}(\point_{\play}) = \point_{\play}\log\point_{\play} - \point_{\play}$ for all $\point_{\play}\in\points_{\play}$, $\play\in\players$.
By a straightforward calculation, the induced mirror map is $\mirror_{\play}(\dpoint_{\play}) = \exp(\dpoint_{\play})$.
As we discuss in \cref{sec:algorithms}, this provides the relevant setup for games with half-space constraints.%
\footnote{Strictly speaking, the regularizer $\point \log\point - \point$ is not strongly convex over $\R_{+}$ but it \emph{is} strongly convex over any bounded subset of $\R_{+}$ \textendash\ and it can be made strongly convex over all of $\R_{+}$ by adding a small quadratic penalty of the form $\eps \point^{2} / 2$.
This issue does not change the essence of our results, so we sidestep the details.}
\end{example}

\subsection{Specific algorithms}
\label{sec:algorithms}

We now proceed to describe a representative range of learning algorithms that can be incorporated as specific instances of the general framework \eqref{eq:method}.
Depending on the information available to the players, we classify the algorithms under study as \emph{oracle-based} or \emph{payoff-based}.

\itempar{Oracle-based methods}

In the first batch of methods under consideration, the players are assumed to have access to a \acdef{SFO}, that is, a ``black-box'' feedback mechanism that returns an estimate of their individual payoff gradients at the chosen action profile.
Formally, when queried at $\point\in\points$, an \ac{SFO} outputs a random vector of the form
\begin{align}
\label{eq:SFO}
\tag{SFO}
\orcl(\point;\seed)
	&= \vecfield(\point)
	+ \err(\point;\seed),
\end{align}
where
$\seed$ is a random variable taking values in some measurable space $\seeds$
and
$\err(\point;\seed)$ is an umbrella error term capturing all sources of uncertainty in the model.
In practice, \eqref{eq:SFO} is queried repeatedly at a sequence of action profiles $\curr\in\points$, $\run=\running$, possibly with a different random seed $\curr[\seed]$ each time.%
\footnote{In some cases, the index set may be enlarged to include all positive half-integers ($\run = \halfrunning$).}
For concreteness, we will assume that the noise in \eqref{eq:SFO} is zero-mean and bounded in $L^{\qexp}$ for some $\qexp\geq2$, \ie
\begin{equation}
\label{eq:oracle}
\exof{\err(\point;\curr[\seed]) \given \curr[\filter]}
	= 0
	\quad
	\text{and}
	\quad
\exof{\dnorm{\err(\point;\curr[\seed])}^{\qexp} \given \curr[\filter]}
	\leq \sdev^{\qexp}
\end{equation}
for some $\sdev\geq0$ and all $\point\in\points$ (with $\qexp=\infty$ taken to mean that $\err(\point;\seed)$ is bounded \acs{wp1}).

We are now in a position to introduce the array of \emph{oracle-based} methods under study;
to lighten notation, we present some of these policies in an unconstrained setting.

\begin{algo}[\Acl{SGA}]
\label{alg:SGA}

Perhaps the most basic iterative policy for multi-agent online learning is the standard (individual) gradient ascent method
\begin{align}
\label{eq:SGA}
\tag{SGA}
\next
	= \curr + \curr[\step] \orcl(\curr;\curr[\seed])
\end{align}
with $\curr[\seed]$ drawn \acs{iid} from $\seeds$.
From a loss minimization viewpoint, \eqref{eq:SGA} is a multi-agent analogue of the standard \acl{SGD} algorithm;
in min-max games, \eqref{eq:SGA} is sometimes referred to as the Arrow\textendash Hurwicz method \citep{AHU58}.
Clearly, \eqref{eq:SGA} is immediately recovered from \eqref{eq:method} if the latter is run with the sequence of gradient signals $\curr[\signal] \gets \orcl(\curr;\curr[\seed])$ and the trivial mirror map $\mirror(\dpoint) = \dpoint$.
\endenv
\end{algo}

\begin{algo}[\Acl{EG}]
\label{alg:EG}
Going a step further from \eqref{eq:SGA}, the (stochastic) \acdef{EG} algorithm of \citet{Kor76} is based on the following principle:
starting at some ``base'' state $\curr$, the players first take a gradient step to an interim, ``leading'' state $\lead$;
subsequently, to anticipate their payoff landscape, they update the base state $\curr$ with gradient information from $\lead$ instead of $\curr$,
and the process repeats.
Formally, this leads to the policy
\begin{equation}
\label{eq:EG}
\tag{\acs{EG}}
\begin{aligned}
\lead
	&= \curr + \curr[\step] \orcl(\curr;\curr[\seed])
	\\
\next
	&= \curr + \curr[\step] \orcl(\lead;\lead[\seed])
\end{aligned}
\end{equation}
with $\curr[\seed],\lead[\seed]$ drawn \acs{iid} from $\seeds$.
Accordingly, \eqref{eq:EG} is readily recovered from \eqref{eq:method} by taking $\curr[\signal] \gets \orcl(\lead;\lead[\seed])$.
\endenv
\end{algo}

\begin{algo}[\Acl{OG}]
\label{alg:OG}
A computational drawback of \eqref{eq:EG} is that it requires two oracle queries per update
\textendash\ and hence, more overhead per iteration.
One way to overcome this hurdle is to reuse past gradient information in the hope that it provides a good enough approximation of the present;
this leads to the \acli{OG} policy
\begin{equation}
\label{eq:OG}
\tag{\acs{OG}}
\begin{aligned}
\lead
	&= \curr + \curr[\step] \orcl(\beforelead;\prev[\seed])
	\\
\next
	&= \curr + \curr[\step] \orcl(\lead;\curr[\seed])
\end{aligned}
\end{equation}
Similarly to \eqref{eq:EG}, \eqref{eq:OG} is recovered from \eqref{eq:method} by setting $\curr[\signal] \gets \orcl(\lead;\curr[\seed])$.
This ``gradient reuse'' idea goes back at least to \citet{Pop80}, and it has resurfaced several times in the literature since then, \cf \cite{RS13-NIPS,DISZ18,GBVV+19,HIMM19} and references therein.
To simplify our presentation, we will assume in the sequel that \eqref{eq:OG} is run with an \ac{SFO} satisfying \eqref{eq:oracle} with $\qexp = \infty$.
\endenv
\end{algo}

The next method concerns learning in mixed extensions of finite games.

\begin{algo}[\Acl{EW}]
\label{alg:EW}
Let $\game = \mixgame$ be the mixed extension of a finite game $\fingamefull$ as per \cref{ex:mixed}.
In this setting, the players' learning process typically unfolds as follows:
at each stage $\run=\running$, every player selects a mixed strategy $\state_{\play,\run} \in \simplex(\pures_{\play})$ and draws a pure strategy $\pure_{\play,\run} \in \pures_{\play}$ according to $\state_{\play,\run}$.
Then, depending on the amount of information available to the players, we have the following oracle models:
\begin{subequations}
\label{eq:SFO-finite}
\begin{enumerate}
\item
\emph{Full information feedback:}
in this case, players observe their \emph{mixed} payoff vectors, \ie
\begin{align}
\label{eq:SFO-finite-full}
\orcl_{\play}(\curr;\curr[\pure])
	&= \payv_{\play}(\state_{\play,\run};\state_{-\play,\run}).
\shortintertext{%
\item
\emph{Realization-based feedback:}
here, players instead observe their \emph{pure} payoff vectors, \ie
} 
\label{eq:SFO-finite-realized}
\orcl_{\play}(\curr;\curr[\pure])
	&= \payv_{\play}(\pure_{\play,\run};\pure_{-\play,\run}).
\end{align}
\end{enumerate}
\end{subequations}
Both models can be seen as \acp{SFO} with seed $\curr[\pure]$, \ie the (pure) action profile chosen by the players at stage $\run$;
the oracle \eqref{eq:SFO-finite-full} is deterministic, while the oracle \eqref{eq:SFO-finite-realized} is stochastic and satisfies \eqref{eq:oracle} with $\qexp=\infty$.

In this context, one of the most widely used learning methods is the so-called \acli{EW} algorithm \textendash\ or \acs{EW} \textendash\ which unfolds iteratively as
\begin{equation}
\label{eq:EW}
\tag{\acs{EW}}
\begin{aligned}
\dstate_{\play,\run+1}
	&= \dstate_{\play,\run} + \curr[\step] \orcl_{\play}(\curr;\curr[\pure])
	\\
\state_{\play,\run+1}
	&= \logit_{\play}(\dstate_{\play,\run+1})
\end{aligned}
\end{equation}
with $\logit_{\play}$ denoting the logit choice map of \eqref{eq:logit} and $\orcl_{\play}$ given by \eqref{eq:SFO-finite-full} or \eqref{eq:SFO-finite-realized} depending on the information available to the players.
In both cases, \eqref{eq:EW} is recovered immediately from \eqref{eq:method} by letting $\curr[\signal] \gets \orcl(\curr;\curr[\pure])$ and $\mirror_{\play} = \logit_{\play}$.
For an overview of the method's history and its applications, see \cite{CBL06,LS20} and references therein.
\endenv
\end{algo}

\itempar{Payoff-based methods}

Moving forward, it is important to recall that \crefrange{alg:SGA}{alg:EW} all assume that players have access to a black-box oracle mechanism, but do not specify how this could be achieved in practice.
Albeit commonplace, this assumption is not realistic in many applications where players may only be able to observe their realized payoffs and have no information about the strategies of other players or actions they did not play.
To bridge this disconnect, we describe below a range of \emph{payoff-based} policies where players estimate their individual payoff gradients indirectly, from their realized, ``in-game'' payoffs.

\begin{algo}[\Acl{SPSA}]
\label{alg:SPSA}
\acused{SPSA}
A straightforward way of reconstructing gradients from zeroth-order feedback is via the \acli{SPSA} framework of \citet{Spa92}.
In the unconstrained case ($\points=\pspace$), the relevant update step is:
\begin{equation}
\label{eq:SPSA}
\tag{SPSA}
\begin{aligned}
\query_{\play,\run}
	&= \state_{\play,\run} + \curr[\mix] \unitvar_{\play,\run},
	\\
\state_{\play,\run+1}
	&= \state_{\play,\run}
		+ \curr[\step] (\pay_{\play}(\curr[\query]) / \curr[\mix]) \,  \unitvar_{\play,\run}.
\end{aligned}
\end{equation}
In \eqref{eq:SPSA}, each player's ``query state'' $\query_{\play,\run}$, $\play=1,\dotsc,\nPlayers$, is a perturbation of the ``base state'' $\state_{\play,\run}$ by a step of magnitude $\curr[\mix]>0$ along a random direction $\unitvar_{\play,\run}$ drawn from the ensemble of signed basis vectors
$\bvecs_{\play} \defeq \{(\pm\bvec_{1},\dotsc,\pm\bvec_{\vdim_{\play}})\}$.
In this manner, \eqref{eq:SPSA} can be seen as a special case of \eqref{eq:method} with $\signal_{\play,\run} \gets (\pay_{\play}(\curr[\query]) / \curr[\mix]) \,  \unitvar_{\play,\run}$ for all $\play\in\players$.%
\footnote{This formulation of \eqref{eq:SPSA} is tailored to unconstrained problems.
In this case, to ensure that the resulting gradient estimator remains bounded, it is customary to include an indicator of the form $\oneof{\norm{\curr[\query]} \leq \curr[R]}$ for some suitably chosen sequence $\curr[R]\to\infty$ \cite{Spa92}.
This would lead to the same analysis but at the cost of heavier notation so, instead, we will assume that the players' payoff functions are bounded when discussing \eqref{eq:SPSA}.
For a detailed discussion of how to adapt \eqref{eq:SPSA} in the presence of constraints, we refer the reader to \citet{BLM18} who show that the relevant entries of \cref{tab:algorithms} apply verbatim when $\points$ is compact.}
\endenv
\end{algo}


\begin{algo}[\Acl{DGA}]
\label{alg:DGA}
\acused{DGA}
An alternative approach to \eqref{eq:SPSA} is the two-point, ``explore-then-update'' approach of \citet{BBF20} who focused on games with $\points_{\play} = [0,\infty)$ for all $\play\in\players$ and introduced the \acli{DGA} policy
\begin{equation}
\label{eq:DGA}
\tag{DGA}
\begin{aligned}
\state_{\play,\run+1/2}
	&= \state_{\play,\run} + (1/\run) \unitvar_{\play,\run}
	\\
\state_{\play,\run+1}
	&= \state_{\play,\run} \bracks{1 + (\pay_{\play}(\lead) - \pay_{\play}(\curr)) \unitvar_{\play,\run}}
\end{aligned}
\end{equation}
In the above, the ``exploration direction'' $\unitvar_{\play,\run}$ is sampled uniformly at random from $\{\pm 1\}$ at each $\run$.
In other words, \eqref{eq:DGA} is a two-stage process where players first ``explore'' their individual payoff functions at a nearby state, and then use this information to estimate their individual payoff gradients and update their base state.

To include \eqref{eq:DGA} in the framework of \eqref{eq:method},
take $\mirror_{\play}(\dpoint_{\play}) = \exp(\dpoint_{\play})$ as per \cref{ex:mirror-orthant}.
Then, letting $\curr[\dstate] = \log\curr$, we get
\begin{equation}
\next[\dstate]
	= \curr[\dstate] + \log\parens{1 + (\pay_{\play}(\lead) - \pay_{\play}(\curr)) \unitvar_{\play,\run}}.
\end{equation}
We may therefore view \eqref{eq:DGA} as an instance of \eqref{eq:method} with $\curr[\step] = 1/\run$ and gradient signals given by $\signal_{\play,\run} \gets \run\cdot\log\parens{1 + (\pay_{\play}(\lead) - \pay_{\play}(\curr)) \unitvar_{\play,\run}}$.
\endenv
\end{algo}

\begin{algo}[The \acs{EXP3} algorithm]
\label{alg:EXP3}
In our final example, we return to finite games, and we focus on the ``bandit'' case where players only observe the payoffs of the pure strategies that they played.
In this setting, it is common to employ the \acli{IWE}
\begin{equation}
\label{eq:IWE}
\tag{IWE}
\orcl_{\play\pure_{\play}}(\curr[\query];\curr[\purequery])
	= \frac
		{\oneof{\purequery_{\play,\run} = \pure_{\play}}}
		{\query_{\play\pure_{\play},\run}}
			\pay_{\play}(\purequery_{\play,\run};\purequery_{-\play,\run})
	\qquad
	\text{for all $\pure_{\play}\in\pures_{\play}$, $\play\in\players$},
\end{equation}
where each player $\play\in\players$ draws an action $\purequery_{\play,\run}$ from $\pures_{\play}$ according to a mixed strategy $\query_{\play,\run}  \in \simplex(\pures_{\play})$.
Then, plugging \eqref{eq:IWE} into \eqref{eq:EW}, we obtain the method known as \acdef{EXP3}, viz.
\begin{align}
\label{eq:EXP3}
\tag{EXP3}
\begin{split}
\dstate_{\play,\run+1}
	&=\dstate_{\play,\run} + \curr[\step] \orcl_{\play}(\curr[\query];\curr[\purequery]),
	\\
\state_{\play,\run+1}
	&= \logit_{\play}(\dstate_{\play,\run+1}),
\end{split}
\shortintertext{where the sampling strategy of the $\play$-th player at stage $\run$ is given by}
\label{eq:explore}
\query_{\play,\run}
	&= (1-\curr[\mix]) \state_{\play,\run} + \curr[\mix] \unif(\pures_{\play}).
\end{align}
In the above, $\curr[\mix] \geq 0$ is an ``explicit exploration'' parameter that determines the mixing between $\state_{\play,\run}$ and the uniform distribution $\unif(\pures_{\play})$ on $\pures_{\play}$.
Accordingly, \eqref{eq:EXP3} can be seen as an instance of \eqref{eq:method} with $\mirror_{\play} = \logit_{\play}$ and $\curr[\signal]$ given by \eqref{eq:IWE} with pure strategies $\curr[\purequery]$ drawn according to $\curr[\query]$.
\endenv
\end{algo}

\itempar{Runtime parameters}

The above justifies the characterization of \eqref{eq:method} as a ``parent scheme'' for \crefrange{alg:SGA}{alg:EXP3}.
In particular, thanks to the explicit expressions for $\curr[\signal]$ derived in each case, we can likewise estimate the error bounds $\curr[\bbound]$, $\curr[\sdev]$ and $\curr[\sbound]$ of each method.


\begin{table}[tbp]
\footnotesize
\centering
\renewcommand{\arraystretch}{1.25}
\setlength{\tabcolsep}{.75em}
\input{Tables/Algorithms}
\caption{The algorithms of \cref{sec:algorithms} as instances of \eqref{eq:method}.
Where applicable, the methods' step-size and sampling parameters are taken to be of the form $\curr[\step] \propto 1/\run^{\stepexp}$ and $\curr[\mix] \propto 1/\run^{\mixexp}$ for some $\stepexp\in[0,1]$ and $\mixexp\in(0,1/2)$ respectively.
}
\label{tab:algorithms}
\vspace{-2ex}
\end{table}


\begin{restatable}{proposition}{errorbounds}
\label{prop:algorithms}
Suppose that \crefrange{alg:SGA}{alg:EXP3} are run with step-size $\curr[\step] \propto 1/\run^{\stepexp}$, $\stepexp\in[0,1]$, and, where applicable, a sampling parameter $\curr[\mix] \propto 1/\run^{\mixexp}$, $\mixexp\in(0,1/2)$. 
Then the corresponding sequence of gradient signals $\curr[\signal]$ in \eqref{eq:method} enjoys the bounds:
\begin{itemize}
\item
For \cref{alg:SGA,alg:EW}:
	\tabto{11em}
	$\curr[\bbound] = 0$, $\curr[\sdev] = \bigoh(1)$, and $\curr[\sbound] = \bigoh(1)$.
\item
For \cref{alg:EG,alg:OG}:
	\tabto{11em}
	$\curr[\bbound] = \bigoh(1/\run^{\stepexp})$, $\curr[\sdev] = \bigoh(1)$, and $\curr[\sbound] = \bigoh(1)$.
\item
For \cref{alg:SPSA,alg:EXP3}:
	\tabto{11em}
	$\curr[\bbound] = \bigoh(1/\run^{\mixexp})$, $\curr[\sdev] = \bigoh(\run^{\mixexp})$, and $\curr[\sbound] = \bigoh(\run^{\mixexp})$.
\item
For \cref{alg:DGA}:
	\tabto{11em}
	$\curr[\bbound] = \bigoh(1/\run)$, $\curr[\sdev] = \bigoh(1)$, and $\curr[\sbound] = \bigoh(1)$.
\end{itemize}
\end{restatable}

For ease of reference, we summarize the above in \cref{tab:algorithms} (for the proof of \cref{prop:algorithms}, see \cref{app:proofs}).
Of course, we can also ``mix'n'match'' different methods to include other algorithms considered in the literature:
for instance,
coupling \eqref{eq:SPSA} with a general mirror map leads to the \acl{BMD} algorithm of \citet{BLM18};
incorporating the gradient reuse step of \eqref{eq:OG} in the setup of \eqref{eq:EW} yields the \acf{OMW} method of \citet{DP19};
etc.
In the sections to come, we will exploit the expressive power of \eqref{eq:method} to provide a synthetic analysis for all these policies.

%% file: Tables/Algorithms.tex

\begin{tabular}{ccccccc}
\toprule
\textbf{Algorithm}
	&\textbf{Actions ($\points_{\play}$)}
	&\textbf{Mirror Map ($\mirror$)}
	&\textbf{Feedback}
	&\textbf{Bias ($\curr[\bbound]$)}
	&\textbf{Magnitude ($\curr[\sbound]$)}
	\\
\midrule
\eqref{eq:SGA}
	&$\R^{\vdim_{\play}}$
	&$\dpoint$
	&oracle
	&$0$
	&$\bigoh(1)$
	\\
\eqref{eq:EG}\,/\,\eqref{eq:OG}
	&$\R^{\vdim_{\play}}$
	&$\dpoint$
	&oracle
	&$\bigoh(1/\run^{\stepexp})$
	&$\bigoh(1)$
	\\
\eqref{eq:EW}
	&$\simplex(\pures_{\play})$
	&$\logit(\dpoint)$
	&oracle
	&$0$
	&$\bigoh(1)$
	\\
\eqref{eq:SPSA}
	&$\R^{\vdim_{\play}}$
	&$\dpoint$
	&payoff
	&$\bigoh(1/\run^{\mixexp})$
	&$\bigoh(\run^{\mixexp})$
	\\
\eqref{eq:DGA}
	&$[0,\infty)$
	&$\exp(\dpoint)$
	&payoff
	&$\bigoh(1/\run)$
	&$\bigoh(1)$
	\\
\eqref{eq:EXP3}
	&$\simplex(\pures_{\play})$
	&$\logit(\dpoint)$
	&payoff
	&$\bigoh(1/\run^{\mixexp})$
	&$\bigoh(\run^{\mixexp})$
	\\
\bottomrule
\end{tabular}

%% file: Body/General.tex


\subsection{Mean dynamics and stochastic approximation}
\label{sec:SA}

In this section, we derive a series of convergence results for \eqref{eq:method} by treating it as a ``noisy'' discretization of the \emph{mean dynamics}
\begin{equation}
\label{eq:MD}
\tag{MD}
\dot\dpoint
	= \vecfield(\point)
	\qquad
\point
	= \mirror(\dpoint).
\end{equation}
In this continuous-time interpretation, $\dot\dpoint$ represents the limit of the finite difference quotient $(\next[\dstate] - \curr[\dstate])/\curr[\step]$.
As such, if $\curr[\step]$ is ``sufficiently small'' and the gradient signal $\curr[\signal]$ is a ``good enough'' approximation of $\vecfield(\curr)$, it is plausible to expect that the iterates of \eqref{eq:method} and the solutions of \eqref{eq:MD} will eventually come together.

Following \cite{BH96,Ben99}, this heuristic can be made precise as follows:
First, let $\dflow\from\R\times\dspace\to\dspace$ denote the \emph{flow} associated to \eqref{eq:MD}, \ie the map which sends an initial condition $\dpoint\in\dspace$ to the point $\dflowof{\time}{\dpoint} \in \dspace$ obtained by following the orbit of \eqref{eq:MD} starting at $\dpoint$ for time $\time\in\R$.
Then, to compare the sequence of iterates $\curr[\dstate]$ generated by \eqref{eq:method} with the solution orbits of \eqref{eq:MD}, define the \emph{effective time} $\curr[\efftime] = \sum\nolimits_{\runalt=\start}^{\run} \iter[\step]$ and the associated \emph{affine interpolation} $\apt{\time}$ of $\curr[\dstate]$ as
\begin{equation}
\label{eq:interpolation}
\apt{\time}
	= \curr[\dstate]
		+ \frac{\time - \curr[\efftime]}{\next[\efftime] - \curr[\efftime]} (\next[\dstate] - \curr[\dstate])
	\quad
	\text{for all $\time \in [\curr[\efftime],\next[\efftime]]$, $\run=\running$}
\end{equation}
We then have the following notion of ``asymptotic closeness'' between \eqref{eq:method} and \eqref{eq:MD}:

\begin{definition}[\citeauthor{Ben99}, \citeyear{Ben99}]
\label{def:APT}
The sequence $\curr[\dstate]$ is an \acdef{APT} of \eqref{eq:MD} if
\acused{APT}
\begin{equation}
\label{eq:APT}
\tag{APT}
\txs
\lim_{\time\to\infty}
	\sup_{0 \leq \timealt \leq \horizon}
		\dnorm{\apt{\time + \timealt} - \dflowof{\timealt}{\apt{\time}}}
	= 0
	\quad
	\text{for all $\horizon>0$}.
\end{equation}
\end{definition}

In words, \cref{def:APT} posits that $\curr[\dstate]$ asymptotically tracks the orbits $\dtrajof{\time}$ of \eqref{eq:MD} with arbitrary precision over windows of arbitrary length.
In our setting,
the following proposition can be used as an explicit criterion guaranteeing this property:

\begin{restatable}{proposition}{APT}
\label{prop:APT}
Suppose that \eqref{eq:method} is run with step-size and gradient signal sequences such that
\begin{enumerate*}
[\itshape a\upshape)]
\item
$\curr[\step] \to 0$;
\item
$\curr[\bbound] \to 0$;
and
\item
$\sum_{\run} \curr[\step]^{1+\qexp/2} \curr[\sbound]^{\qexp} < \infty$.
\end{enumerate*}
Then, \acl{wp1}, the sequence $\curr = \mirror(\curr[\dstate])$ is an \ac{APT} of \eqref{eq:MD}.
\end{restatable}

\Cref{prop:APT} shadows a basic result of \citet[Prop.~4.1, \cf Eq.~(13) and onwards]{Ben99} so we omit its proof.
For our purposes, it is more important to note that a tandem application of \cref{prop:APT,prop:algorithms} immediately yields the following concrete conditions for \crefrange{alg:SGA}{alg:EXP3}:

\begin{corollary}
\label{cor:APT}
Suppose that \crefrange{alg:SGA}{alg:EXP3} are run with parameters as in \cref{prop:algorithms}.
Then the sequence $\curr[\dstate]$ comprises an \ac{APT} of \eqref{eq:MD} provided that:
\begin{itemize}
\item
For \cref{alg:SGA,alg:EG,alg:OG}:
	\tabto{11em}
	$\stepexp > 2/(2+\qexp)$
\item
For \cref{alg:EW}:
	\tabto{11em}
	$\stepexp > 0$
\item
For \cref{alg:SPSA,alg:EXP3}:
	\tabto{11em}
	$\stepexp > 2\mixexp > 0$
\end{itemize}
\end{corollary}

\cref{cor:APT} provides a minimal set of hypotheses under which \eqref{eq:MD} is a faithful representation of \crefrange{alg:SGA}{alg:EXP3}.
For some of these algorithms, this property is already known in the literature, see \eg
\cite{Ben99} for \eqref{eq:SGA}
and 
\cite{BBF20} for \eqref{eq:DGA}.
For others, the link with \eqref{eq:MD} appears to be new:
especially in the case of \eqref{eq:EG}\,/\,\eqref{eq:OG}, \cref{cor:APT} settles a standing question in the  literature concerning the mean dynamics of \acl{OG} methods.

\subsection{The primal-dual dichotomy}
\label{sec:primaldual}

To proceed, we will need some basic definitions from the theory of dynamical systems.
Specifically, given
a flow $\flow\from\R\times\mfld\to\mfld$ on some metric space $\mfld$
and
a nonempty compact subset $\set$ of $\mfld$, we say that:
\begin{enumerate}
\addtolength{\itemsep}{\smallskipamount}
\item
$\set$ is \emph{invariant} under $\flow$ if $\flowof{\time}{\set} = \set$ for all $\time\in\R$.
\item
$\set$ is an \emph{attractor} of $\flow$ if it admits a neighborhood $\dnhd \subseteq \dspace$ such that $\dist(\flowof{\time}{\dpoint},\set)\to0$ uniformly in $\dpoint\in\dnhd$ as $\time\to\infty$.
\item
$\set$ is \acdef{ICT} if it is invariant and $\flow\vert_{\set}$ has no attractors except $\set$.
\acused{ICT}
\end{enumerate}
With all this in hand, the general theory of \citef{BH96} yields the following convergence result when applied to \eqref{eq:MD}:

\begin{theorem}
[\citeauthor{BH96}, \citeyear{BH96}]
\label{thm:Benaim}
Let $\curr[\dstate]$, $\run=\running$, be an \ac{APT} of \eqref{eq:MD} with $\sup_{\run} \dnorm{\curr[\dstate]} < \infty$.
Then $\curr[\dstate]$ converges to an \ac{ICT} set of \eqref{eq:MD}.
\end{theorem}

\begin{proof}
\Cref{lem:mirror} in \cref{app:mirror} shows that $\mirror$ is Lipschitz continuous.
Since $\vecfield$ is also Lipschitz continuous by \cref{asm:game}, our assertion follows from \citef[Theorem 0.1]{BH96}.
\end{proof}

Taken together, \cref{thm:Benaim,cor:APT} suggest that the behavior of the various algorithms presented in \cref{sec:algorithms} (and many more) can be understood by looking at the \ac{ICT} sets of the \emph{same} mean dynamics.
However, from a practical viewpoint, this conclusion carries two important limitations:
First, the boundedness caveat for $\curr[\dstate]$ cannot be readily checked against the game's primitives, so it is not clear when \cref{thm:Benaim} applies \textendash\ and, in much of the literature, this assumption has persisted as a condition that needs to be enforced ``by hand'' \cite{KY97,Ben99}.
Second \textendash\ and perhaps more importantly \textendash\ this reasoning ignores the fact that $\curr$ evolves in $\points$, the game's \emph{action space}, whereas the orbits of \eqref{eq:MD} live in $\dspace$, the game's \emph{dual space}.
In turn, this leads to a fundamental mismatch:
a dual orbit $\dtrajof{\time}$ may \emph{diverge} in $\dspace$, even though the induced primal orbit $\trajof{\time} = \mirror(\dtrajof{\time})$ \emph{converges} in $\points$.

\begin{example*}
Consider the single-player game with $\pay(\point) = -\point$, $\point\geq0$.
Then the dynamics \eqref{eq:MD} give
$\dot\dpoint = -1$, so $\dtrajof{\time} \to -\infty$ as $\time\to\infty$, \ie $\dtrajof{\time}$ \emph{diverges};
however, under the exponential mirror map of \cref{ex:mirror-orthant}, the player's trajectory of actions evolves as $\trajof{\time} = \exp(\dtrajof{\time})$, \ie $\trajof{\time}$ \emph{converges} (to $0$).
In this case, even if we were to ignore the boundedness issue, \cref{thm:Benaim} becomes vacuous (and, in a sense, misleading):
the dynamics \eqref{eq:MD} do not have any \ac{ICT} sets and are divergent, even though the induced trajectory of actions converges in $\points$.
\endenv
\end{example*}

The above creates a relatively awkward situation in which \emph{dynamical} notions of stationarity and stability are defined on $\dspace$, whereas the corresponding \emph{game-theoretic} notions reside in $\points$.
To reconcile this incompatibility, it is instead natural to focus directly on $\points$ and ask whether the notion of an \ac{ICT} set can be transposed there.
However, this is only meaningful if the ensemble of trajectories $\trajof{\time} = \mirror(\dtrajof{\time})$ constitute a flow on $\points$;
that is, formally, we must posit the existence of a flow $\pflow$ on $\points$ (or a subset thereof) that is \emph{conjugate} to $\dflow$ in the sense that $\mirror\circ\dflow_{\time} = \pflow_{\time}\circ\mirror$ for all $\time\in\R$.

In general, this may fail to hold:
for example, consider the single-player game $\pay(\point) = \point$, $\point\in[0,1]$, and consider the Euclidean projector $\mirror(\dpoint) = \clip{\dpoint}_{0}^{1} \equiv \min\braces{\max\braces{y,0},1}$ induced by the quadratic regularizer $\hreg(\point) = \point^{2}/2$ on $\points = [0,1]$.
Clearly, \eqref{eq:MD} gives $\dflowof{\time}{\dpoint} = \dpoint + \time$ for all $\time\in\R$ and all $\dpoint \in \R$.
Nevertheless, even though the orbits $\dflowof{\time}{0}$ and $\dflowof{\time}{-1}$ both have the same starting point $\mirror(0) = \mirror(-1) = 0$ in $\points$, the induced primal trajectories evolve differently:
for all $\time\in[0,1]$, we have $\mirror(\dflowof{\time}{0}) = \time$ and $\mirror(\dflowof{\time}{-1}) = 0$, implying in turn that there can be no flow $\pflow$ on $\points$ whose orbits are the images of the orbits of \eqref{eq:MD}.

Because this discrepancy arises at the boundary of $\points$, \cref{thm:Benaim} is more relevant for cases where all orbits are contained in $\relint\points$.
In this case, we have the following result.

\begin{proposition}
\label{prop:conjflow}
Let $\dbasin\subseteq\dspace$ be a forward-invariant set of $\dflow$ such that $\basin = \mirror(\dbasin)$ is contained in $\relint\points$.
Then there exists a flow $\pflow$ on $\basin$ such that $\pflowof{\time}{\mirror(\dpoint)} = \mirror(\dflowof{\time}{\dpoint})$ for all $\dpoint\in\dbasin$ and all $\time\geq0$;
in particular, $\pflow$ can be defined on all of $\relint\points$ if $\im\mirror = \relint\points$.
\end{proposition}

\begin{proof}
Clearly, the only candidate for $\pflow$ is to set $\pflowof{\time}{\point} = \mirror(\dflowof{\time}{\dpoint})$ whenever $\point = \mirror(\dpoint)$ for $\dpoint\in\dbasin$.
To see that this construction is well-defined, suppose that $\mirror(\dpointalt) = \mirror(\dpoint) = \point$ for some $\dpoint,\dpointalt \in \dbasin$, and let $\dvec = \dpointalt - \dpoint$;
then, it suffices to show that $\mirror(\dflowof{\time}{\dpoint}) = \mirror(\dflowof{\time}{\dpoint+\dvec})$ for all $\time\geq0$.

As a first step, we claim that $\mirror(\dflowof{\time}{\dpoint} + \dvec) = \mirror(\dflowof{\time}{\dpoint})$ for all $\time\geq0$.
Indeed, since $\mirror(\dpoint) \in \relint\points$, \cref{lem:mirror} in \cref{app:mirror} shows that $\dvec$ annihilates all tangent directions to $\points$ at $\point$, \ie $\braket{\dvec}{\pointalt - \point} = 0$ for all $\pointalt\in\points$.
However, since $\dflowof{\time}{\dbasin} \subseteq \dbasin$ and $\mirror(\dbasin) \subseteq \relint\points$, the point $\point_{\time} = \mirror(\dflowof{\time}{\dpoint})$ will also be interior;
hence, with $\dflowof{\time}{\dpoint} \in \subd\hreg(\point_{\time})$ and $\braket{\dvec}{\pointalt - \point_{\time}} = \braket{\dvec}{\pointalt - \point} + \braket{\dvec}{\point - \point_{\time}} = 0$ for all $\pointalt\in\points$, invoking \cref{lem:mirror} in the converse direction gives $\dflowof{\time}{\dpoint} + \dvec \in \subd\hreg(\point_{\time})$, \ie $\mirror(\dflowof{\time}{\dpoint} + \dvec) = \mirror(\dflowof{\time}{\dpoint})$.

In view of the above, a simple differentiation yields
\begin{equation}
\ddt\bracks{\dflowof{\time}{\dpoint} + \dvec}
	= \ddt \dflowof{\time}{\dpoint}
	= \vecfield(\mirror(\dflowof{\time}{\dpoint}))
	= \vecfield(\mirror(\dflowof{\time}{\dpoint} + \dvec))
\end{equation}
which, by the Picard-Lindelöf theorem, shows that $\dflowof{\time}{\dpoint} + \dvec$ is the (necessarily unique) solution orbit of \eqref{eq:MD} starting at $\dpoint+\dvec$ at time $\time = \tstart$.
We thus conclude that
$\mirror(\dflowof{\time}{\dpoint + \dvec}) = \mirror(\dflowof{\time}{\dpoint} + \dvec) = \mirror(\dflowof{\time}{\dpoint})$, \ie $\pflow$ is well-defined.
Our second assertion follows from the fact that $\im\mirror = \dom\subd\hreg = \relint\points$, so we can apply our first claim to $\dbasin = \dspace$.
\end{proof}

\Cref{prop:conjflow} indicates that, if $\mirror$ is interior-valued, we can map the flow on $\dspace$ to an induced primal flow on $\points$.
For concreteness, before discussing the precise connection between \eqref{eq:method} and the induced dynamics on $\points$, we illustrate the latter in the case of \cref{ex:mirror-Eucl,ex:mirror-logit,ex:mirror-orthant}:

\begin{example}
Take $\points=\pspace$ and $\mirror(\dpoint) = \dpoint$ as in \cref{ex:mirror-Eucl}.
Then \eqref{eq:MD} trivially gives the (individual) \emph{gradient dynamics}
\begin{equation}
\label{eq:GD}
\tag{GD}
\dot\point
	= \vecfield(\point)
\end{equation}
\end{example}

\begin{example}
Let $\points_{\play} = \simplex(\pures_{\play})$ and take $\mirror_{\play} = \logit_{\play}$ as in \cref{ex:mirror-logit}.
Then, by a standard calculation, \eqref{eq:MD} boils down to the \emph{replicator dynamics} of \citet{TJ78}
\begin{equation}
\label{eq:RD}
\tag{RD}
\dot\point_{\play\pure_{\play}}
	= \point_{\play\pure_{\play}}
		\bracks{\pay_{\play}(\pure_{\play};\point_{-\play}) - \pay_{\play}(\point_{\play};\point_{-\play})}.
\end{equation}
\end{example}

\begin{example}
Let $\points_{\play} = [0,\infty)$ and take $\mirror_{\play}(\dpoint_{\play}) = \exp(\dpoint_{\play})$ as in \cref{ex:mirror-orthant}.
Then, by differentiating $\point_{\play} = e^{\dpoint_{\play}}$ we obtain the \emph{dampened gradient dynamics} of \citet{BBF20}, viz.
\begin{equation}
\label{eq:DGD}
\tag{DGD}
\dot\point_{\play}
	= \point_{\play} \vecfield_{\play}(\point).
\end{equation}
\end{example}

\subsection{Subcoercivity and convergence}
\label{sec:subcoercive}


Other than the primal-dual dichotomy described above, the other important caveat in \cref{thm:Benaim} is the boundedness of the generated sequence $\curr[\dstate]$.
In the optimization literature, a standard way to establish this type of control on an algorithm designed to find a zero of a vector field $\vecfield$ on $\R^{\vdim}$ is to assume that it is \emph{coercive}, \ie $\lim_{\norm{\point}\to\infty}\braket{\vecfield(\point)}{\point} / \norm{\point} = -\infty$.
Intuitively, coercivity means that $\vecfield(\point)$ is strongly ``inward pointing'' for large values of $\point$, so it acts as a natural barrier against escape phenomena;
at the same time, it also imposes that $\dnorm{\vecfield(\point)}$ grows superlinearly in $\point$, and is only relevant for unconstrained problems, so it is not particularly well-suited for our purposes.
Instead, we will consider the following, weaker requirement:

\begin{definition}
\label{def:subco}
We say that $\game$ is \emph{subcoercive} if there exists a compact set $\cpt \in \relint\points$ and a reference point $\base\in\cpt$ such that
\begin{equation}
\label{eq:subco}
\tag{SC}
\braket{\vecfield(\point)}{\point - \base}
	\leq 0
	\quad
	\text{for all $\point \in \points \setminus \cpt$}.
\end{equation}
\end{definition}

Geometrically, subcoercivity simply posits that the Nash field $\vecfield(\point)$ of the game points weakly towards $\base$ outside $\cpt$, so any ``attracting'' behavior in $\game$ must be contained in $\cpt$:
for example, it is straightforward to verify that any \acl{VS} state of $\game$ must lie within $\cpt$ if \eqref{eq:subco} holds.
Beyond this, it is important to note that $\vecfield$ may vanish at infinity, and $\cpt$ can be arbitrarily large (relative to $\points$).
We provide some examples below:

\begin{example}
[Potential games]
Suppose that $\game$ admits a quasiconcave potential $\pot$ with $\argmax\pot \subseteq \relint\points$.
If we fix a maximizer $\base$ of $\pot$, we have $\braket{\vecfield(\point)}{\point - \base} = \braket{\nabla\pot(\point)}{\point - \base} \leq 0$ for all $\point\in\points$, so $\game$ is subcoercive.
More generally, $\game$ is subcoercive whenever $\pot$ is ``eventually quasiconcave'', \ie the upper level sets $\suplvl{\lvl}{\pot} = \setdef{\point\in\points}{\pot(\point) \geq \lvl}$ of $\pot$ are convex for sufficiently small $\lvl>\inf\pot$ and at least one such set is contained in $\relint\points$.%
\footnote{To see this, let $\cpt = \suplvl{\lvl_{0}}{\pot}$ be a convex upper level set of $\pot$ in $\relint\points$.
Then, for all $\lvl \leq \lvl_{0}$ and all $\point$ with $\pot(\point) = \lvl$, the segment $\point + \tau(\base-\point)$, $\tau\in[0,1]$, is contained in $\suplvl{\lvl}{\pot} \supseteq \suplvl{\lvl_{0}}{\pot}$, so the function $\phi(\tau) = \pot(\point + \tau(\base-\point))$ cannot have $\phi'(0) < 0$.
This implies that $0 \leq \braket{\nabla\pot(\point)}{\base - \point} = \braket{\vecfield(\point)}{\base - \point}$ for all $\point\in\points\setminus\cpt$, \ie $\game$ is subcoercive.}
\endenv
\end{example}

\begin{example}
[Min-max games]
Consider the toy game $\min_{\minvar\in[-1,1]} \max_{\maxvar\in[-1,1]} \minvar\maxvar$.
Since $\braket{\vecfield(\point)}{\point} = -\maxvar\minvar + \minvar\maxvar = 0$ for all $\point\in[-1,1]\times[-1,1]$, the game is trivially subcoercive.
More generally, it is easy to check that any two-player, quasi-convex\,/\,quasi-concave game with an interior equilibrium is subcoercive.
\endenv
\end{example}

By itself, subcoercivity ensures that there is no consistent drift pointing away from $\cpt$, so it is reasonable to expect that $\curr[\dstate]$ does not escape to infinity either.
To control the inherent stochasticity in $\curr[\dstate]$ and make this intuition precise, we will require the following summability conditions on the bias, variance, and magnitude of the gradient signal process $\curr[\signal]$:
\begin{equation}
\label{eq:sum}
\tag{Sum}
\txs
\sum_{\run} \curr[\step] \curr[\bbound]
	< \infty
	\qquad
\sum_{\run} \curr[\step]^{2} \curr[\sdev]^{2}
	< \infty
	\qquad
	\text{and}
	\qquad
\sum_{\run} \curr[\step]^{2} \curr[\sbound]^{2}
	< \infty.
\end{equation}
Under these conditions, we have the following stability result.

\begin{proposition}
\label{prop:bounded}
Suppose that \eqref{eq:method} is run with step-size and gradient signal sequences satisfying \eqref{eq:sum}.
If $\game$ is subcoercive,
the sequence of iterates $\curr[\dstate]$ generated by \eqref{eq:method} is bounded \ac{wp1}.
\end{proposition}

Before proving \cref{prop:bounded}, it is important to note that subcoercivity only concerns the primitives of the game under study, and it is otherwise ``algorithm-agnostic''.
In this regard, given the primal-dual nature of the underlying dynamics \eqref{eq:MD}, \cref{prop:bounded} plays a major role in enabling the use of stochastic approximation tools and techniques (otherwise, the boundedness of $\points$ by itself does not suffice).

Moreover, from an operational viewpoint, \cref{prop:algorithms} makes the verification of \eqref{eq:sum} a trivial affair for the algorithms under study.
In particular, a joint application of \cref{cor:APT,prop:algorithms,prop:bounded,thm:Benaim} readily yields the general convergence result below:

\begin{theorem}
\label{thm:ICT}
Suppose that \crefrange{alg:SGA}{alg:EXP3} are run with step-size $\curr[\step] \propto 1/\run^{\stepexp}$, $\stepexp\in(1/2,1]$, and, where applicable, a sampling parameter $\curr[\mix] \propto 1/\run^{\mixexp}$ with $1-\stepexp < \mixexp < \stepexp-1/2$.
If $\game$ is subcoercive, then:
\begin{enumerate*}
[\itshape a\upshape)]
\item
$\curr[\dstate]$ converges to an \ac{ICT} set of \eqref{eq:MD} \acl{wp1};
and, in addition,
\item
if the players' mirror map $\mirror$ is interior-valued, the induced sequence of play $\curr = \mirror(\curr[\dstate])$ converges \acl{wp1} to an \ac{ICT} set of the primal flow $\pflow$ on $\points$.
\end{enumerate*}%
\end{theorem}

We then have the following consequences for potential and zero-sum games (both stated for simplicity under the assumption that $\mirror$ is interior-valued): 

\begin{corollary}
\label{cor:potential}
If $\game$ admits a subcoercive potential, $\curr$ converges to a component of critical points of $\game$ \ac{wp1}.
In particular, if the potential is concave, $\curr$ converges to the set of \aclp{NE} of $\game$.
\end{corollary}

\begin{corollary}
\label{cor:minmax}
Suppose that $\game$ is a strictly convex-concave min-max game with an interior equilibrium $\eq\in\relint\points$.
Then $\curr$ converges to $\eq$ \ac{wp1}.
\end{corollary}

\Cref{cor:potential,cor:minmax} follow respectively from the fact that the only \ac{ICT} sets of potential games and strictly convex-concave games are their sets of critical points, see \eg \cite{Ben99,MZ19} and references therein.
As for \cref{thm:ICT} (which we prove below), it should not be viewed as an equilibrium convergence guarantee, but as a characterization of what types of behaviors may arise in the limit of a game-theoretic learning process \textendash\ equilibrium and non-equilibrium alike.
However, because of the subcoercivity requirement, this characterization only extends to limit sets that are contained in the relative interior of the players' action spaces;
games with boundary solutions require a different treatment, which we undertake in the next two sections.

\subsection{Technical proofs}
\label{sec:proofs-ICT}

We conclude this section with the proof of \cref{prop:bounded,thm:ICT};
we begin with the latter, which is more conceptual and less technical.

\begin{proof}[Proof of \cref{thm:ICT}]
By \cref{prop:algorithms,prop:bounded}, $\curr[\dstate]$ is a bounded \ac{APT} of \eqref{eq:MD}, so the first part of the theorem follows directly from \cref{thm:Benaim}.
As for the second, since $\mirror$ is Lipschitz continuous (\cf \cref{lem:mirror} in \cref{app:mirror}), the sequence $\curr = \mirror(\curr[\dstate])$ is also bounded and, in addition, we have:
\begin{align}
\norm{\papt{\time + \timealt} - \pflowof{\timealt}{\papt{\time}}}
	&= \norm{\mirror(\apt{\time + \timealt}) - \pflowof{\timealt}{\mirror(\apt{\time})}}
	\notag\\
	&= \norm{\mirror(\apt{\time + \timealt}) - \mirror(\dflowof{\timealt}{\apt{\time}})}
	\leq (1/\hstr) \dnorm{\apt{\time + \timealt} - \dflowof{\timealt}{\apt{\time}}}
\end{align}
where $\hstr$ is the strong convexity modulus of $\hreg$ and we used the fact that $\pflow_{\time} \circ \mirror = \mirror \circ \dflow_{\time}$ for all $\time$ (by \cref{prop:conjflow}).
This implies that $\curr$ is an \ac{APT} of $\pflow$,%
\footnote{We are grateful to V.~Boone for pointing out this simple argument.}
so our claim follows from the limit set theorem of \citet[Theorem 0.1]{BH96}.
\end{proof}

We are thus left to prove the boundedness guarantee of \cref{prop:bounded}.

\begin{proof}[Proof of \cref{prop:bounded}]
Our proof hinges on the construction of a suitable ``energy function'' $\energy\from\dspace\to\R_{+}$ for \eqref{eq:method}.
To define it, we will assume for simplicity \textendash\ and without loss of generality \textendash\ that $\points$ has nonempty topological interior in $\pspace$ (which can be achieved by redefining $\pspace$ to be the affine hull of $\points$), that the reference point $\base$ in \cref{def:subco} is the origin $0 \in \pspace$, and that $\hreg(\base) = 0$ (which can be achieved by a simple translation).

With this in mind, let $\hconj(\dpoint) = \max_{\point\in\points}\{\braket{\dpoint}{\point} - \hreg(\point)\}$ denote the convex conjugate of $\hreg$.
Then, by \cref{lem:Fench} in \cref{app:mirror}, we have
\begin{equation}
\label{eq:hconj-bounds}
(\hstr/2) \norm{\mirror(\dpoint)}^{2}
	\leq \hconj(\dpoint)
	\leq -\min\hreg
		+ \braket{\dpoint}{\mirror(\dpoint)}
		+ (2/\hstr) \dnorm{\dpoint}^{2}
	\quad
	\text{for all $\dpoint\in\dspace$}
\end{equation}
where we note that $\min\hreg \leq \hreg(\base) = 0$ by assumption.
Since $\hreg$ is lower-semicontinuous, we have $\hreg = \hreg^{\ast\ast}$ by the Fenchel\textendash Moreau theorem.
In addition, the Moreau\textendash Rockafellar theorem \citep[Theorem~4.17]{BC17} implies that $\hconj$ is coercive because it can be written as $\hconj(\dpoint) = \hconj(\dpoint) - \braket{\dpoint}{\base}$ and $0 = \base \in\relint\points = \relint\dom\hreg^{\ast\ast}$ by subcoercivity.
Finally, since $\points$ has nonempty interior, it follows that the polar cone $\pcone(\point)$ is trivially $0$ for all $\point\in\relint\points$, so the subdifferential $\subd\hreg$ of $\hreg$ is compact-valued on $\cpt \subseteq \relint\points$.
Thus, by the upper hemicontinuity of the subdifferential and the compactness of $\cpt$, we deduce that the image $\dset = \subd\hreg(\cpt)$ of $\cpt$ under $\subd\hreg$ is compact, \cf \cite[p.~201]{HUL01}.
Hence, by the coercivity of $\hconj$ and the fact that $\mirror(\dpoint) = \point$ if and only if $\subd\hreg(\point) \ni \dpoint$ (\cf \cref{lem:mirror} in \cref{app:mirror}), there exists some $\lvl>0$ such that $\hconj(\dpoint) \leq \lvl$ whenever $\mirror(\dpoint) \in \cpt$, \ie $\mirror^{-1}(\cpt)$ is contained in the $\lvl$-sublevel set $\sublvl{\lvl}{\hconj}$ of $\hconj$.

With all this said and done, fix some $\alt\lvl > \lvl$ and let
\begin{equation}
\label{eq:energy-gauge}
\energy(\dpoint)
	= \gauge(\hconj(\dpoint))
	\quad
	\text{for all $\dpoint\in\dspace$}
\end{equation}
where $\gauge\from\R_{+}\to\R_{+}$ is a $C^{2}$-smooth ``gauge function'' with the following properties:
\begin{enumerate*}
[\itshape i\hspace*{.5pt}\upshape)]
\item
$\gauge(u) = 0$ for $u\leq \lvl$;
\item
$\gauge(u) = \sqrt{u}$ for $u \geq \alt\lvl$;
\item
$\gauge'(u) \geq 0$ and $\gauge''(u) \leq 1$ for all $u\in\R_{+}$.%
\footnote{That such a function exists is an exercise in the construction of aproximate identities, which we omit.}
\end{enumerate*}
Then, setting $\point = \mirror(\dpoint)$ and differentiating, we readily obtain
\begin{equation}
\label{eq:energy-diff}
\nabla\energy(\dpoint)
	= \gauge'(\hconj(\dpoint)) \cdot \nabla\hconj(\dpoint)
	= \gauge'(\hconj(\dpoint)) \cdot \point
	\quad
	\text{for all $\dpoint\in\dspace$}
\end{equation}
and hence, by the smoothness properties of $\gauge$ and $\hconj$, there exists some constant $\Const_{2}\geq0$ such that
\begin{align}
\label{eq:energy-incr}
\energy(\dpoint + \dvec)
	= \energy(\dpoint)
		+ \gauge'(\hconj(\dpoint)) \cdot \braket{\dvec}{\point}
		+ \Const_{2} \dnorm{\dvec}^{2}
	\quad
	\text{for all $\dpoint,\dvec\in\dspace$}.
\end{align}
Therefore, combining \cref{eq:energy-diff,eq:energy-incr} and letting $\curr[\energy] = \energy(\curr[\dstate])$, we obtain
\begin{align}
\label{eq:energy-bound1}
\next[\energy]
	&\leq \curr[\energy]
		+ \gauge'(\hconj(\curr[\dstate])) \cdot \braket{\curr[\signal]}{\curr}
		+ \Const_{2} \dnorm{\curr[\signal]}^{2}
	\leq \curr[\energy]
		+ \curr[\varphi] \braket{\curr[\bias] + \curr[\noise]}{\curr}
		+ \Const_{2} \dnorm{\curr[\signal]}^{2}
\end{align}
where we set $\curr[\varphi] = \gauge'(\hconj(\curr[\dstate]))$ and we used the fact that $\gauge(\hconj(\dpoint)) \cdot \braket{\vecfield(\point)}{\point} \leq 0$ for all $\dpoint\in\dspace$ (the latter being a consequence of subcoercivity and the defining properties of $\gauge$).
Accordingly, conditioning on $\curr[\filter]$ and taking expectations, we finally get
\begin{equation}
\label{eq:energy-bound2}
\exof{\next[\energy] \given \curr[\filter]}
	\leq \curr[\energy]
		+ \curr[\step] \curr[\varphi] \curr[\bbound] \norm{\curr}
		+ \curr[\step]^{2} \curr[\sbound]^{2},
\end{equation}
where we used the Cauchy-Schwarz inequality to bound $\braket{\curr[\bias]}{\curr}$ from above by $\curr[\bbound] \norm{\curr}$ (recall also that $\exof{\curr[\noise] \given \curr[\filter]} = 0$ by definition).

Now, let $\curr[\eps] = \curr[\step] \curr[\varphi] \curr[\bbound] \norm{\curr} + \curr[\sbound]^{2}$ denote the ``residual'' term in \eqref{eq:energy-bound2}, and consider the auxiliary process $\curr[E] = \next[\energy] + \sum_{\runalt=\run+1}^{\infty} \iter[\eps]$.
By \eqref{eq:energy-bound2}, we have $\exof{\curr[E] \given \curr[\filter]} \leq \curr[\energy] + \sum_{\runalt=\run}^{\infty} \curr[\eps] = \prev[E]$, \ie $\curr[E]$ is a supermartingale relative to $\curr[\filter]$.
Moreover, by \eqref{eq:hconj-bounds} and the definition of $\gauge$, we further have
\begin{equation}
\curr[\varphi]
	= \frac{1}{2\sqrt{\hconj(\curr[\dstate])}}
	\leq \frac{1}{\sqrt{2\hstr} \norm{\curr}}
	\quad
	\text{whenever $\hconj(\curr[\dstate]) \geq \alt\lvl$}
\end{equation}
so there exists some (deterministic) positive constant $\Const_{1}$ such that $\sup_{\run} \curr[\varphi] \norm{\curr} \leq \Const_{1}$.
We thus get
\begin{equation}
\sum_{\run=\start}^{\infty} \curr[\eps]
	\leq \Const_{1} \sum_{\run=\start}^{\infty} \curr[\step]\curr[\bbound]
		+ \Const_{2} \sum_{\run=\start}^{\infty} \curr[\step]^{2} \curr[\sbound]^{2}
	< \infty
\end{equation}
by the summability condition \eqref{eq:sum}.
This shows that $\exof{\sum_{\run} \curr[\eps]} < \infty$ and, in turn, that $\exof{\curr[E]} \leq \exof{\init[E]} < \infty$, \ie $\curr[E]$ is uniformly bounded in $L^{1}$.
Accordingly, by Doob's submartingale convergence theorem \citep[Theorem~2.5]{HH80}, it follows that $\curr[E]$ converges \acl{wp1} to some finite random limit $E_{\infty}$.
Since $\sum_{\run} \curr[\eps] < \infty$, this implies that $\curr[\energy] = \prev[\energy] - \sum_{\runalt=\run}^{\infty} \curr[\eps]$ also converges to some (random) finite limit \as.
Therefore, by the coercivity of $\energy$, we deduce that $\limsup_{\run} \dnorm{\curr[\dstate]} < \infty$ \ac{wp1}, as claimed.
\end{proof}

%% file: Body/Primal.tex

Even though \Cref{thm:ICT} provides a universal characterization of the long-run behavior of any algorithm of the general form \eqref{eq:method}, there are several issues that remain open, namely:
\begin{enumerate}
\item
How is the long-run behavior of an algorithm affected by a small perturbation in its initialization?
\item
Does the update structure of the gradient-like signals $\curr[\signal]$ affect the algorithm's end-state?
\item
Are some limit sets independent of the amount of information available to the players?
\end{enumerate}
In view of all this, the rest of this section will focus on whether we can identify a class of ``robust'' limit sets that satisfy the above desiderata.
We present and discuss our main results in \cref{sec:results} after some necessary definitions and prerequisites in \cref{sec:energy}.


\begin{figure}[t]
\footnotesize
\centering
\includegraphics[height=.425\textwidth]{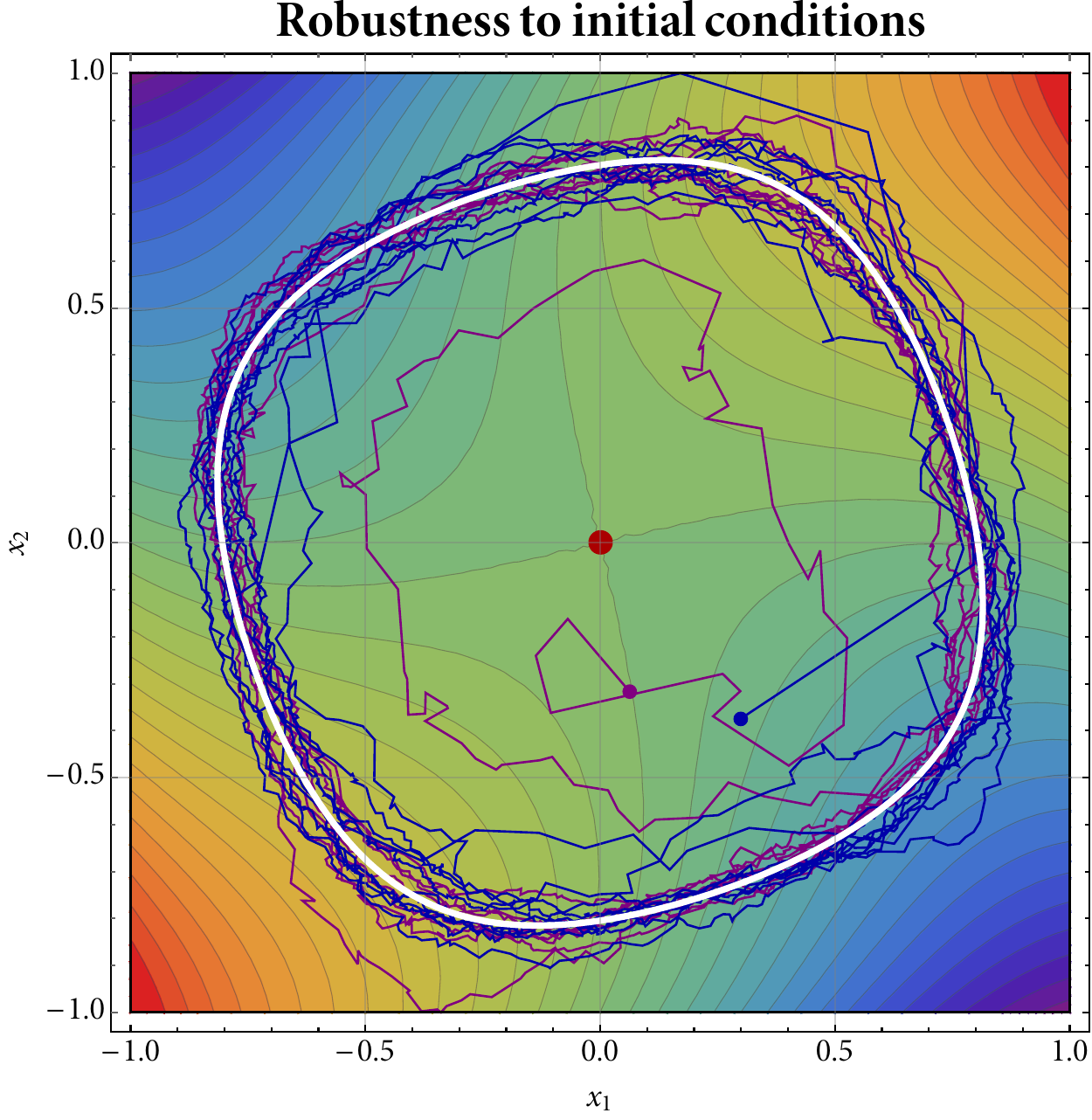}
\hspace{3em}
\includegraphics[height=.425\textwidth]{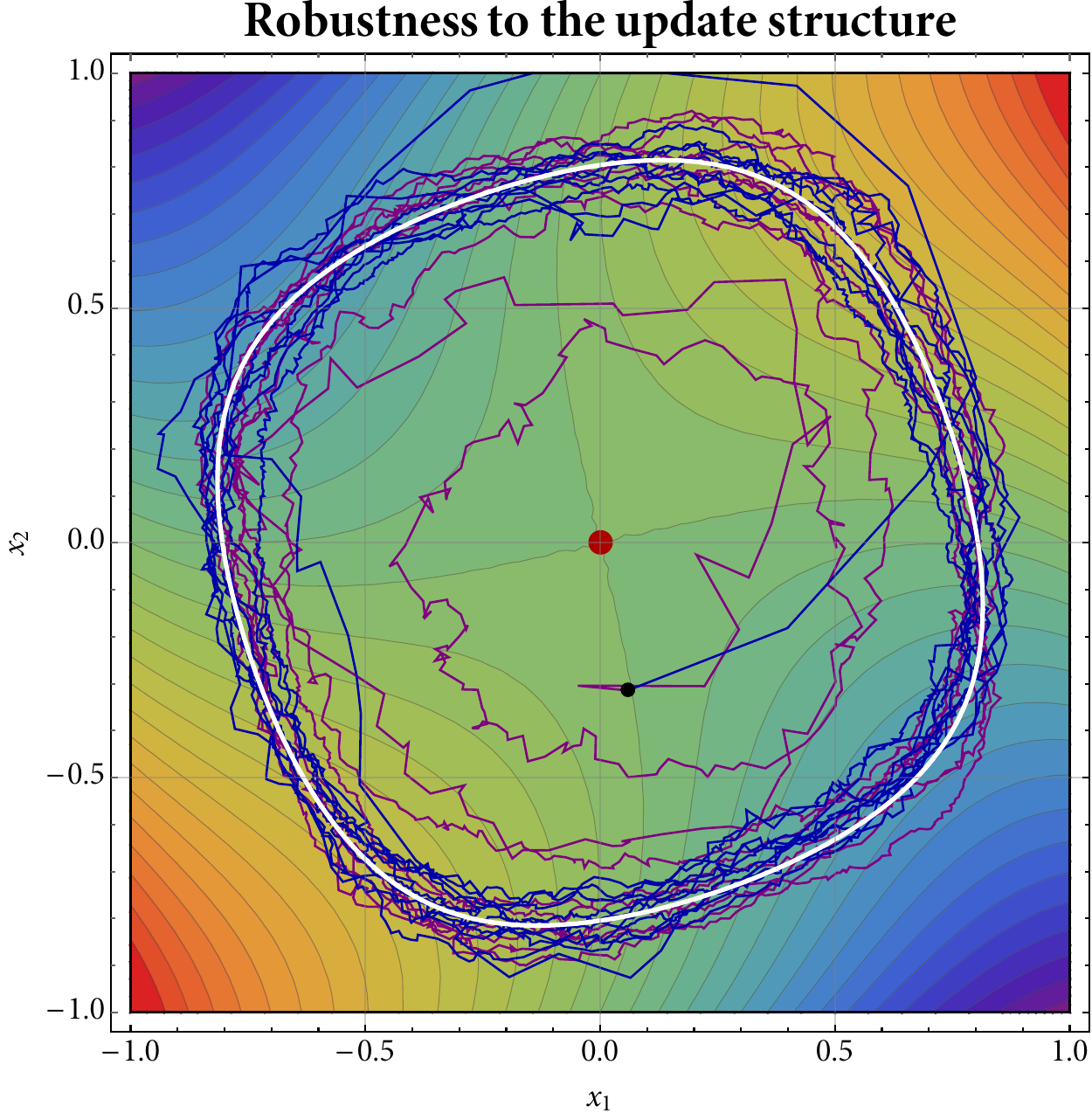}
\quad\
\caption{The long-run behavior of different online learning algorithms in the two-player min-max game defined by the loss function $\minmax(\minvar,\maxvar) = \minvar\maxvar + \eps \bracks{\phi(\maxvar) - \phi(\minvar)}$ with $\phi(z) = 2z^{2} - 4z^{4}$, $\eps \geq 0$, and $\minvar,\maxvar \in [-1,1]$.
The left plot shows two different initializations of \eqref{eq:SGA}, while the right plot displays the trajectories of \eqref{eq:SGA} and \eqref{eq:EG} from the same initialization.
In both cases, the existence of an attractor allows for robust predictions that are largely independent of the initialization or exact update structure of \eqref{eq:method}.
}
\label{fig:limits}
\end{figure}


\subsection{Stochastically attracting sets and energy functions}
\label{sec:energy}

In the theory of dynamical systems, the established way of analyzing such questions is via the notion of an attractor (\cf the relevant discussion in \cref{sec:general}).
Following \citet{NK76}, this notion can be adapted to our stochastic setting as follows:

\begin{definition}
\label{def:attract}
Let $\set$ be a nonempty closed subset of $\points$.
We say that $\set$ is \emph{stochastically attracting} under \eqref{eq:method} for a given tolerance level $\conf>0$ if there exists a neighborhood $\nhd$ of $\set$ in $\points$ such that
\begin{equation}
\label{eq:attract}
\probof{\text{$\curr$ converges to $\set$} \given \init\in\nhd}
	\geq 1 - \conf.
\end{equation}
\end{definition}

In the context of stochastic approximation algorithms, the requirement \eqref{eq:attract} is reminiscent of results guaranteeing convergence with positive probability toward an attractor.
Such guarantees are usually conditioned on the notion of \emph{attainability}, as pioneered by \citet{Ben97} and \citet{Duf97};%
\footnote{A point $\dpoint\in\dspace$ is said to be \emph{attainable} by $\curr[\dstate]$ if, for every neighborhood $\dnhd$ of $\dpoint$ in $\dspace$ and for all $\run\geq\start$, we have $\probof{\iter[\dstate]\in\dnhd \text{ for some $\runalt\geq\run$}} > 0$.}
however, in the present setting, there are two salient difficulties with this approach, both having to do with the primal-dual nature of the mean dynamics \eqref{eq:MD}.
On the one hand, if the players' mirror map $\mirror\from\dspace\to\points$ is surjective on the boundary of $\points$ (\eg as in the case of Euclidean projections), it is in general impossible to define a conjugate primal flow on $\points$ \textendash\ and hence, it is not possible to treat $\set$ as an attractor.
On the other hand, if $\mirror$ is interior-valued (that is, $\im\mirror = \relint\points$),
the defining \acl{RM} process $\curr[\dstate]$ may escape to infinity if $\mirror^{-1}(\set) = \varnothing$, and establishing attainability in this case can be as difficult as the original problem of proving \eqref{eq:attract} directly.

On account of all this, we will instead seek to establish \eqref{eq:attract} via a primal-dual variant of Lyapunov's direct method.
In particular, since we are interested in the attracting properties of subsets of the \emph{primal space} $\points\subseteq\pspace$, but the dynamics evolve in the \emph{dual space} $\dspace = \pspace^{\ast}$ of $\pspace$, our analysis will hinge on the following construction:

\begin{definition}
\label{def:energy}
Let $\set$ be a nonempty closed subset of $\points$.
We will say that $\energy\from\dspace\to[0,\infty)$ is a \emph{local energy function for $\set$ under \eqref{eq:MD} if}
\begin{enumerate*}
[\itshape a\upshape)]
\item
$\energy$ is Lipschitz continuous and smooth;
\item
$\mirror(\dpoint) \to \set$ if and only if $\energy(\dpoint) \to 0$;
and
\item
$\sup\setdef{\dot\energy(\dpoint)}{\loen < \energy(\dpoint) < \hien} < 0$
for all sufficiently small $\hien > \loen > \ground$.
\end{enumerate*}
In particular, if the last requirement holds for all $\hien \leq \sup\energy$, we will refer to $\energy$ as a \emph{global} energy function for $\set$.
\end{definition}

Informally, \cref{def:energy} posits that $\energy$ is smooth, positive-definite, and strictly decreasing along all nearby primal orbits $\trajof{\time} = \mirror(\dtrajof{\time})$ that do not lie in $\set$.
For concreteness, we provide below a series of representative examples that will play an essential part in the sequel.

\begin{example}
[Variational stability]
\label{ex:energy-VS}
Suppose that $\eq$ satisfies \eqref{eq:VS}, \ie $\braket{\vecfield(\point)}{\point - \eq} < 0$ for all $\point \neq \eq$ in some neighborhood $\nhd$ of $\eq$ in $\points$.
Then a suitable primal-dual measure of distance from $\eq$ is provided by the so-called ``Fenchel coupling'' \cite{MS16}
\begin{equation}
\label{eq:coupling}
\fench(\dpoint)
	= \hreg(\eq)
		+ \hconj(\dpoint)
		- \braket{\dpoint}{\eq}.
\end{equation}
The key property of this coupling is that, under \eqref{eq:MD}, we have
\begin{equation}
\label{eq:dFench}
\dot\fench(\dpoint)
	= \braket{\dot\dpoint}{\nabla\hconj(\dpoint)}
		- \braket{\dot\dpoint}{\eq}
	= \braket{\vecfield(\point)}{\point - \eq}
	< 0
	\quad
	\text{whenever $\point\in\nhd\exclude{\eq}$}
\end{equation}
where, in the penultimate step, we set $\point = \mirror(\dpoint)$ and we invoked \cref{lem:mirror} in \cref{app:mirror} to write $\mirror(\dpoint) = \nabla\hconj(\dpoint)$.
By the Fenchel-Young inequality, we also have $\fench(\dpoint) \geq 0$ with equality if and only if $\mirror(\dpoint) = \eq$ (\cf \cref{lem:Fench}), so $\fench(\dpoint)$ is a prime candidate for a local energy function.

To meet the entire range of requirements of \cref{def:energy}, we will need two further technical ingredients.
The first is a regularity assumption on $\hreg$, namely that
\begin{equation}
\label{eq:rec}
\tag{R}
\hreg(\curr[\point]) + \braket{\curr[\dpoint]}{\point - \curr[\point]}
	\to \hreg(\point)
\end{equation}
for all $\point\in\points$ and all sequences of primal points $\curr[\point]\to\point$ and subgradients $\curr[\dpoint] \in \subd\hreg(\curr[\point])$.
This condition simply posits that the first-order approximation of $\hreg(\point)$ from $\hreg(\curr[\point])$ is always accurate when $\curr[\point]\to\point$, a property which is satisfied by all examples of regularizers that we have considered so far;
for an in-depth discussion, \cf \citet{AIMM22} and references therein.

The second technicality is that $\fench$ must grow at most linearly in $\dnorm{\dpoint}$ in order to ensure the global Lipschitz continuity requirement of \cref{def:energy}.
To achieve this, it suffices to rescale $\fench$ for large values of $\point$ by means of the gauge function
\begin{equation}
\label{eq:gauge}
\gauge(z)
	= \begin{dcases*}
		z
			&\quad
			if $0 \leq z \leq 1$
		\\
		2 \sqrt{z} - 1
			&\quad
			if $z \geq 1$
	\end{dcases*}
\end{equation}
which ensures that $\gauge \circ \fench$ behaves like $\fench$ for small values of $\fench$, and like $\sqrt{\fench}$ for large values of $\fench$.
We then have the following result:

\begin{lemma}
\label{lem:energy-VS}
Suppose that $\eq$ satisfies \eqref{eq:VS} and the players' regularizers satisfy \eqref{eq:rec}.
Then, with notation as above, the function $\energy(\dpoint) = \gauge(\fench(\dpoint))$ is a local energy function for $\eq$ under \eqref{eq:MD};
moreover, if $\eq$ is globally stable, $\energy$ is a global energy function for $\eq$ under \eqref{eq:MD}.
\end{lemma}

To streamline our presentation, we defer the proof of \cref{lem:energy-VS} to \cref{app:mirror}.
We only note here that, by the relevant discussion in \cref{sec:prelims}, the above yields an energy function for a wide class of games, including
\begin{enumerate*}
[\itshape a\upshape)]
\item
strictly monotone games (a global one in this case);
\item
games with second-order stationary equilibria as per \eqref{eq:SOS};
and
\item
all finite games admitting a strict \acl{NE}.
\end{enumerate*}
\endenv
\end{example}

\begin{example}
[Convex optimization]
\label{ex:energy-cvx}
Consider the convex minimization problem $\min_{\point\in\points} \obj(\point)$
where $\obj\from\points\to\R$ is a smooth convex function with a nonempty, compact set of minimizers $\set = \argmin\obj$.
To get a primal-dual measure of distance from $\set$, we may extend the definition of the coupling \eqref{eq:coupling} to the current setting as
\begin{equation}
\label{eq:coupling-set}
\fench(\dpoint)
	= \hconj(\dpoint) - \hconj_{\set}(\dpoint)
\end{equation}
where $\hconj_{\set}(\dpoint) = \max_{\point\in\set} \{\braket{\dpoint}{\point} - \hreg(\dpoint)\}$ denotes the convex conjugate of $\hreg$ relative to $\set$.
As we show below, rescaling $\fench$ by the gauge function \eqref{eq:gauge} yields a global energy function for $\set$ under \eqref{eq:MD}:

\begin{lemma}
\label{lem:energy-cvx}
Suppose that \eqref{eq:rec} holds.
Then, with notation as above, the function $\energy(\dpoint) = \gauge(\fench(\dpoint))$ is a global energy function for $\set = \argmin\obj$.
\end{lemma}

\noindent
As before, to keep the discussion going, we defer the proof of \cref{lem:energy-cvx} to \cref{app:mirror}.
\endenv
\end{example}

\begin{example}
[Discoordination games]
\label{ex:energy-discord}
As a last example, consider a two-player discoordination game with payoff functions
$\pay_{1}(\point_{1},\point_{2}) = (\point_{1} - \point_{2})^{2}/2$
and
$\pay_{2}(\point_{1},\point_{2}) = (\point_{1}+\point_{2})^{2}/2$
for $\point_{1},\point_{2}\in[-1,1]$.
This game admits five critical points, the origin $(0,0)$ and the four vertices $\{\pm1,\pm1\}$ of $\points = [-1,1]^{2}$.
None of these critical points is an equilibrium:
the origin is unstable to deviations by both players, whereas the vertices are unstable to deviations by one of the players (but not the other).
Given the lack of an equilibrium in pure strategies (a standard feature of discoordination games), the players' limiting behavior is quite difficult to predict;
however, since the critical point at $(0,0)$ is unstable for both players, it is reasonable to expect that it should be selected against.

To examine this issue in the context of \eqref{eq:MD}, consider for concreteness the mirror map $\mirror_{\play}(\dpoint_{\play}) = \tanh(\dpoint_{\play}/2)$ that is induced by the entropic regularizer $\hreg_{\play}(\point_{\play}) = (1-\point_{\play}) \log(1-\point_{\play}) + (1+\point_{\play})\log(1+\point_{\play})$.%
\footnote{For the general case, take $\energy(\dpoint) = [\hconj(\dpoint) - \inf\hconj]^{-1}$.}
In this case, it is straightforward to check that $\energy(\dpoint_{1},\dpoint_{2}) = 2 \sech(\dpoint_{1}/2)\sech(\dpoint_{2}/2)$ is an (almost global) local energy function for the four-corner set $\set = \{-1,1\}\times\{-1,1\}$.
As a result, the sequence of play generated by \eqref{eq:method} is expected to spend most time near one of these points, \cf \cref{fig:discord}.
\endenv
\end{example}


\begin{figure}[t]
\footnotesize
\centering
\subcaptionbox{Gradient ascent}{\includegraphics[width=.3\textwidth]{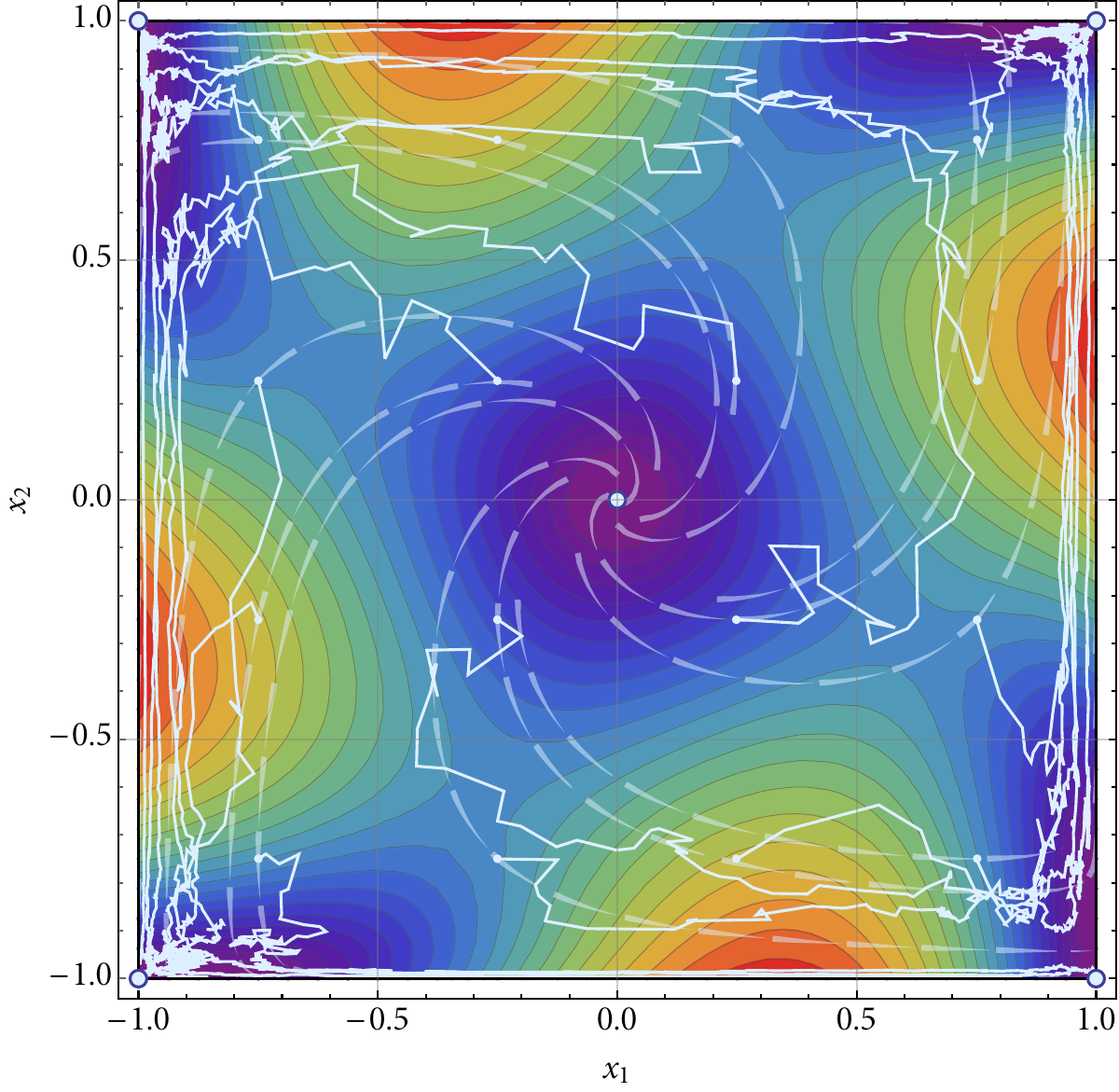}}
\hfill
\subcaptionbox{\Acl{EG}}{\includegraphics[width=.3\textwidth]{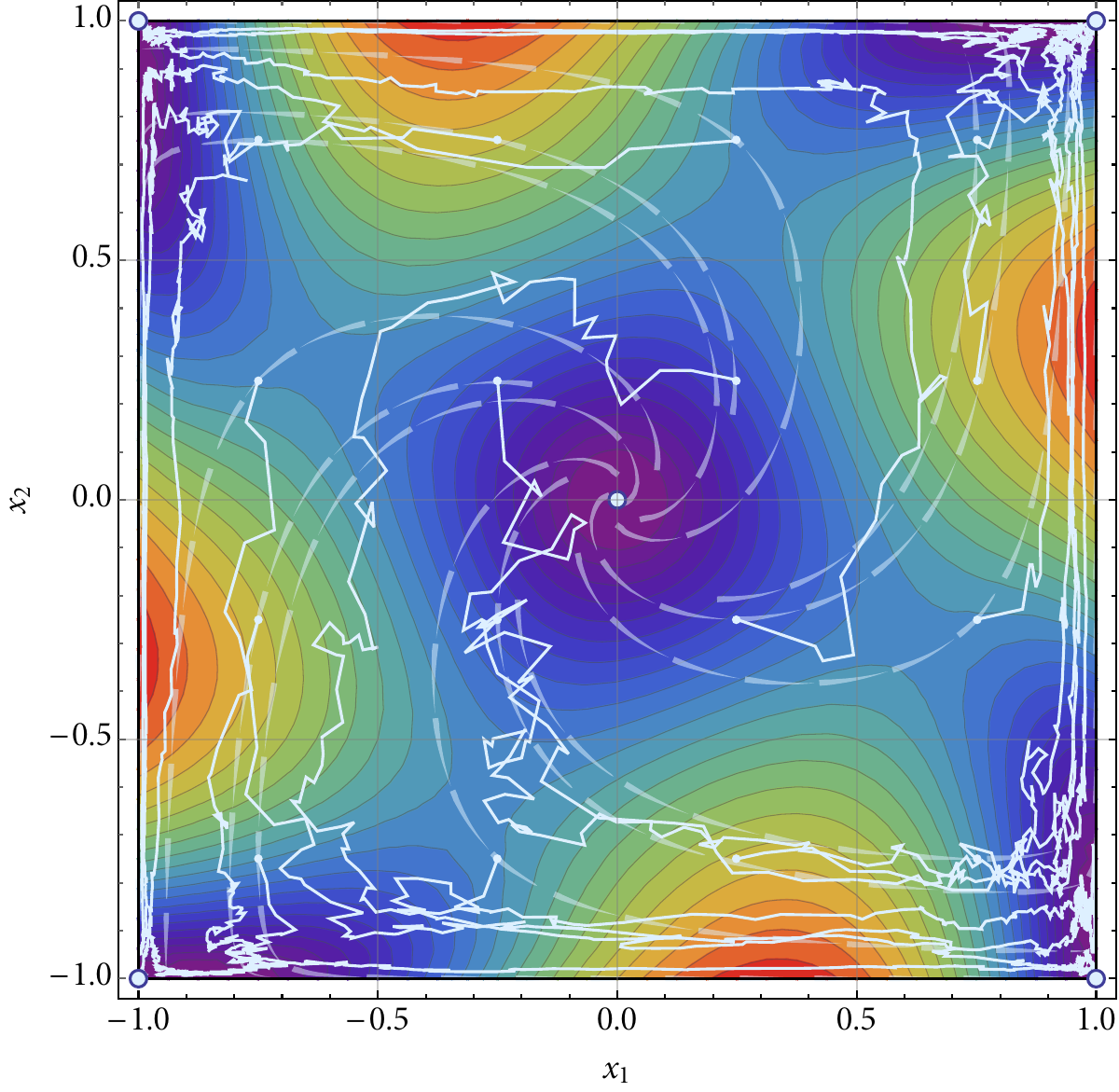}}
\hfill
\subcaptionbox{\Acl{OG}}{\includegraphics[width=.3\textwidth]{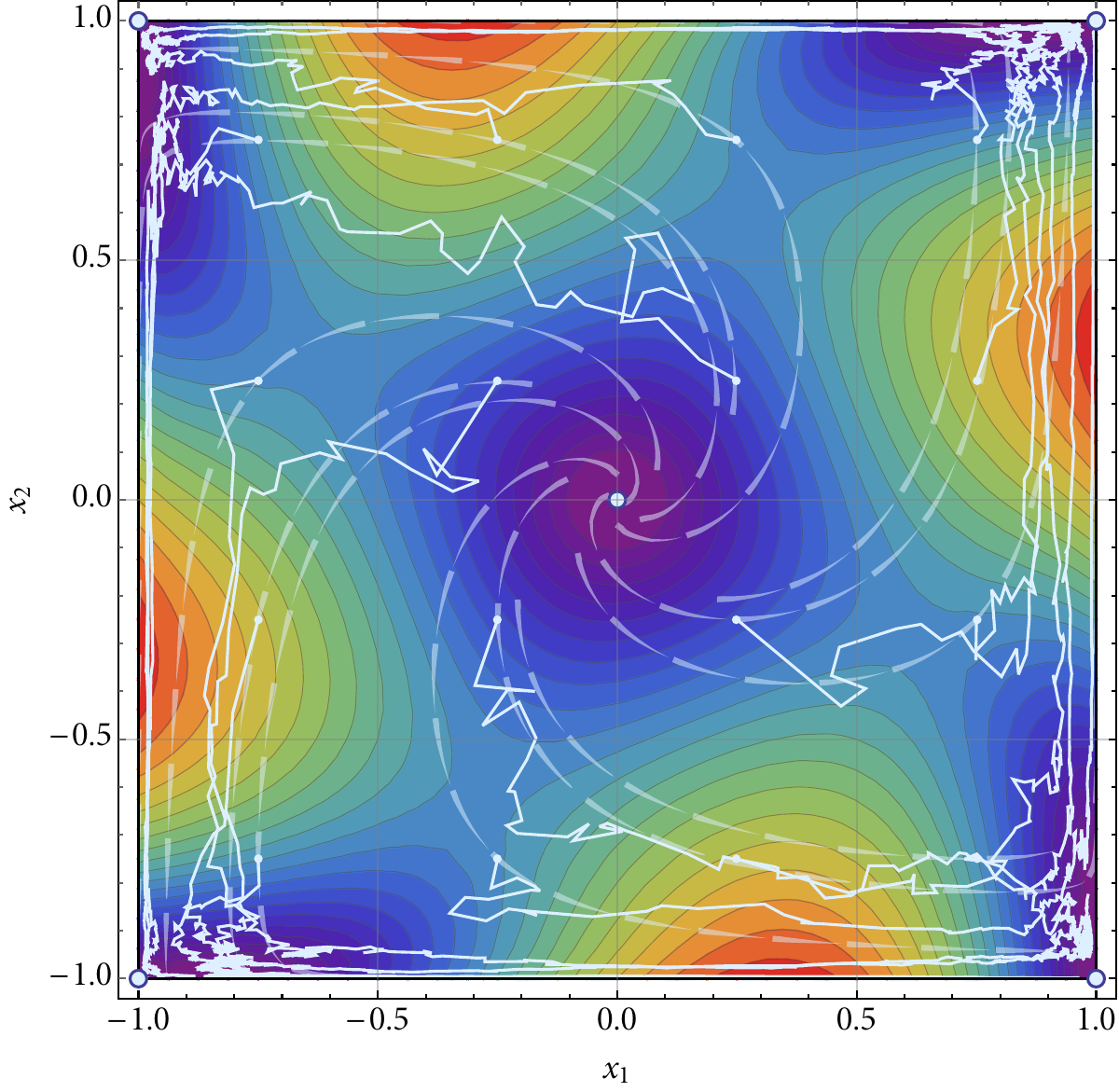}}
\caption{Learning in a $2$-player discoordination game.
All algorithms under study converge to a corner critical point which resists deviations by one of the players (but not the other).
The critical point at $(0,0)$ is unstable to deviations by \emph{both} players, and no trajectories converge there, even though it is the only interior \ac{ICT} of \eqref{eq:MD}.}
\label{fig:discord}
\vspace{-2ex}
\end{figure}


In \cref{sec:sharp}, we present an additional range of examples that cover such cases as (stochastic) linear programming, the set of undominated strategies of a game, etc.

%
%

\subsection{Main results, implications, and applications}
\label{sec:results}

We are now in a position to state our main results on the stable limit sets of \eqref{eq:method}.
To do so, we will assume for concreteness that \eqref{eq:method} is run with step-size and gradient signal sequences such that
\begin{equation}
\label{eq:schedule}
\curr[\step]
	= \step/\run^{\stepexp}
	\qquad
\curr[\bbound]
	= \bigoh(1/\run^{\biasexp})
	\qquad
	\text{and}
	\qquad
\curr[\sbound]
	= \bigoh(\run^{\noisexp})
\end{equation}
for some $\stepexp\in[0,1]$, $\biasexp>0$ and $\noisexp<1/2$.
Since the schedule \eqref{eq:schedule} involves $\curr[\bbound]$ and $\curr[\sbound]$ (which, depending on the algorithm, may be beyond the players' control), this requirement may seem unverifiable at first glance.
However, in view of \cref{prop:algorithms}, the exponents $\biasexp$ and $\noisexp$ can be directly expressed in terms of the parameters of the specific algorithm under study, so this is not an issue.
\smallskip

Without further ado, we have the following general result:

\begin{theorem}
\label{thm:attract}
Fix a tolerance level $\conf>0$, and let $\curr = \mirror(\curr[\dstate])$ be the sequence of play generated by \eqref{eq:method} with step-size and gradient signal sequences such that $\stepexp + \biasexp >1$ and $\stepexp - \noisexp > 1/2$ in \eqref{eq:schedule}.
If $\set$ admits a local energy function and $\step$ is sufficiently small, $\set$ is stochastically attracting;
specifically, there exists a neighborhood $\nhd$ of $\set$, independent of the tolerance level $\conf$, such that
\(
\probof{\text{$\curr$ converges to $\set$} \given \init\in\nhd}
	\geq 1-\conf.
\)
In addition, if the energy function on $\set$ is global, then, \acl{wp1}, $\curr$ converges to $\set$ from any initialization.
\end{theorem}

\begin{corollary}
\label{cor:attract}
Suppose that \crefrange{alg:SGA}{alg:EXP3} are run with step-size $\curr[\step] \propto 1/\run^{\stepexp}$, $\stepexp\in(1/2,1]$, and, where applicable, a sampling parameter $\curr[\mix] = \mix/\run^{\mixexp}$ such that $1-\stepexp < \mixexp < \stepexp-1/2$.
Then the conclusions of \cref{thm:attract} hold.
\end{corollary}

\begin{remark}
\label{rem:step}
In the baseline case $\curr[\bbound] = 0$, $\sup_{\run} \curr[\sdev] \eqdef \sdev < \infty$, the proof of \cref{thm:attract} shows that it suffices to take $\step = \bigoh(\min\braces{\hien / (\sdev\ebound), \sbound \sqrt{\hien/\esmooth}}) \cdot \sqrt{\conf}$.
\end{remark}

\Cref{thm:attract,cor:attract} are our main results concerning the stable limit sets of \eqref{eq:method} so, before discussing their proof, we present a series of corollaries and applications thereof.

\begin{corollary}
\label{cor:GVS}
Suppose that \crefrange{alg:SGA}{alg:EXP3} are run with parameters as in \cref{cor:attract}.
If $\eq$ is \acl{GVS}, then $\curr$ converges to $\eq$ \ac{wp1}.
\end{corollary}

\begin{corollary}
\label{cor:mono}
Suppose that \crefrange{alg:SGA}{alg:EXP3} are run with parameters as in \cref{cor:attract}.
If $\game$ is strictly monotone, then $\curr$ converges to the game's unique \acl{NE} \ac{wp1}.
\end{corollary}

The guarantees of \cref{cor:attract} are particularly important from an equilibrium convergence standpoint, because, as we mentioned in \cref{sec:prelims}, strictly monotone games account for a very wide range of applications \textendash\ socially concave games \citep{EMN09}, Cournot oligopolies \citep{MS96}, Kelly auctions \citep{KMT98}, etc.

We should also stress here that neither of the above results can be inferred by the \ac{ICT} convergence analysis of \cref{sec:general}.
In particular, if $\eq$ lies at the boundary of $\points$, it may fail to be accessible unless the dual process $\curr[\dstate]$ escapes to infinity, in which case \cref{thm:ICT} no longer applies.
This illustrates the flexibility of \cref{def:energy}, as it allows us to tackle at the same time both boundary \emph{and} interior solutions, in both bounded and unbounded domains.

To the best of our knowledge, the only comparable global convergence results in the literature for oracle-based methods concern the convergence of the standard \acl{MD} algorithm ($\curr[\bbound] = 0$, $\sup_{\run}\curr[\sdev]^{2} < \infty$) in strictly monotone games with compact domains \cite{MZ19}.
For payoff-based algorithms, the closest results we are aware of are by \citet{BLM18} and \citet{TK19,TK19-CDC} for a constrained variant of \eqref{eq:SPSA} in strictly monotone games with compact domains (the latter actually showing convergence in probability, but without requiring strict monotonicity).
\smallskip

Finally, in terms of local results, \cref{thm:attract} further yields the following corollaries:

\begin{corollary}
\label{cor:VS}
Suppose that \crefrange{alg:SGA}{alg:EXP3} are initialized and run as per \cref{cor:attract}.
If $\eq$ is \acl{VS} \textendash\ or, more narrowly, if it satisfies \eqref{eq:SOS} \textendash\ then $\curr$ converges locally to $\eq$ with arbitrarily high probability.
\end{corollary}

\begin{corollary}
\label{cor:strict}
Let $\eq$ be a strict \acl{NE} of a finite game.
If \cref{alg:EW,alg:EXP3} are initialized and run as per \cref{cor:attract}, $\curr$ converges locally to $\eq$ with arbitrarily high probability.
\end{corollary}


Of the above results, a special case of \cref{cor:VS} was proven in \cite{MZ19} for unbiased signal sequences with finite unconditional variance (\ie $\curr[\bbound] = 0$ and $\sup_{\run} \exof{\dnorm{\curr[\noise]}^{2}} < \infty$);
at the time of writing, this seems to be the closest antecedent of our results in the literature.
In particular, the convergence of \cref{alg:EG,alg:OG,alg:DGA} to \acl{VS} states and \ac{LNE} satisfying \eqref{eq:SOS} seems to be new.

Importantly, points satisfying \eqref{eq:SOS} are the game-theoretic analogue of minimizers with a positive-definite Hessian in non-convex minimization problems \cite{RBS16}.
In this regard, \cref{cor:VS} is particularly important as it shows that such equilibria are attracting under the entire class of algorithms under study.
Likewise, \cref{cor:strict} is a key result because, generically
\textendash\ \ie except on a set of games which is meager in the sense of Baire \textendash\
pure \aclp{NE} in finite games are always strict.
Thus, coupled with the inherent instability of mixed equilibria in finite games \cite{FVGL+20}, \cref{cor:strict} goes a long way toward establishing a learning analogue of the ``folk theorem'' of evolutionary game theory which states that a \acl{NE} is stable and attracting if and only if it is strict \cite{HS03}.

\subsection{Technical proofs}
\label{sec:proofs-attract}

We conclude this section with the proof of \cref{thm:attract}.
The main ingredient of our analysis is a ``template inequality'' for \eqref{eq:method} when the set under study admits an energy function (local or global).

To state it, note first that if $\energy$ is an energy function for $\set$ under \eqref{eq:MD}, there exists some $\emax>\ground$ (possibly equal to $\sup\energy$) such that
the sublevel set
\begin{equation}
\label{eq:dbasin}
\dbasin
	= \setdef{\dpoint\in\dspace}{\energy(\dpoint) \leq \emax}
\end{equation}
is forward invariant under \eqref{eq:MD}
and
$\sup\setdef{\dot\energy(\dpoint)}{\emax \geq \energy(\dpoint) > \loen} < 0$ for all $\loen\in(\ground,\emax)$.
Moreover, by assumption, there exist positive constants $\esmooth,\ebound > 0$ such that
$\norm{\nabla\energy(\dpoint)} \leq \ebound$
and
\begin{equation}
\label{eq:dsmooth}
\energy(\dpointalt)
	\leq \energy(\dpoint)
		+ \braket{\nabla\energy(\dpoint)}{\dpointalt - \dpoint}
		+ \tfrac{1}{2} \esmooth \dnorm{\dpointalt - \dpoint}^{2}
\end{equation}
for all $\dpoint,\dpointalt\in\dspace$.
With all this in hand, we have the following template inequality:

\begin{lemma}
\label{lem:template}
Let $\curr[\energy] \defeq \energy(\curr[\dstate])$.
Then, for all $\run = \running$, we have
\begin{equation}
\label{eq:template}
\next[\energy]
	\leq \curr[\energy]
		+ \curr[\step] \braket{\vecfield(\curr)}{\nabla\energy(\curr[\dstate])}
		+ \curr[\step] \curr[\snoise]
		+ \curr[\step] \curr[\sbias]
		+ \curr[\step]^{2} \curr[\second]^{2},
\end{equation}
where the error terms $\curr[\snoise]$, $\curr[\sbias]$, and $\curr[\second]$ are given by
\begin{equation}
\label{eq:errors}
\curr[\snoise]
	= \braket{\curr[\noise]}{\nabla\energy(\curr[\dstate])},
	\quad
\curr[\sbias]
	= \ebound \dnorm{\curr[\bias]}
	\quad
	\text{and}
	\quad
\curr[\second]^{2}
	= \tfrac{1}{2} \esmooth \dnorm{\curr[\signal]}^{2}.
\end{equation}
\end{lemma}

\begin{proof}
Simply unroll \eqref{eq:dsmooth} after substituting
$\dpoint \gets \curr[\dstate]$
and
$\dpointalt \gets \next[\dstate] = \curr[\dstate] + \curr[\step] \curr[\signal]$ with $\curr[\signal]$ as in \eqref{eq:signal}.
\end{proof}

Now, by the definition of $\energy$, we have $\dot\energy(\dpoint) = \braket{\vecfield(\mirror(\dpoint))}{\nabla\energy(\dpoint)} < 0$
whenever $\dpoint \in \dbasin \setminus \mirror^{-1}(\set)$.
Hence, for $\curr\in\mirror(\dbasin)$, \eqref{eq:template} becomes
\begin{equation}
\label{eq:template-err}
\next[\energy]
	\leq \curr[\energy]
		+ \curr[\step] \curr[\snoise]
		+ \curr[\step] \curr[\sbias]
		+ \curr[\step]^{2} \curr[\second]^{2}.
\end{equation}
Of course, each of these error terms can be positive, so $\curr[\energy]$ may fail to be decreasing, even when $\curr\in\mirror(\dbasin)$.
On that account, it will be convenient to introduce the error processes
\begin{equation}
\label{eq:errors-agg}
\txs
\curr[\aggnoise]
	= \sum_{\runalt=\start}^{\run} \iter[\step]\iter[\snoise]
	\qquad
\curr[\aggbias]
	= \sum_{\runalt=\start}^{\run} \iter[\step]\iter[\sbias]
	\qquad
	\text{and}
	\qquad
\curr[\aggsecond]
	= \sum_{\runalt=\start}^{\run} \iter[\step]^{2}\iter[\second]^{2}
\end{equation}
which measure directly the aggregate effect of each error term in \eqref{eq:template}.
As it turns out, under \eqref{eq:sum}, these errors can be compensated by the negative drift of \eqref{eq:template}, leading to the following global result:

\begin{proposition}
\label{prop:attract-global}
Suppose that $\set$ admits a global energy function, and let $\curr = \mirror(\curr[\dstate])$ be the sequence of play generated by \eqref{eq:method}.
If \eqref{eq:sum} holds, then, \acl{wp1}, $\curr$ converges to $\set$.
\end{proposition}

To streamline our discussion, before proving \cref{prop:attract-global}, we present a similar convergence result for sets that only admit \emph{local} energy functions.
In this case, even if the algorithm begins play close to $\set$, a single ``bad'' realization of the noise could force the process to exit the basin of attraction of $\set$, possibly never to return.
With a fair degree of hindsight, we will control the probability with which this ``bad event'' occurs via the stability requirement
\begin{equation}
\label{eq:stab}
\tag{Stab}
\probof{\text{$\curr[Z] > \hien/4$ for some $\run$}}
	< \conf/3
\end{equation}
where
$\conf > 0$ is the target tolerance level,
and
$\curr[Z] \gets \curr[\aggnoise]$, $\curr[\aggbias]$ or $\curr[\aggsecond]$,
depending on the error term that we wish to control.
Modulo this requirement, we obtain the following local analogue of \cref{prop:attract-global}:

\begin{proposition}
\label{prop:attract-local}
Suppose that $\set$ admits a local energy function,
and let $\curr = \mirror(\curr[\dstate])$ be the sequence of play generated by \eqref{eq:method}.
Assume further that the algorithm begins play at a neighborhood $\nhd$ of $\set$ such that $\energy(\init[\dstate]) \leq \hien/4$.
If \eqref{eq:sum} and \eqref{eq:stab} hold, then
\begin{equation}
\txs
\probof{\text{$\energy(\curr[\dstate]) < \hien$ for all $\run$ and $\lim_{\run\to\infty} \dist(\curr,\set) = 0$}}
	\geq 1-\conf.
\end{equation}
\end{proposition}

Of course, \cref{prop:attract-global,prop:attract-local} can be difficult to employ in practice because of their reliance on the conditions \eqref{eq:sum} and \eqref{eq:stab}.
Because of this,
we defer the proof of \cref{prop:attract-global,prop:attract-local} to the end of this section,
and
we proceed below to complete the proof of \cref{thm:attract} by showing that \eqref{eq:sum} and \eqref{eq:stab} both hold under the stated step-size and gradient signal requirements.

\begin{proof}[Proof of \cref{thm:attract}]
We begin by noting that \eqref{eq:sum} holds trivially under the stated conditions for $\curr[\step] \propto 1/\run^{\stepexp}$, $\curr[\bbound] = \bigoh(1/\run^{\biasexp})$ and $\curr[\sbound] = \bigoh(\run^{\noisexp})$.
As a result, the first part of the theorem follows immediately from \cref{prop:attract-global}.

Likewise, for the second part, it will suffice to establish the stability condition \eqref{eq:stab}.
To that end,
proceeding term-by-term, we have:

\begin{enumerate}
\item
Since $\curr[\aggnoise]$ is a martingale, Kolmogorov's inequality \citep[Corollary 2.1]{HH80} gives
\begin{align}
\probof*{\max_{\start\leq\runalt\leq\run} \iter[\aggnoise] \geq \hien/4}
	&\leq \probof*{\max_{1\leq\runalt\leq\run} \abs{\iter[\aggnoise]} \geq \hien/4}
	\notag\\
	&\leq \frac{16\exof{\curr[\aggnoise]^{2}}}{\hien^{2}}
	= \frac{16\exof[\big]{\parens[\big]{\sum_{\runalt=\start}^{\run} \iter[\step]\iter[\snoise]}^{2}}}{\hien^{2}}
	\notag\\
	&\leq \frac{16\ebound^{2} \sum_{\runalt=\start}^{\run} \iter[\step]^{2} \iter[\sdev]^{2}}{\hien^{2}}
	\eqdef \Const_{\aggnoise}
\end{align}
where we used the variance bound
\begin{equation}
\exof{\iter[\snoise]^{2}}
	= \exof{ \exof{ \abs{\braket{\iter[\noise]}{\nabla\energy(\iter[\dstate])}}^{2} \given \iter[\filter] } }
	\leq \ebound^{2} \iter[\sdev]^{2}
\end{equation}
and the fact that $\exof{\snoise_{\runalt}\snoise_{\runaltalt}} = \exof{\snoise_{\runalt}\snoise_{\runaltalt} \given \filter_{\runalt\vee\runaltalt}} = 0$ whenever $\runalt \neq \runaltalt$.
Thus, given that $\{ \curr \geq \hien/4 \; \text{for some $\run$}\} = \union_{\run} \{ \curr[\maxnoise] \geq \hien/4 \}$ is a union of nested events, we conclude that \eqref{eq:stab} holds for $\curr[Z] \gets \curr[\aggnoise]$ whenever
$\Const_{\aggnoise} \leq \conf/3$.

\item
For the second term, we have $\curr[\aggbias] \leq \ebound \sum_{\runalt=\start}^{\infty} \iter[\step] \iter[\bbound]$ for all $\run$ \acl{wp1}, so \eqref{eq:stab} holds for $\curr[Z] \gets \curr[\aggbias]$ as long as $\Const_{\aggbias} \defeq (4\ebound/\hien) \sum_{\run} \curr[\step] \curr[\bbound] \leq 1$.

\item
Finally, for the last term, Markov's inequality yields
\begin{equation}
\probof{\curr[\aggsecond] \geq \hien/4}
	\leq \frac{4\exof{\curr[\aggsecond]}}{\hien}
	= \frac{2\esmooth \sum_{\runalt=\start}^{\run} \iter[\step]^{2} \iter[\sbound]^{2}}{\hien}
	\eqdef \Const_{\aggsecond}.
\end{equation}
We thus see that the event $\{ \curr[\aggsecond] \geq \hien/4 \; \text{for some $\run$}\} = \union_{\run} \{ \curr[\aggsecond] \geq \hien/4 \}$ occurs with probability no more than
$\Const_{\aggsecond}$,
which implies in turn that the requirement \eqref{eq:stab} for $\curr[Z] \gets \curr[\aggsecond]$ holds whenever $\Const_{\aggsecond} \leq \conf/3$.
\end{enumerate}
\noindent
Since $\Const_{\aggnoise}$, $\Const_{\aggbias}$ and $\Const_{\aggsecond}$ are all $\bigoh(\step^{2})$, we can choose $\step$ sufficiently small so that $\Const_{\aggnoise} \leq \conf/3$, $\Const_{\aggbias} \leq 1$ and $\Const_{\aggsecond} \leq \conf/3$.
In this case, \eqref{eq:stab} holds by construction, and our claim follows from \cref{prop:attract-local}.
\end{proof}

We are thus left to prove \cref{prop:attract-global,prop:attract-local}.
To that end, we begin with a technical lemma showing that the aggregate error processes $\curr[\aggnoise],\curr[\aggbias]$ and $\curr[\aggsecond]$ of \eqref{eq:errors-agg} are subleading relative to the long-run drift of \eqref{eq:template}.

\begin{lemma}
\label{lem:sub}
Under \eqref{eq:sum}, the aggregate error processes of \eqref{eq:errors-agg} are sublinear in $\curr[\efftime]$, \ie we have
\begin{equation}
\label{eq:sub}
\tag{Sub}
\curr[Z] / \curr[\efftime]
	\to 0
	\quad
	\text{\acl{wp1}},
\end{equation}
where $\curr[Z] \gets \curr[\aggnoise]$, $\curr[\aggbias]$ or $\curr[\aggsecond]$, depending on the error term under study.
\end{lemma}

\begin{proof}
We treat each case $\curr[Z] \gets \curr[\aggnoise]$, $\curr[\aggbias]$ or $\curr[\aggsecond]$ separately.
\begin{enumerate}
\item
For $\curr[\aggnoise]$, \eqref{eq:sum} readily gives
\begin{equation}
\sum_{\run=\start}^{\infty} \exof{\curr[\step]^{2} \curr[\snoise]^{2} \given \curr[\filter]}
	\leq \sum_{\run=\start}^{\infty} \curr[\step]^{2} \exof{ \norm{\nabla\energy(\curr[\dstate])}^{2} \dnorm{\curr[\noise]}^{2} \given \curr[\filter]}
	\leq \ebound^{2} \sum_{\run=\start}^{\infty} \curr[\step]^{2} \curr[\sdev]^{2}
	< \infty.
\end{equation}
Thus, by the strong law of large numbers for martingale difference sequences \citep[Theorem 2.18]{HH80}, we conclude that $\curr[\aggnoise] / \curr[\efftime]$ converges to $0$ \acl{wp1}.

\item
For $\curr[\aggbias]$, the conclusion is immediate by the fact that $\sum_{\run} \curr[\step] \curr[\bbound] < \infty$ under \eqref{eq:sum}.

\item
Finally, for the submartingale term $\curr[\aggsecond]$, we have
\begin{align}
\exof{\curr[\aggsecond]}
	&= \sum_{\runalt=\start}^{\run} \iter[\step]^{2} \exof{\iter[\second]^{2}}
	\leq \frac{\esmooth}{2} \sum_{\runalt=\start}^{\run} \iter[\step]^{2}
			\exof{\dnorm{\iter[\signal]}^{2}},
	\leq \frac{\esmooth}{2} \sum_{\runalt=\start}^{\run} \iter[\step]^{2} \iter[\sbound]^{2},
\end{align}
so, by \eqref{eq:sum}, it follows that $\curr[\aggsecond]$ is bounded in $L^{1}$.
Therefore, by Doob's submartingale convergence theorem \citep[Theorem~2.5]{HH80}, we further deduce that $\curr[\aggsecond]$ converges \as to some (finite) random variable $\aggsecond_{\infty}$, implying in turn that $\curr[\aggsecond] / \curr[\efftime] \to 0$ \acl{wp1}.
\qedhere
\end{enumerate}
\end{proof}

Moving forward, we present two lemmas that will allow us to deduce the convergence of the energy iterates $\curr[\energy] \defeq \energy(\curr[\dstate])$ modulo the occurrence of the favorable event
\begin{equation}
\label{eq:good}
\good
	= \braces{ \curr[\dstate]\in\dbasin \; \text{for all $\run$} }
\end{equation}
where
$\dbasin = \setdef{\dpoint\in\dspace}{\energy(\dpoint) \leq \emax}$ is defined as in \eqref{eq:dbasin}.
In particular, we have the following results:

\begin{lemma}
\label{lem:subseq}
Suppose that $\probof{\good} > 0$.
If \eqref{eq:sub} holds, then $\probof{\liminf_{\run\to\infty}\curr[\energy] = \ground \given \good} = 1$.
\end{lemma}

\begin{lemma}
\label{lem:finitelim}
Suppose that $\probof{\good} > 0$.
If \eqref{eq:sum} holds,
there exists some finite random variable $\energy_{\infty}$ such that $\probof{\lim_{\run\to\infty} \curr[\energy] = \energy_{\infty} \given \good} = 1$.
\end{lemma}

\begin{proposition}
\label{prop:Fejer}
Suppose that $\energy$ is a local energy function for $\set$.
If $\probof{\good} > 0$ and \eqref{eq:sum} holds,
then $\probof{\text{$\curr$ converges to $\set$} \given \good} = 1$.
\end{proposition}

\begin{proof}[Proof of \cref{lem:subseq}]
Since $\probof{\good} > 0$, it suffices to show that the hitting time $\stoptime_{\elvl} = \inf\setdef{\run\in\N}{\curr[\energy] \leq \elvl}$ is finite \acl{wp1} on $\good$ for all sufficiently small $\elvl > \ground$.
More precisely, building on an argument of \citet{DMSV23}, we will show that the event $\bad_{\elvl} = \good \cap \{\stoptime_{\elvl} = \infty\}$ has $\probof{\bad_{\elvl}} = 0$ whenever $\ground < \elvl \leq \emax$:
indeed, if this is the case and $\iter[\elvl] \in (\ground,\emax)$, $\runalt=\running$, is a sequence converging monotonically to $\ground$, we will have $\probof{\bad_{\iter[\elvl]}} = 0$ for all $\runalt \in \N$.
Thus, with only a countable number of $\bad_{\iter[\elvl]}$ in play, we will have
\begin{align}
\probof{\liminf\nolimits_{\run\to\infty} \curr[\energy] = \ground \given \good}
	&= \probof{\stoptime_{\iter[\elvl]} < \infty\;\textrm{for all $\runalt$} \given \good}
	\notag\\
	&\txs
	= \probof*{\intersect\nolimits_{\runalt=\start}^{\infty} \{\stoptime_{\iter[\elvl]} < \infty\} \given \good}
	= 1 - \probof*{\union\nolimits_{\runalt=\start}^{\infty} \{\stoptime_{\iter[\elvl]} < \infty\} \given \good}
	\notag\\
	&= 1 - \frac
		{\probof*{\good \cap \parens*{\union\nolimits_{\runalt=\start}^{\infty} \{\stoptime_{\iter[\elvl]} < \infty\}}}}
		{\probof{\good}}
	= 1 - \frac{\probof*{\union\nolimits_{\runalt=\start}^{\infty} \bad_{\iter[\elvl]}}}{\probof{\good}}
	= 1,
\end{align}
as per our original assertion.

Now, to establish our claim for $\bad_{\elvl}$, assume to the contrary that $\probof{\bad_{\elvl}} > 0$ for some sufficiently small $\elvl>\ground$, and let $\const_{\elvl} = -\sup\setdef{\dot\energy(\dpoint)}{\elvl \leq \energy(\dpoint) \leq \emax}$, so $\const_{\elvl} > 0$ by \cref{def:energy}.
Then, by telescoping \eqref{eq:template}, we get
\begin{align}
\label{eq:energy-escape}
\next[\energy]
	&\leq \init[\energy]
		+ \sum_{\runalt=\start}^{\run} \iter[\step] \dot\energy(\iter[\dstate])
			+ \sum_{\runalt=\start}^{\run} \iter[\step] \iter[\snoise]
			+ \sum_{\runalt=\start}^{\run} \iter[\step] \iter[\sbias]
			+ \sum_{\runalt=\start}^{\run} \iter[\step] \iter[\second]^{2}
	\notag\\
	&\leq \init[\energy]
		- \bracks*{
			\const_{\elvl}
			- \frac{\curr[\aggnoise] + \curr[\aggbias] + \curr[\aggsecond]}{\curr[\efftime]}}
		\cdot \curr[\efftime]
	\qquad
	\text{for all $\run=\running$}
\end{align}
\acl{wp1} on $\bad_{\elvl}$.
Since $\probof{\bad_{\elvl}} > 0$ by assumption and $(\curr[\aggnoise] + \curr[\aggbias] + \curr[\aggsecond]) \big/ \curr[\efftime] \to 0$ \acl{wp1} by \eqref{eq:sub}, the above gives
$\probof{\lim_{\run\to\infty}\curr[\energy] = -\infty \given \bad_{\elvl}} = 1$.
However, with $\inf_{\run}\curr[\energy] \geq \elvl > \ground$ on $\bad_{\elvl}$ by construction, we get a contradiction, and our proof is complete.
\end{proof}

\begin{proof}[Proof of \cref{lem:finitelim}]
Consider the nested sequence of events
\begin{equation}
\label{eq:events}
\curr[\good]
	= \braces{ \dot\energy(\iter[\dstate]) \leq 0 \; \text{for all $\runalt = \running,\run$} }
\end{equation}
so $\good = \intersect_{\run=\start}^{\infty} \curr[\good]$.
Then, letting $\curr[\tilde\energy] = \one_{\curr[\good]} \curr[\energy]$, \cref{eq:template} readily gives
\begin{align}
\next[\tilde\energy]
	= \one_{\next[\good]} \next[\energy]
	&\leq \one_{\curr[\good]} \next[\energy]
	\notag\\
	&\leq \one_{\curr[\good]} \curr[\energy]
		+ \parens[\big]{
			\curr[\step] \dot\energy(\curr[\dstate])
				+ \curr[\step] \curr[\snoise]
				+ \curr[\step] \curr[\sbias]
				+ \curr[\step]^{2} \curr[\second]^{2}
			} 
			\one_{\curr[\good]}
	\notag\\
	&\leq \curr[\tilde\energy]
		+ \curr[\step] \one_{\curr[\good]} \curr[\snoise]
		+ \parens[\big]{
			\curr[\step] \curr[\sbias]
			+ \curr[\step]^{2} \curr[\second]^{2}
			} 
			\one_{\curr[\good]},
\end{align}
where we used the fact that $\dot\energy(\iter[\dstate]) = \braket{\vecfield(\iter)}{\nabla\energy(\iter[\dstate])} \leq 0$ for all $\runalt = \running,\run$ if $\curr[\good]$ occurs.
Since $\curr[\good]$ is $\curr[\filter]$-measurable, conditioning on $\curr[\filter]$ and taking expectations then yields
\begin{align}
\exof{\next[\tilde\energy] \given \curr[\filter]}
	&\leq \curr[\tilde\energy]
		+ \curr[\step] \one_{\curr[\good]} \exof{\curr[\snoise] \given \curr[\filter]}
		+ \one_{\curr[\good]} \exof{\curr[\step] \curr[\sbias] + \curr[\step]^{2} \curr[\second]^{2} \given \curr[\filter]}
	\notag\\
	&\leq \curr[\tilde\energy]
		+ \exof{\curr[\step] \curr[\sbias]  + \curr[\step]^{2} \curr[\second]^{2} \given \curr[\filter]}
	\notag\\
	&\leq \curr[\tilde\energy]
		+ \curr[\step] \ebound \curr[\bbound]
		+ \tfrac{1}{2} \esmooth \curr[\step]^{2} \curr[\sbound]^{2}.
\end{align}
Now, given that $\sum_{\run} \curr[\step] \curr[\bbound]$ and $\sum_{\run} \curr[\step]^{2} \curr[\sbound]^{2}$ are both finite by \eqref{eq:sum}, $\curr[\tilde\energy]$ is an almost supermartingale with summable increments, \ie $\sum_{\run} \bracks*{\exof{\next[\tilde\energy] \given \curr[\filter]} - \curr[\tilde\energy]} < \infty$ \ac{wp1}.
Therefore, by Gladyshev's lemma \citep[p.~49]{Pol87}, we conclude that $\curr[\tilde\energy]$ converges almost surely to some (finite) random variable.
Since $\probof{\good} > 0$ and $\one_{\curr[\good]} = 1$ for all $\run$ if and only if $\good$ occurs, we further deduce that
\(
\probof{\text{$\curr[\energy]$ converges} \given \good}
	= \probof{\text{$\curr[\tilde\energy]$ converges} \given \good}
	= 1,
\)
and our claim follows.
\end{proof}

\begin{proof}[Proof of \cref{prop:Fejer}]
By \cref{lem:sub}, \eqref{eq:sub} is satisfied whenever \eqref{eq:sum} is.
Thus, by a tandem application of \cref{lem:subseq,lem:finitelim}, we conclude that $\lim_{\run\to\infty} \curr[\energy] = \ground$ \acl{wp1} on $\good$.
Finally, since $\mirror(\dpoint)\to\set$ whenever $\energy(\dpoint)\to\ground$, our claim follows.
\end{proof}

We are now in a position to prove \cref{prop:attract-global,prop:attract-local}.

\begin{proof}[Proof of \cref{prop:attract-global}]
By the definition of a global attractor, we have $\emax = \sup\energy$, so $\probof{\good} = 1$.
Our claim is then an immediate consequence of \cref{prop:Fejer}.
\end{proof}

\begin{proof}[Proof of \cref{prop:attract-local}]
Suppose that $\energy(\init[\dstate]) \leq \emax/4$.
We then claim that the event $\good$ always occurs on the intersection of the events $\good_{\aggnoise}$, $\good_{\aggbias}$, and $\good_{\aggsecond}$, where $\good_{Z} = \{\curr[Z] \leq \emax/4 \; \text{for all $\run$}\}$.
Indeed, this being trivially the case for $\run=\start$, assume that $\iter[\dstate]\in\dbasin$ for all $\runalt=\running,\run$ for some $\run \geq \start$.
Then, telescoping \eqref{eq:template} yields
\begin{equation}
\next[\energy]
	\leq \init[\energy]
		+ \sum_{\runalt=\start}^{\run} \iter[\step] \dot\energy(\iter[\dstate])
		+ \curr[\aggnoise]
		+ \curr[\aggbias]
		+ \curr[\aggsecond]
	\leq \emax/4
		+ 0
		+ \emax/4
		+ \emax/4
		+ \emax/4
	= \emax
\end{equation}
by the inductive hypothesis and our other assumptions.
This shows that $\next[\energy] \in \dbasin$, so the induction argument is complete, and we conclude that $\good \supseteq \good_{\aggnoise} \cap \good_{\aggbias} \cap \good_{\aggsecond}$.
Now, by \eqref{eq:stab}, we have $\probof{\good_{\aggnoise}} \leq \conf/3$ and likewise for the rest, so we get
\begin{equation}
\probof{\good}
	\geq \probof{\good_{\aggnoise} \cap \good_{\aggbias} \cap \good_{\aggsecond}} = 1- \probof{\good_{\aggnoise} \cup \good_{\aggbias} \cup \good_{\aggsecond}} \geq 1 - \probof{\good_{\aggnoise}} - \probof{\good_{\aggbias}} - \probof{\good_{\aggsecond}} \geq 1 - \conf.
\end{equation}
Our claim then follows directly from \cref{prop:Fejer}.
\end{proof}

%% file: Body/Sharp.tex

\subsection{The notion of coherence: definition and examples}
\label{sec:coherent-intro}

In this section, we will show that the analysis of the previous section can be strengthened considerably under a ``structural alignment'' notion, which we dub \emph{coherence}.
We begin with the definition and a series of motivating examples.

\begin{definition}
\label{def:coherent}
A nonempty compact subset $\set$ of $\points$ will be called \emph{coherent} if it admits a \textpar{finite} set of \emph{deviation directions} $\devs = \braces{\dev_{1},\dotsc,\dev_{\nDevs}} \subseteq \pspace$ such that
\begin{subequations}
\label{eq:coherent}
\setlength{\abovedisplayskip}{3pt}
\setlength{\belowdisplayskip}{3pt}
\begin{flalign}
\label{eq:polar}
\quad
a)
	&\txs
	\quad
	\braket{\vecfield(\point)}{\dev}
		< 0
		\quad
		\text{for all $\point\in\set$ and all $\dev\in\devs$}.
	&
	\\
\label{eq:target}
b)
	&\txs
	\quad
	\mirror(\dpoint)
		\to \set
		\quad
		\text{whenever $\max_{\dev\in\devs} \braket{\dpoint}{\dev} \to -\infty$}.
	&
\end{flalign}
\end{subequations}
In particular, if \eqref{eq:polar} holds for all $\point\in\points$, we will say that $\set$ is \emph{globally coherent};
and if we want to stress that $\set$ is coherent but not globally so, we will say that $\set$ is \emph{locally coherent}.
\end{definition}

The motivation behind \cref{def:coherent} is as follows.
First, the geometric condition \eqref{eq:polar} posits that any deviation from $\set$ along a vector $\dev\in\devs$ is actively disincentivized by the players' individual gradient field $\vecfield$ so, in a certain sense, $\vecfield$ points locally ``toward'' $\set$.
The second condition is game-independent and asks that the elements of $\devs$ are sufficient to identify $\set$ by acting as primal-dual ``support vectors'' for $\set$ under $\mirror$.
The terminology ``coherence'' has been chosen precisely to indicate that these two properties dovetail to create a favorable convergence landscape under \eqref{eq:method}.
\smallskip

To illustrate this notion, we proceed below with a series of examples.
The first two concern finite games;
the last two concern continuous ones.

\begin{example}
[Strict equilibria in finite games]
\label{ex:strict}
Recall that a strict \acl{NE} of a finite game $\fingame = \fingamefull$ is a strategy profile $\eq$ such that \eqref{eq:NE} holds as a strict inequality for all $\strat\neq\eq$.
An immediate consequence of this definition is that
\begin{enumerate*}
[\itshape a\upshape)]
\item
$\eq$ is \emph{pure}, \ie it is supported on a single pure strategy profile $\pureq\in\pures$;
and that
\item
unilateral deviations from $\pureq$ lead to \emph{strictly} inferior payoffs, \ie $\pay_{\pure}(\pureq_{\play};\pureq_{-\play}) > \pay_{\play}(\pure_{\play};\pureq_{-\play})$ for all $\pure_{\play} \in \pures_{\play}\exclude{\pureq_{\play}}$, $\play\in\players$.
\end{enumerate*}

With this in mind, consider the set of unilateral deviations
\begin{equation}
\label{eq:devs-strict}
\txs
\devs
	= \setdef
	{\bvec_{\play\pure_{\play}} - \bvec_{\play\pureq_{\play}}}
	{\pure_{\play}\in\pures_{\play}\exclude{\pureq_{\play}}, \play\in\players}.
\end{equation}
Since $\braket{\vecfield(\eq)}{\bvec_{\play\pure_{\play}} - \bvec_{\play\pureq_{\play}}} = \pay_{\play}(\pure_{\play};\pureq_{-\play}) - \pay_{\play}(\pureq_{\play};\pureq_{-\play}) < 0$ for all $\pure_{\play}\in\pures_{\play}\exclude{\pureq_{\play}}$, $\play\in\players$, condition \eqref{eq:polar} is satisfied.
\Cref{lem:support} further shows that $\mirror_{\play\pure_{\play}}(\score) \to 0$ whenever $\score_{\play\pure_{\play}} - \score_{\play\pureq_{\play}} \to -\infty$, so the requirement $\mirror(\score) \to \eq$ of \eqref{eq:target} is also satisfied.
In other words, \emph{strict equilibria are coherent}.
\endenv
\end{example}

\begin{example}
[Undominated strategies]
\label{ex:doms}
Recall that a pure strategy $\pure_{\play}\in\pures_{\play}$ is \emph{dominated} by $\purealt_{\play} \in \pures_{\play}$ if $\pay_{\play}(\pure_{\play};\strat_{-\play}) < \pay_{\play}(\purealt_{\play};\strat_{-\play})$ for all $\strat\in\strats$.
We then say that $\pure_{\play}$ is \emph{eliminated} in a mixed strategy profile $\strat\in\strats$ if $\pure_{\play}$ is not supported in $\strat_{\play}$, \ie if $\strat_{\play\pure_{\play}} = 0$.
A fundamental requirement for game-theoretic learning is that dominated strategies become extinct over time, \ie that the trajectory of play converges to the set $\eqs$ of action profiles that eliminate all dominated strategies.%
\footnote{The case of mixed strategies dominated by mixed strategies requires heavier notation, so we do not treat it.}

This set is globally coherent.
To see this, consider the set of dominating deviations
\begin{equation}
\txs
\devs
	= \setdef
	{\bvec_{\play\pure_{\play}} - \bvec_{\play\purealt_{\play}}}
	{\text{$\pure_{\play}$ is dominated by $\purealt_{\play}$}}.
\end{equation}
By definition, $\braket{\payv(\strat)}{\bvec_{\play\pure_{\play}} - \bvec_{\play\purealt_{\play}}} = \pay_{\play}(\pure_{\play};\strat_{-\play}) - \pay_{\play}(\purealt_{\play};\strat_{-\play}) < 0$ for all $\strat\in\strats$, so \eqref{eq:polar} holds globally.
Moreover, for any finite game, $\eqs$ is a face of $\strats$ \cite{SZ92} and hence compact.
Finally, \cref{lem:support} shows that $\mirror_{\play\pure_{\play}}(\score)\to0$ if $\score_{\play\pure_{\play}} - \score_{\play\purealt_{\play}} \to -\infty$, so the requirement $\mirror(\score)\to\eqs$ of \eqref{eq:target} is also satisfied, and we conclude that \emph{the set of undominated strategies is globally coherent}.
\endenv
\end{example}

\begin{example}
[Sharp equilibria in concave games]
\label{ex:sharp}
Following \citet{Pol87}, a \acl{NE} of a concave game is \emph{sharp} if the stationarity condition \eqref{eq:FOS} holds as a strict inequality for all $\point\neq\eq$, \ie
\begin{equation}
\label{eq:sharp}
\tag{Sharp}
\braket{\vecfield(\eq)}{\point - \eq}
	< 0
	\quad
	\text{for all $\point\neq\eq$}.
\end{equation}
Examples of sharp equilibria include
deterministic Nash policies in generic stochastic games \citep{ZRL21},
power control and resource allocation games \citep{SFPP10}, etc.

Geometrically, sharp equilibria can be characterized by the condition that $\vecfield(\eq)$ lies in the (topological) interior of the polar cone $\pcone(\eq)$ to $\points$ at $\eq$.
This means in particular that there exists a polyhedral cone $\cone$ that is spanned by a finite set of vectors $\devs = \braces{\dev_{1},\dotsc,\dev_{\nDevs}} \subseteq \pspace$ such that
\begin{enumerate*}
[\itshape a\upshape)]
\item
the tangent cone $\tcone(\eq)$ to $\points$ at $\eq$ is contained in the interior of $\cone$;
and
\item
$\braket{\vecfield(\eq)}{\dev} < 0$ for all $\dev\in\devs$.
\end{enumerate*}
\Cref{lem:target} in \cref{app:mirror} shows that $\mirror(\dpoint)\to\set$ if $\max_{\dev\in\devs}\braket{\dpoint}{\dev} \to -\infty$, so we conclude that
\emph{sharp equilibria are coherent}.
\endenv
\end{example}

\begin{example}
[Stochastic linear programming]
\label{ex:linear}
To borrow an example from optimization (viewed here as a single-player game), let $\points$ be a convex polytope and consider the stochastic linear program
\begin{equation}
\label{eq:SLP}
\tag{SLP}
\begin{aligned}
\textrm{maximize}
	&\quad
	\pay(\point)
		= \ex_{\seed}\bracks{\braket{\orcl(\seed)}{\point}}
	\\
\textrm{subject to}
	&\quad
	\point \in \points
\end{aligned}
\end{equation}
where $\orcl(\seed)$ is a random payoff vector drawn from some complete probability space $(\seeds,\prob_{\seed})$.
By linearity, the set of solutions $\sols = \argmax\pay$ of \eqref{eq:SLP} is a face of $\points$;
moreover, if we let $\payv = \ex_{\seed}\bracks{\orcl(\seed)} = \nabla\pay(\point)$, we have $\braket{\payv}{\point - \sol} < 0$ whenever $\sol\in\sols$ and $\point\in\points\setminus\sols$.
Finally, since $\points$ is a convex polytope, there exists a finite set of vectors $\devs = \braces{\dev_{1},\dotsc,\dev_{\nDevs}}$ such that
\begin{enumerate*}
[\itshape a\upshape)]
\item
$\sol+\dev \in \points\setminus\sols$ for all $\sol\in\sols$, $\dev\in\devs$;
and
\item
every point $\point\in\points\setminus\sols$ can be decomposed as $\point = \sol + \coef\dev$ for some $\sol\in\sols$, $\dev\in\devs$ and $\coef>0$.
\end{enumerate*}
\Cref{lem:target} in \cref{app:mirror} shows that $\mirror(\dpoint)\to\sols$ whenever $\braket{\dpoint}{\dev}\to-\infty$ for all $\dev\in\devs$, so \eqref{eq:target} is satisfied and we conclude that
\emph{the solution set $\sols$ of \eqref{eq:SLP} is globally coherent}.
\endenv
\end{example}

The above examples illustrate that the notion of coherence underlies a diverse range of game-theoretic settings and problems.
In light of this, we devote the rest of this section to analyzing the convergence properties of coherent sets under \eqref{eq:method}.

\subsection{Convergence analysis and results}
\label{sec:coherent-results}

The first thing to note is that, if $\set$ is coherent, it admits the local energy function
\begin{equation}
\label{eq:template-coherent}
\energy(\dpoint)
	= \log\parens*{1+\sum\nolimits_{\dev\in\devs} \exp{\braket{\dpoint}{\dev}}}
\end{equation}
Indeed, if $\energy(\dpoint) \to \inf\energy = 0$, we must have $\braket{\dpoint}{\dev} \to -\infty$ for all $\dev\in\devs$, and hence $\mirror(\dpoint) \to \set$ by \cref{def:coherent}.
Moreover, for all $\dpoint$ such that $\point = \mirror(\dpoint)$ is sufficiently close to $\set$, we have $\nabla\energy(\dpoint) = \sum_{\dev\in\devs} \braket{\vecfield(\point)}{\dev} e^{\braket{\dpoint}{\dev}} \big/ \parens{1 + \sum_{\dev\in\devs} e^{\braket{\dpoint}{\dev}}} < 0$ by the continuity of $\vecfield$.
This shows that the requirements of \cref{def:energy} are all satisfied, leading to the following corollary of \cref{thm:attract}:

\begin{corollary}
\label{cor:coherent-first}
Suppose that $\set$ is coherent, and let $\curr$ be the sequence of play of \eqref{eq:method} with step-size and gradient signal assumptions as in \cref{thm:attract}.
Then the conclusions of \cref{thm:attract} hold, namely
\begin{enumerate*}
[(\itshape i\hspace*{.5pt}\upshape)]
\item
if $\set$ is globally coherent, $\curr$ converges to $\set$ \acl{wp1};
and
\item
if $\set$ is locally coherent, $\curr$ converges locally to $\set$ with probability at least $1-\conf$ if $\step$ is small enough.
\end{enumerate*}
\end{corollary}

\cref{cor:coherent-first} is a strong convergence guarantee in itself, but it does not exploit the sharper structural properties of coherent sets.
As we show below, the assumptions of \cref{thm:attract} on the method's step-size and gradient signals can be relaxed considerably, allowing in many cases the use of even \emph{constant} step-sizes.
To simplify the presentation, we will assume throughout that \eqref{eq:method} adheres to the general parameter schedule \eqref{eq:schedule}.
In this general setting, we have:

\begin{theorem}
\label{thm:coherent}
Let $\curr = \mirror(\curr[\dstate])$ be the sequence of play generated by \eqref{eq:method} with step-size and gradient signal sequences as per \eqref{eq:schedule}.
Then:
\begin{enumerate}
[left=0em,label={\normalshape\bfseries Case \arabic*:}]
\item
If $\set$ is globally coherent, then, from any initialization, $\curr$ converges to $\set$ \ac{wp1}.
\item
If $\set$ is locally coherent and, in addition,
\begin{enumerate*}
[\upshape(\itshape i\hspace*{1pt}\upshape)]
\item
$\stepexp - \noisexp > 1/2$;
or
\item
$0 \leq \stepexp < \qexp/(2+\qexp)$ and $\noisexp < 1/2 - 1/\qexp$,
\end{enumerate*}
there exists an open initialization domain $\dnhd\subseteq\dspace$ such that, for any $\conf>0$
\begin{equation}
\label{eq:coherent-local}
\probof{\text{$\curr$ converges to $\set$} \given \init[\dstate]\in\dnhd}
	\geq 1-\conf
\end{equation}
provided that $\step>0$ is small enough.
\end{enumerate}
\end{theorem}

Before discussing the proof of \cref{thm:coherent}, we present a series of explicit convergence guarantees for specific algorithms:

\begin{corollary}
\label{cor:coherent-global}
Suppose that \crefrange{alg:SGA}{alg:EXP3} are run with step-size $\curr[\step] \propto 1/\run^{\stepexp}$, $\stepexp\in[0,1]$, and, where applicable, a sampling parameter $\curr[\mix] \propto 1/\run^{\mixexp}$, $\mixexp\in(0,1/2)$.
If $\set$ is globally coherent, $\curr$ converges to $\set$ \acl{wp1} provided the following conditions are met:
\begin{itemize}
\item
For \cref{alg:SGA,alg:EW,alg:SPSA,alg:DGA,alg:EXP3}:
	\tabto{11em}
	no other requirements needed.
\item
For \cref{alg:EG,alg:OG}:
	\tabto{11em}
	$\stepexp>0$.
\end{itemize}
\end{corollary}

\begin{corollary}
\label{cor:coherent-local}
Suppose that \crefrange{alg:SGA}{alg:EXP3} are run with step-size $\curr[\step] \propto 1/\run^{\stepexp}$, $\stepexp\in[0,1]$, and, where applicable, a sampling parameter $\curr[\mix] \propto 1/\run^{\mixexp}$, $\mixexp\in(0,1/2)$.
Then the conclusions of \cref{thm:coherent} for locally coherent sets continue to hold provided the following conditions are met:
\begin{itemize}
\item
For \cref{alg:SGA}:
	\tabto{11em}
	$\stepexp>1/2$ if $\qexp=2$;
	no such requirement needed if $\qexp>2$.
\item
For \cref{alg:EG,alg:OG}:
	\tabto{11em}
	$\stepexp > 1/2$ if $\qexp=2$;
	$\stepexp > 0$ otherwise.
\item
For \cref{alg:EW,alg:SPSA,alg:DGA,alg:EXP3}:
	\tabto{11em}
	no other requirements needed.
\end{itemize}
\end{corollary}

We should stress here that, depending on the statistical properties of the players' feedback mechanism, the above results imply convergence even with a \emph{constant} step-size, a feature which is quite unique in the context of stochastic approximation.
To the best of our knowledge, the only comparable result in the literature in terms of step-size assumptions is the recent work of \citet{GVM21b} for local convergence to strict \aclp{NE}:
since strict equilibria are locally coherent, the analysis of \citet{GVM21b} corresponds to the last item of \cref{cor:coherent-local}.

Perhaps surprisingly, the principal reason for this relaxation in terms of step-size requirements is \emph{not} the boundedness of the $\qexp$-th moments of the players' oracle:
the step-size requirements of \cref{sec:primal} cannot be relaxed for non-coherent attractors even if $\qexp = \infty$;
at the same time,
the convergence guarantees of \cref{thm:coherent} for globally coherent sets yield convergence with a constant step-size even when $\qexp=2$.
Instead, as we hinted at before, these sharper convergence properties are due to the fact that the quadratic error term $\curr[\aggsecond] = \sum_{\runalt=\start}^{\run} \iter[\step]^{2} \iter[\second]^{2}$ is not present in the case of coherent sets:
it is precisely this simplification that leads to convergence with significantly faster step-size schedules.

Our last result builds on this observation to show that convergence occurs at a finite number of iterations if the mirror map of the process is surjective (\eg if it is a Euclidean projection):%
\footnote{As we explain in \cref{app:mirror}, the image $\im\mirror$ of $\mirror$ coincides with the prox-domain $\proxdom = \dom\subd\hreg$ of $\hreg$.
As such, a sufficient condition for $\mirror$ to be surjective is for $\hreg$ to be Lipschitz continuous on $\points$.}

\begin{theorem}
\label{thm:rate}
Suppose that the mirror map $\mirror\from\dspace\to\points$ of \eqref{eq:method} is surjective.
If $\set$ is coherent, then, \acl{wp1}, every trajectory $\curr = \mirror(\curr[\dstate])$ that converges to $\set$ does so in a finite number of iterations, \ie
there exists some $\run_{0}$ such that $\curr\in\set$ for all $\run\geq\run_{0}$.
\end{theorem}

\begin{corollary}
\label{cor:rate}
Suppose that \eqref{eq:method} is run with Euclidean projections and step-size and gradient signal sequences as per \eqref{eq:schedule}.
If $\set$ is globally coherent and $\points$ is compact, the induced sequence of play $\curr = \mirror(\curr[\dstate])$ converges to $\set$ in a finite number of iterations \as.
\end{corollary}

In view of the above, coherent sets comprise perhaps the most well-behaved class of rational outcomes under \eqref{eq:method}:
the agents' sequence of play converges to such sets in a finite number of iterations, even with bandit, payoff-based feedback.
In turn, this means that the algorithms' long-run behavior remains robust in the face of uncertainty, a property with important implications for the general theory of learning in games.

\subsection{Technical proofs}
\label{sec:proofs-attract}

We conclude this section with the proof of \cref{thm:coherent}.
The key step to achieve this is the following refinement of \cref{lem:template} for coherent sets.

\begin{lemma}
\label{lem:template-coherent}
Suppose that $\set\subseteq\points$ is coherent, and let $\energy_{\dev}(\dpoint) = \braket{\dpoint}{\dev}$ for $\dpoint\in\dspace$, $\dev\in\devs$.
Then the iterates $\curr[\energy] = \energy_{\dev}(\curr[\dstate])$ of $\energy_{\dev}$ satisfy the template inequality
\begin{equation}
\label{eq:template-coherent}
\next[\energy]
	\leq \curr[\energy]
		+ \curr[\step] \braket{\vecfield(\curr)}{\dev}
		+ \curr[\step] \curr[\snoise]
		+ \curr[\step] \curr[\sbias].
\end{equation}
where the error terms $\curr[\snoise]$ and $\curr[\sbias]$ are now given by
\begin{equation}
\label{eq:errors-coherent}
\curr[\snoise]
	= \braket{\curr[\noise]}{\dev}
	\qquad
	\text{and}
	\qquad
\curr[\sbias]
	= \max\nolimits_{\dev\in\devs} \norm{\dev} \cdot \curr[\bbound].
\end{equation}
\end{lemma}

\begin{proof}
Simply set $\score \gets \next[\dstate]$ in $\energy_{\dev}(\score)$ and invoke the definition of \eqref{eq:method}.
\end{proof}

Compared to \cref{lem:template}, the template inequality \eqref{eq:template-coherent} \emph{does not} have a second-order term, so the second moment of $\curr[\signal]$ plays a much more minor role when dealing with coherent sets.
This can be seen very clearly in the following coherent analogue of \cref{prop:attract-global}:

\begin{proposition}
\label{prop:coherent-global}
Suppose that $\set$ is globally coherent,
and
let $\curr = \mirror(\curr[\dstate])$ be the sequence of play generated by \eqref{eq:method}.
If \eqref{eq:sub} holds, then $\curr$ converges to $\set$ \acl{wp1}.
\end{proposition}

The crucial difference between \cref{prop:attract-global,prop:coherent-global} is that the former requires the summability condition \eqref{eq:sum}, while the latter requires \emph{only} the subleading growth requirement \eqref{eq:sub}.
The latter assumption grants much more flexibility to the players because they can employ practically \emph{any} step-size of the form $\curr[\step] \propto 1/\run^{\stepexp}$ for some $\stepexp\in[0,1]$.
A similar situation arises for locally coherent sets, in which case the stability requirement \eqref{eq:stab} can be replaced by the ``dominance'' condition
\begin{subequations}
\makeatletter
\def\@currentlabel{Dom}
\makeatother
\label{eq:dom}
\renewcommand{\theequation}{Dom.\Roman{equation}}
\label{eq:dom}
\noindent%
\begin{flalign}
\label{eq:dom-noise}
\tag{Dom.$\aggnoise$}
\probof{\curr[\aggnoise] \leq \Const\curr[\efftime]^{\texp}/2\;\;\text{for all $\run$}}
	&\geq 1-\conf
	\\
\label{eq:dom-bias}
\tag{Dom.$\aggbias$}
\probof{\curr[\aggbias] \leq \Const\curr[\efftime]^{\texp}/2\;\;\text{for all $\run$}}
	&\geq 1-\conf
\end{flalign}
\end{subequations}
for some $\Const>0$ and $\texp\in[0,1)$.
Under this milder condition, we have:

\begin{proposition}
\label{prop:coherent-local}
Suppose that $\set$ is locally coherent,
fix some confidence level $\conf>0$,
and
let $\curr = \mirror(\curr[\dstate])$ be the sequence of play generated by \eqref{eq:method}.
If \eqref{eq:sub} and \eqref{eq:dom} hold, there exists an unbounded open initialization domain $\dnhd\subseteq\dspace$ such that
\begin{equation}
\probof{\text{$\curr$ converges to $\set$} \given \init[\dstate]\in\dnhd}
	\geq 1-(\nDevs+1)\conf.
\end{equation}
\end{proposition}

To prove \cref{prop:coherent-global,prop:coherent-local} \textendash\ and, through them \cref{thm:coherent,thm:rate} \textendash\ it will be convenient to introduce the family of sets
\begin{equation}
\label{eq:dset}
\txs
\dnhd(\elvl)
	= \setdef
		{\dpoint\in\dspace}
		{\max_{\dev\in\devs} \braket{\dpoint}{\dev} < -\elvl}.
\end{equation}
By \cref{def:coherent}, these sets are mapped to neighborhoods of $\set$ under $\mirror$, so they are particularly well-suited to serve as initialization domains for \eqref{eq:method}.
In particular, by the requirements of \cref{def:coherent} and the continuity of $\vecfield$, there exists some $\elvl$ such that $\const \defeq -\sup\setdef{\braket{\vecfield(\mirror(\dpoint))}{\dev}}{\dpoint\in\dnhd(\elvl), \dev\in\devs} < 0$.
With all this in hand, the proofs of \cref{prop:coherent-global,prop:coherent-local} are fairly straightforward.

\begin{proof}[Proof of \cref{prop:coherent-global}]
Since $\set$ is globally coherent, we can take $\elvl = -\infty$ in the definition of $\dnhd(\elvl)$ above.
Then, telescoping \eqref{eq:template-coherent} readily yields
\begin{equation}
\label{eq:template-agg}
\energy_{\dev}(\next[\dstate])
	\leq \energy_{\dev}(\init[\dstate])
		- \const \curr[\efftime]
		+ \curr[\aggnoise]
		+ \curr[\aggbias]
		\quad
		\text{for all $\dev\in\devs$}.
\end{equation}
Thus, if \eqref{eq:sub} holds, we get $\energy_{\dev}(\curr[\dstate]) \to -\infty$ for all $\dev\in\devs$, \ie $\curr = \mirror(\curr[\dstate]) \to \set$.
\end{proof}

\begin{proof}[Proof of \cref{prop:coherent-local}]
Let $\texp\in[0,1)$ be such that \eqref{eq:dom} holds for every $\dev\in\devs$ (recall that $\curr[\snoise]$ depends on $\dev$), and let
$\ediff = \max_{\run}\braces{\Const\curr[\efftime]^{\texp} - \const\curr[\efftime]}$.
Then, if $\init[\dstate]$ is initialized in $\dnhd \defeq \dnhd(\elvl + \ediff)$, we claim that $\curr[\dstate]\in\dnhd(\elvl)$ for all $\run$.
Indeed, this being trivially true for $\run=1$, assume it to be the case for all $\runalt = \running,\run$.
Then, by \eqref{eq:template-coherent} and our inductive hypothesis, we get
\begin{flalign}
\label{eq:template-agg2}
\energy_{\dev}(\next[\dstate])
	&\leq \energy_{\dev}(\init[\dstate])
		- \sum_{\runalt=\start}^{\run} \iter[\step] \braket{\vecfield(\iter)}{\dev}
		+ \curr[\aggnoise]
		+ \curr[\aggbias]
	\notag\\
	&\leq \energy_{\dev}(\init[\dstate])
		- \const \curr[\efftime]
		+ \Const \curr[\efftime]^{\texp}/2
		+ \Const \curr[\efftime]^{\texp}/2
	\leq -\elvl - \ediff
		+ \ediff
	\leq -\elvl
\end{flalign}
\ie $\next[\dstate] \in \dnhd(\elvl)$, as claimed.
Since $\curr[\dstate] \in \dnhd(\elvl)$ for all $\run$, we conclude that \eqref{eq:template-agg} holds \acl{wp1} on the event that \eqref{eq:dom-noise} and \eqref{eq:dom-bias} both hold for all $\dev\in\devs$.
Since \eqref{eq:dom-noise} involves $\abs{\devs} = \nDevs$ separate events (one for each $\dev\in\devs$) and $\curr[\aggbias]$ does not depend on $\dev$, it follows that $\energy_{\dev}(\curr[\dstate]) \to -\infty$ for all $\dev\in\devs$ with probability at least $1-(\nDevs+1)\conf$.
Our claim then follows from \cref{def:coherent}.
\end{proof}

We are now in a position to prove \cref{thm:coherent}.

\begin{proof}[Proof of \cref{thm:coherent}]
As in the case of \cref{thm:attract}, our proof will hinge on showing that \eqref{eq:sub} and \eqref{eq:dom} hold under the stated step-size and sampling parameter schedules.
Our claim will then follow by a direct application of \cref{prop:coherent-global,prop:coherent-local}.

First, regarding \eqref{eq:sub}, the law of large numbers for martingale difference sequences \citep[Theorem 2.18]{HH80} shows that $\curr[\aggnoise] / \curr[\efftime] \to 0$ \ac{wp1} on the event $\braces*{\sum_{\run} \curr[\step]^{2} \exof{\curr[\snoise]^{2} \given \curr[\filter]} / \curr[\efftime]^{2} < \infty}$.
However
\begin{equation}
\exof{\curr[\snoise]^{2} \given \curr[\filter]}
	\leq \norm{\dev}^{2} \exof{\dnorm{\curr[\noise]}^{2} \given \curr[\filter]}
	\leq \norm{\dev}^{2} \curr[\sdev]^{2}
	= \bigoh(\run^{2\noisexp})
\end{equation}
so, in turn, given that $\noisexp < 1/2$, we get
\begin{equation}
\sum_{\run} \frac{\curr[\step]^{2} \exof{\curr[\snoise]^{2} \given \curr[\filter]}}{\curr[\efftime]^{2}}
	= \bigoh\parens*{ \sum_{\run} \frac{\curr[\step]^{2} \curr[\sdev]^{2}}{\curr[\efftime]^{2}} }
	= \bigoh\parens*{ \sum_{\run} \frac{\run^{-2\stepexp} \run^{2\noisexp}}{\run^{2(1-\stepexp)}} }
	= \bigoh\parens*{ \sum_{\run} \frac{1}{\run^{2 - 2\noisexp}} }
	< \infty.
\end{equation}
This establishes \eqref{eq:sub} for $\curr[Z] \gets \curr[\aggnoise]$;
as for the case of $\curr[\aggbias]$, our claim follows by noting that $\sum_{\runalt=\start}^{\run} \iter[\step] \iter[\bbound] \big/ \sum_{\runalt=\start}^{\run} \iter[\step] \to 0$ if and only if $\curr[\bbound] \to 0$, which is immediate from \eqref{eq:schedule}.
This shows that \eqref{eq:sub} holds, so the first case of the theorem follows from \cref{prop:coherent-global}.

Now, for the second case of the theorem, since $\curr[\bbound]$ is deterministic and $\curr[\bbound] = \bigoh(1/\run^{\biasexp})$ for some $\biasexp>0$, it is always possible to find $\Const>0$ and $\texp\in(0,1)$ so that \eqref{eq:dom-bias} holds.
We are thus left to establish \eqref{eq:dom-noise}.
To that end, let $\curr[\maxnoise] = \sup_{1\leq\runalt\leq\run} \abs{\curr[\aggnoise]}$ and set $\curr[P] \defeq \probof*{\curr[\maxnoise] > \Const\curr[\efftime]^{\texp}/2}$ so
\begin{equation}
\label{eq:dom1}
\curr[P]
	\leq \frac{\exof{\abs{\curr[\aggnoise]}^{\qexp}}}{(\Const/2)^{\qexp}\curr[\efftime]^{\texp\qexp}}
	\leq \const_{\qexp}
		\frac{\exof{\parens*{\sum_{\runalt=\start}^{\run} \iter[\step]^{2} \dnorm{\iter[\noise]}^{2}}^{\qexp/2}}}{\curr[\efftime]^{\texp\qexp}}
\end{equation}
where $c_{\qexp}$ is a positive constant depending only on $\Const$ and $\qexp$, and we used Kolmogorov's inequality \citep[Corollary 2.1]{HH80} in the first step and the \acl{BDG} inequality \citep[Theorem 2.10]{HH80} in the second.
To proceed, we will require the following variant of Hölder's inequality \citep[p.~15]{Ben99}:
\begin{equation}
\label{eq:Holder}
\parens*{\sum_{\runalt=\start}^{\run} \iter[a]\iter[b]}^{\rho}
	\leq \parens*{\sum_{\runalt=\start}^{\run} \iter[a]^{\frac{\lambda\rho}{\rho-1}}}^{\rho-1}
		\sum_{\runalt=\start}^{\run} \iter[a]^{(1-\lambda)\rho} \iter[b]^{\rho}
\end{equation}
valid for all $\iter[a],\iter[b] \geq 0$ and all $\rho>1$, $\lambda\in[0,1)$.
Then, substituting $\iter[a] \gets \iter[\step]^{2}$, $\iter[b] \gets \dnorm{\iter[\noise]}^{2}$, $\rho \gets \qexp/2$ and $\lambda \gets 1/2 - 1/\qexp$, \eqref{eq:dom1} gives
\begin{align}
\label{eq:dom2}
\curr[P]
	\leq \const_{\qexp}
		\frac{
		\parens*{\sum_{\runalt=\start}^{\run} \iter[\step]}^{\qexp/2 - 1}
			\sum_{\runalt=\start}^{\run} \iter[\step]^{1+\qexp/2} \exof{\dnorm{\iter[\noise]}^{\qexp}}}
			{\curr[\efftime]^{\texp\qexp}}
	\leq \const_{\qexp}
		\frac{\sum_{\runalt=\start}^{\run} \iter[\step]^{1+\qexp/2} \iter[\sdev]^{\qexp}}
			{\curr[\efftime]^{1+(\texp-1/2)\qexp}}
\end{align}

We now consider two cases, depending on whether the numerator of \eqref{eq:dom2} is summable or not.
\begin{enumerate}
[left=0em,label={\itshape Case \arabic*:}]
\item
$\stepexp(1+\qexp/2) \geq 1 + \qexp\noisexp$.
In this case, the numerator of \eqref{eq:dom2} is summable under \eqref{eq:schedule}, so the fraction in \eqref{eq:dom2} behaves as $\bigoh(1/\run^{(1-\stepexp) (1 + (\texp-1/2)\qexp)})$.
\item
$\stepexp(1+\qexp/2) < 1 + \qexp\noisexp$.
In this case, the numerator of \eqref{eq:dom2} is not summable under \eqref{eq:schedule}, so the fraction in \eqref{eq:dom2} behaves as $\bigoh\parens*{\run^{1-\stepexp(1+\qexp/2) + \qexp\noisexp} \big/ \run^{(1-\stepexp) (1 + (\texp-1/2)\qexp)}}$.
\end{enumerate}
Thus, working out the various exponents, a straightforward \textendash\ if tedious \textendash\ calculation shows that there exists some $\texp\in(0,1)$ such that $\curr[P]$ is summable as long as $\noisexp < 1/2 - 1/\qexp$ and $0 \leq \stepexp < \qexp/(2+\qexp)$.
Hence, if $\step$ is sufficiently small relative to $\conf$, we conclude that
\begin{equation}
\txs
\probof{\curr[\aggnoise] \leq \Const\curr[\efftime]^{\texp}/2 \; \text{for all $\run$}}
	\geq 1 - \sum_{\run} \curr[P]
	\geq 1 - \conf/2.
\end{equation}
Finally, if $\stepexp>1/2+\noisexp$, \eqref{eq:dom-noise} is a straightforward consequence of \eqref{eq:stab} with $\curr[Z] \gets \curr[\aggnoise]$.
Our assertion then follows by putting everything together and invoking \cref{prop:coherent-local}.
\end{proof}

We conclude this section with the proof of our finite-time convergence result.

\begin{proof}[Proof of \cref{thm:rate}]
Since $\mirror$ is surjective, \cref{lem:mirror} shows that $\mirror^{-1}(\set)$ contains a shifted copy of $\union_{\point\in\set} \pcone(\point)$.
Thus, given that $\max_{\dev\in\devs} \braket{\curr[\dstate]}{\dev} \to -\infty$ by the proof of \cref{thm:coherent}, it follows that, for every $\elvl\in\R$, there exists some (possibly random) $\run_{0} \equiv \run_{0}(\elvl)$ such that $\max_{\dev\in\devs} \braket{\curr[\dstate]}{\dev} < -\elvl$ for all $\run\geq\run_{0}$.
This shows that $\curr[\dstate]$ converges to $\mirror^{-1}(\set)$ within a finite number of iterations, as claimed.
\end{proof}

%% file: Body/Discussion.tex

The proposed \ac{method} stochastic approximation framework captures a wide range of existing algorithms, both first- and zeroth-order, and it allows us to derive a series of convergence results in a unified way.
Conceptually speaking, an appealing feature of this framework lies in the fact that it provides an analysis bluepring that can be used in several other settings and algorithms of interest.
The associated workflow for this is as follows:
it suffices to simply estimate the bounds $\curr[\bbound]$ and $\curr[\sbound]$ for the method under study (in the sense of \cref{tab:algorithms});
the long-run behavior of the method may then be harvested from \cref{thm:ICT,thm:attract,thm:coherent}.
We leave the inclusion of even more general frameworks \textendash\ such as algorithms with adaptive step-sizes, asyncrhonous and/or delayed feedback \cite{ORKB19,HMZ20}, etc. \textendash\ to future work.

Another fruitful direction for future research concerns the wandering behavior of \ac{method} methods.
By Conley's decomposition theorem (the ``fundamental theorem of dynamical systems''), a flow decomposes into a chain recurrent part and an attracting part;
of these, the former is fragile because of the Kupka-Smale theorem, and is actually absent in all but a meager set of flows (in the Baire category sense of the term). 
This means that, generically, one would expect a learning process to wander about different basins of attraction, until, by the attainability theory of \citet{Ben97} and \citet{Duf97}, it is captured by one of them.
Making this statement precise would complement our results in an essential way, and would give us a better understanding of the complex phenomena that arise in game-theoretic learning.

%% file: Appendix/App-Mirror.tex

In this appendix we present some basic properties of the mirror map $\mirror$.
To state them, recall first that the subdifferential of a $\hreg$ at $\point\in\points$ is defined as
\(
\subd\hreg(\point)
	\defeq \setdef
		{\dpoint\in\dspace}
		{\hreg(\pointalt) \geq \hreg(\point) + \braket{\dpoint}{\pointalt - \point} \; \text{for all $\pointalt\in\pspace$}},
\)
the \emph{domain of subdifferentiability} of $\hreg$ is $\dom\subd\hreg \defeq \setdef{\point\in\dom\hreg}{\subd\hreg \neq \varnothing}$,
and
the convex conjugate of $\hreg$ is defined as $\hconj(\dpoint) = \max_{\point\in\points} \{ \braket{\dpoint}{\point} - \hreg(\point) \}$ for all $\dpoint\in\dspace$.
We then have the following basic results.

\begin{lemma}
\label{lem:mirror}
Let $\hreg$ be a regularizer on $\points$, and let $\mirror\from\dspace\to\points$ be its induced mirror map.
Then:
\begin{enumerate}
\item
$\mirror$ is single-valued on $\dspace$:
in particular, for all $\point\in\points$, $\dpoint\in\dspace$, we have $\point = \mirror(\dpoint) \iff \dpoint \in \subd\hreg(\point)$.
\item
The prox-domain $\proxdom \defeq \im\mirror$ of $\hreg$ satisfies $\proxdom = \dom\subd\hreg$ and, hence, $\relint\points \subseteq \proxdom \subseteq \points$.
\item
$\mirror$ is $(1/\hstr)$-Lipschitz continuous and $\mirror = \nabla\hconj$.
\item
For all $\point\in\relint\points$,
we have $\dpoint,\dpointalt\in\subd\hreg(\point)$ if and only if $\braket{\dpointalt - \dpoint}{\pointalt - \point} = 0$ for all $\pointalt \in \points$.
\end{enumerate}
\end{lemma}


Our second basic result concerns the Fenchel coupling
\begin{equation}
\label{eq:Fench}
\fench(\base,\dpoint)
	= \hreg(\base)
		+ \hconj(\dpoint)
		- \braket{\dpoint}{\base}
	\quad
	\text{for $\base\in\points$, $\dpoint\in\dspace$}.
\end{equation}
For our purposes, the most relevant properties of $\fench$ are as follows:

\begin{lemma}
\label{lem:Fench}
For all $\base\in\points$ and all $\dpoint,\dpointalt\in\dspace$, we have:
\begin{subequations}
\setlength{\abovedisplayskip}{3pt}
\setlength{\belowdisplayskip}{3pt}
\begin{flalign}
\label{eq:Fench-posdef}
\qquad
a)
	&\quad
	\fench(\base,\dpoint)
	\geq 0
	\quad
	\text{with equality if and only if $\base = \mirror(\dpoint)$}.
	&
	\\
\label{eq:Fench-norm}
\qquad
b)
	&\quad
	\fench(\base,\dpoint)
	\geq \tfrac{1}{2} \hstr \, \norm{\mirror(\dpoint) - \base}^{2}.
	&
	\\
\label{eq:Fench-bound}
\qquad
c)
	&\quad
	\fench(\base,\dpointalt)
	\leq \fench(\base,\dpoint) + \braket{\dpointalt - \dpoint}{\mirror(\dpoint) - \base} + \tfrac{1}{2\hstr} \dnorm{\dpointalt-\dpoint}^{2}.
\end{flalign}
\end{subequations}
In particular, if $\hreg(0) = 0$, we have
\begin{equation}
(\hstr/2) \norm{\mirror(\dpoint)}^{2}
	\leq \hconj(\dpoint)
	\leq -\min\hreg
		+ \braket{\dpoint}{\mirror(\dpoint)}
		+ (2/\hstr) \dnorm{\dpoint}^{2}
	\quad
	\text{for all $\dpoint\in\dspace$}
\end{equation}
\end{lemma}

Variants of \cref{lem:mirror,lem:Fench} already exist in the literature (see \eg \cite{MZ19} and references therein),
so we do not provide a proof.
Instead, we proceed below to show how the above extends to the \emph{setwise} Fenchel coupling
\begin{equation}
\label{eq:Fench-set}
\fench_{\set}(\dpoint)
	\defeq \hconj(\dpoint) - \hconj_{\set}(\dpoint)
	= \min_{\base\in\set} \{ \hreg(\base) + \hconj(\dpoint) - \braket{\dpoint}{\base} \}
	= \min_{\base\in\set} \fench(\base,\dpoint)
\end{equation}
where
$\set$ is a nonempty compact convex subset of $\points$
and
$\hconj_{\set}(\dpoint) = \max_{\point\in\set} \{\braket{\dpoint}{\point} - \hreg(\dpoint)\}$ denotes the convex conjugate of $\hreg$ relative to $\set$.
The most important properties of $\fench_{\set}$ are encoded in the following lemma.

\begin{lemma}
\label{lem:Fench-set}
With notation as above, we have:
\begin{enumerate}
\item
$\fench_{\set}(\dpoint) \geq 0$ with equality if and only if $\mirror(\dpoint) \in \set$.
Moreover, under the reciprocity condition \eqref{eq:rec}, we have $\fench_{\set}(\dpoint) \to 0$ if and only if $\mirror(\dpoint) \to \set$.
\item
$\fench_{\set}$ is differentiable and $\nabla\fench_{\set}(\dpoint) = \mirror(\dpoint) - \mirror_{\set}(\dpoint)$,
where $\mirror_{\set}(\dpoint) = \argmax_{\point\in\set} \{ \braket{\dpoint}{\point} - \hreg(\point) \}$.
\item
For all $\dpoint,\dpointalt\in\dspace$ we have $\norm{\nabla\fench_{\set}(\dpointalt) - \nabla\fench_{\set}(\dpoint)} \leq (2/\hstr) \dnorm{\dpointalt - \dpoint}$.
\end{enumerate}
\end{lemma}

\begin{proof}
Since $\set\subseteq\points$, we have $\hconj_{\set} \leq \hconj$ by definition, and hence $\fench_{\set} \geq 0$.
Moreover, since the minimum in \eqref{eq:Fench-set} must be attained in $\set$, we get $\fench_{\set}(\dpoint) = 0$ if and only if $\hreg(\base) - \hconj(\dpoint) - \braket{\dpoint}{\base} = 0$ for some $\base\in\set$;
by \cref{lem:Fench}, this occurs if and only if $\mirror(\dpoint) = \base \in \set$, so our first claim follows.

Moving forward, to show that $\fench_{\set}(\dpoint)\to0$ if and only if $\mirror(\dpoint) \to \set$, let $\curr[\dpoint]$ be a sequence in $\dpoints$, and let $\curr[\point] = \mirror(\curr[\dpoint])$.
For the ``if'' part, since $\set$ is compact, we may assume without loss of generality (and by descending to  a subsequence if necessary) that $\curr[\point]$ converges to some $\point\in\set$.
Observe now that
\begin{enumerate*}
[\itshape a\upshape)]
\item
$0 \leq \fench_{\set}(\curr[\dpoint]) \leq \fench(\point,\curr[\dpoint])$ by the minimum \eqref{eq:Fench-set};
and
\item
$\fench(\point,\curr[\dpoint])\to0$ by \eqref{eq:rec}.
\end{enumerate*}
Thus, by sandwiching, we conclude that $\fench_{\set}(\curr[\dpoint])\to0$.
Conversely, if $\fench_{\set}(\curr[\dpoint]) \to 0$, we may again assume by compactness (and by descending to a subsequence if necessary) that $\curr[\point] = \mirror(\curr[\dpoint])$ converges to some $\hat\point\in\points$.
If $\hat\point\not\in\set$, then, by \eqref{eq:rec}, we have $\lim_{\run\to\infty} \fench(\point,\curr[\dpoint]) > 0$ for all $\point\in\set$.
Since $\set$ is compact and $\fench(\point,\dpoint)$ is continuous in $\point$, we conclude that $\liminf_{\run\to\infty} \fench_{\set}(\curr[\dpoint]) > 0$, a contradiction which establishes our claim.

Our last two claims follow by applying \cref{lem:mirror} to $\hreg$ and $\hreg+\delta_{\set}$ where $\delta_{\set}$ denotes the convex indicator of $\set$.
\end{proof}

The next properties we discuss concern the way that different regions of $\dspace$ are mapped to $\points$ under $\mirror$.

\begin{lemma}[{\citeauthor{MS16}, \citeyear{MS16}, Prop.~A.1}]
\label{lem:support}
Let $\hreg$ be a regularizer on the simplex $\simplex(\pures) \subseteq \R^{\pures}$.
If $\dpoint_{\pure} - \dpoint_{\purealt} \to -\infty$, then $\mirror_{\pure}(\dpoint) \to 0$.
\end{lemma}

\begin{lemma}
\label{lem:target}
Let $\hreg$ be a regularizer on $\points$, let $\curr[\dpoint]$, $\run=\running$ be a sequence in $\dspace$, and fix some $\point\in\points$.
If $\braket{\curr[\dpoint]}{\tvec} \to -\infty$ for every nonzero $\tvec\in\tcone(\point)$, we have $\mirror(\curr[\dpoint]) \to \point$.
\end{lemma}

\begin{proof}
Assume that $\limsup_{\run} \norm{\curr[\point] - \point} > 0$.
Then, given that $\curr[\dpoint] \in \subd\hreg(\curr[\point])$, we get
\(
\hreg(\point)
	\geq \hreg(\curr[\point])
		+ \braket{\curr[\dpoint]}{\point - \curr[\point]}
	\geq \hreg(\curr[\point])
		- \braket{\curr[\dpoint]}{\curr[\tvec]} \norm{\curr[\point] - \point},
\)
where we set $\curr[\tvec] = (\curr[\point] - \point) / \norm{\curr[\point] - \point}$.
If we further assume (by descending to a subsequence if needed) that $\curr[\tvec]$ converges in the unit sphere of $\norm{\cdot}$, there exists some $\tvec\in\tcone(\point)$ with $\norm{\tvec} = 1$ and such that $\braket{\curr[\dpoint]}{\curr[\tvec]} \leq (1+\eps) \braket{\curr[\dpoint]}{\tvec}$ for some $\eps>0$.
Thus, taking the $\limsup$ of the above estimate gives $\hreg(\point) \geq \infty$, a contradiction which proves our claim.
\end{proof}

\begin{lemma}
\label{lem:poly}
Let $\hreg$ be a regularizer on a convex polytope $\poly$ of $\pspace$, let $\set$ be a face of $\poly$, and let $\devs = \braces{\dev_{1},\dotsc,\dev_{\nDevs}}$ be a set of unit vectors of $\pspace$ such that every point $\point\in\poly\setminus\set$ can be written as $\point = \base + \coef\dev$ for some $\base\in\set$, $\dev\in\devs$ and $\coef>0$.
If $\max_{\dev\in\devs} \braket{\dpoint}{\dev} \to -\infty$, then $\mirror(\dpoint) \to \set$.
\end{lemma}

\begin{proof}
By the compactness of $\poly$ (and descending to a subsequence if necessary), we may assume that $\curr[\point] = \mirror(\curr[\dpoint])$ converges to some $\point\in\poly$.
If $\point\notin\set$, there exist $\base\in\set$, $\dev\in\devs$ and $\coef>0$ such that $\point = \base + \coef\dev$.
In turn, this gives $\hreg(\base) \geq \hreg(\curr[\point]) + \braket{\curr[\dpoint]}{\base - \curr[\point]} = \hreg(\curr[\point]) - \braket{\curr[\dpoint]}{\curr[\tvec]} \norm{\curr[\point] - \base}$ where we set $\curr[\tvec] = (\curr[\point] - \base)/\norm{\curr[\point] - \base}$.
Since $\curr[\tvec] \to \tvec$, taking $\run\to\infty$ yields $\hreg(\base) \geq \infty$, a contradiction which shows that $\point = \lim \curr[\point] \in\set$, as claimed.
\end{proof}

We conclude this appendix with the dynamical properties of the Fenchel coupling under \eqref{eq:MD}.
The heavy lifting will be provided by the following simple lemma:

\begin{lemma}
\label{lem:Fench-MD}
Let $\trajof{\time} = \mirror(\dtrajof{\time})$ be an orbit of \eqref{eq:MD}.
Then, for every nonempty closed convex subset $\set$ of $\points$, we have
\begin{equation}
\label{eq:Fench-MD}
\dot\fench_{\set}(\dpoint)
	= \braket{\vecfield(\point)}{\point - \point_{\set}}
\end{equation}
where $\point_{\set} = \mirror_{\set}(\dpoint)$ denotes the mirror image of $\dpoint$ on $\set$.
In particular, if $\set = \{\base\}$, we have $\dot\fench(\base,\dpoint) = \braket{\vecfield(\point)}{\point - \base}$.
\end{lemma}

\begin{proof}
Simply note that $\dot\dpoint = -\vecfield(\point)$ and apply \cref{lem:Fench-set}.
\end{proof}

%% file: Appendix/App-Proofs.tex

\subsection{Error bounds for specific algorithms}

Our aim in this appendix is to prove the bounds on the bias and magnitude of $\curr[\signal]$ reported in \cref{prop:algorithms,tab:algorithms}.

\begin{proof}[Proof of \cref{prop:algorithms}]
We proceed in a method-by-method basis starting with the oracle-based methods of \cref{sec:algorithms}, that is, \crefrange{alg:SGA}{alg:EW}.
For this, we will make free use of the fact that we can take $\curr[\sbound]^{\qexp} = 3^{\qexp-1} (\vbound^{\qexp} + \curr[\bbound]^{\qexp} + \curr[\sdev]^{\qexp})$ in \eqref{eq:errorbounds}, \cf the discussion after \eqref{eq:signal}.

\itempar{\cref{alg:SGA}: \Acl{SGA}}

For \eqref{eq:SGA}, we have $\curr[\noise] = \err(\curr;\curr[\seed])$ and $\curr[\bias] = 0$, so our claim follows immediately from the stated assumptions for \eqref{eq:SFO}.

\itempar{\cref{alg:EG}: \Acl{EG}}

For \eqref{eq:EG}, we have $\curr[\signal] = \orcl(\lead;\lead[\seed])$ so $\exof{\curr[\signal] \given \curr[\filter]} = \exof{\vecfield(\lead) \given \curr[\filter]}$.
We thus get
\begin{align}
\dnorm{\curr[\bias]}
	= \dnorm{\exof{\curr[\signal] \given \curr[\filter]} - \vecfield(\curr)}
	&\leq \exof{\dnorm{\vecfield(\lead) - \vecfield(\curr)} \given \curr[\filter]}
	\notag\\
	&\leq \lips \exof{\norm{\lead - \curr} \given \curr[\filter]}
	\notag\\
	&\leq \curr[\step] \lips \exof{\dnorm{\orcl(\curr;\curr[\seed])} \given \curr[\filter]}
	\notag\\
	&= \curr[\step] \lips \exof{\dnorm{\vecfield(\curr) + \err(\curr;\curr[\seed])} \given \curr[\filter]}
	\notag\\
	&\leq \curr[\step] \lips (\vbound + \sdev)
	= \bigoh(\curr[\step])
		= \bigoh(1/\run^{\stepexp})
\end{align}
and, analogously
\begin{equation}
\dnorm{\curr[\noise]}
	= \dnorm{
		\curr[\signal]
		- \exof{\curr[\signal] \given \curr[\filter]}
		}
	= \dnorm{
		\vecfield(\lead)
		- \exof{\vecfield(\lead) \given \curr[\filter]}
		+ \err(\lead;\lead[\seed])}
\end{equation}
so $\exof{\dnorm{\curr[\noise]}^{\qexp} \given \curr[\filter]} = \bigoh(\vbound^{\qexp} + \sdev^{\qexp}) = \bigoh(1)$ under \eqref{eq:oracle}, as claimed.

\itempar{\cref{alg:OG}: \Acl{OG}}

For \eqref{eq:OG}, we have again $\exof{\curr[\signal] \given \curr[\filter]} = \exof{\vecfield(\lead) \given \curr[\filter]}$, so the same series of arguments as above gives
\begin{align}
\dnorm{\curr[\bias]}
	&= \dnorm{\exof{\curr[\signal] \given \curr[\filter]} - \vecfield(\curr)}
	\notag\\
	&\leq \lips \exof{\norm{\lead - \curr} \given \curr[\filter]}
	\notag\\
	&\leq \curr[\step] \lips \exof{\dnorm{\orcl(\beforelead;\prev[\seed])} \given \curr[\filter]}
	\notag\\
	&= \curr[\step] \lips \exof{\dnorm{\vecfield(\beforelead) + \err(\beforelead;\prev[\seed])} \given \curr[\filter]}
	\notag\\
	&\leq \curr[\step] \lips (\vbound + \sdev)
	= \bigoh(\curr[\step])
		= \bigoh(1/\run^{\stepexp})
\end{align}
under \eqref{eq:oracle} with $\qexp=\infty$.
The noise term $\curr[\noise]$ can be bounded in exactly the same way, so we omit the calculations.

\itempar{\cref{alg:EW}: \Acl{EW}}

We consider two cases, based on the information available to the players.
For the full information oracle \eqref{eq:SFO-finite-full}, we have $\curr[\signal] = \vecfield(\curr)$ so $\curr[\bias] = \curr[\noise] = 0$ by definition (\ie the oracle is perfect).
Otherwise, under the realization-based oracle \eqref{eq:SFO-finite-realized}, we have $\exof{\curr[\signal] \given \curr[\filter]} = \exof{\vecfield(\curr[\pure]) \given \curr[\filter]} = \vecfield(\curr)$ because $\curr[\pure]$ is sampled according to $\curr$.
We thus get $\curr[\bias] = 0$ and $\curr[\noise] = \bigoh(1)$, which proves our assertion.

%

\smallskip
We now proceed with the payoff-based methods of \cref{sec:algorithms}, namely \crefrange{alg:SPSA}{alg:EXP3}.

\itempar{\cref{alg:SPSA}: \Acl{SPSA}}

Since $\pay_{\play}$ is assumed bounded in the context of \eqref{eq:SPSA}, the bound for $\curr[\sbound]$ follows trivially.
As for the bias of \eqref{eq:SPSA}, it will be convenient to set
$\orcl_{\play}^{\mix}(\point;\unitvec) = (\vdim_{\play}/\mix) \, \pay_{\play}(\point+\mix\unitvec) \, \unitvec_{\play}$ so, in obvious notation, $\signal_{\play,\run} = \orcl_{\play}^{\curr[\mix]}(\curr;\curr[\unitvar])$.
Thus, if we fix a pivot point $\pivot\in\points$ and a query point $\querypoint = \pivot + \mix\unitvec$ for some $\unitvec\in\bvecs = \prod_{\play} \bvecs_{\play}$, a first-order Taylor expansion of $\pay_{\play}$ with integral remainder gives
\begin{subequations}
\label{eq:Taylor}
\begin{align}
\label{eq:Taylor1}
\orcl_{\play}^{\mix}(\pivot;\unitvec)
	= \frac{\vdim_{\play}}{\mix} \pay_{\play}(\querypoint) \cdot \unitvec_{\play}
	= \frac{\vdim_{\play}}{\mix} \pay_{\play}(\pivot) \cdot \unitvec_{\play}
		&+ \frac{\vdim_{\play}}{\mix} \braket{\nabla\pay_{\play}(\pivot)}{\perturb}
		\cdot \unitvec_{\play}
	\\
\label{eq:Taylor2}
		&+ \int_{0}^{1}
				\braket{\nabla\pay_{\play}(\pivot+\tau\perturb) - \nabla\pay_{\play}(\pivot)}{\perturb}
				\dd\tau \cdot
		\unitvec_{\play}
\end{align}
\end{subequations}
where we set $\perturb = \querypoint - \pivot = \mix\unitvec$.
Hence, if $\unitvec$ is drawn uniformly at random from $\bvecs$, taking expectations yields
\begin{align}
\label{eq:Taylor1-temp}
\exof{\eqref{eq:Taylor1}}
	&= \frac{\vdim_{\play}}{\mix}
		\exof{\braket{\vecfield_{\play}(\pivot)}{\perturb_{\play}} \,
			\unitvec_{\play}}
	+ \frac{\vdim_{\play}}{\mix}
		\sum_{\playalt\neq\play}
		\braket{\nabla_{\pivot_{\playalt}} \pay_{\play}(\pivot)}{\exof{\perturb_{\playalt}}} \,
		\exof{\unitvec_{\play}}
	\notag\\
	&= \vdim_{\play} \exof{\braket{\vecfield_{\play}(\pivot)}{\unitvec_{\play}} \, \unitvec_{\play}}
	= \vdim_{\play} \cdot
		\frac{1}{2\vdim_{\play}}
		\sum_{\ell=1}^{\vdim_{\play}}
			\bracks{\vecfield_{\play\ell}(\pivot) \bvec_{\play\ell} - \vecfield_{\play\ell}(\pivot) (-\bvec_{\play\ell})}
	= \vecfield_{\play}(\pivot)
\end{align}
where we used the fact that $\exof{\unitvec_{\play}} = 0$ for all $\play\in\players$ and that $\unitvec_{\play}$ and $\unitvec_{\playalt}$ are independent for all $\play,\playalt\in\players$, $\play\neq\playalt$.
As for the second term, \cref{asm:game} readily yields
\begin{align}
\label{eq:Taylor2-temp}
\norm{\exof{\eqref{eq:Taylor2}}}
	&\leq \frac{\vdim_{\play}}{\mix}
		\int_{0}^{1} \lips_{\play} \mix^{2} \norm{\unitvec}^{2} \tau\dd\tau
	= \bigoh\parens*{\lips\mix}.
\end{align}
Thus, by combining \eqref{eq:Taylor1-temp} and \eqref{eq:Taylor2-temp}, we conclude that $\bias_{\play,\run} = \exof{\orcl_{\play}^{\curr[\mix]}(\curr;\curr[\unitvar]) \given \curr[\filter]} - \vecfield_{\play}(\curr) = \bigoh(\curr[\mix])$, which immediately yields the desired bound $\curr[\bbound] = \bigoh(\curr[\mix]) = \bigoh(1/\run^{\mixexp})$  for \eqref{eq:SPSA}.

\itempar{\cref{alg:DGA}: \Acl{DGA}}

Recall that $\signal_{\play,\run} = \run\cdot\log\parens{1 + (\pay_{\play}(\lead) - \pay_{\play}(\curr)) \unitvar_{\play,\run}}$.
Since $\pay_{\play}(\lead) - \pay_{\play}(\curr) = (1/\run) \vecfield_{\play}(\curr) \unitvar_{\play,\run} + \bigoh(1/\run^{2})$ by the definition of $\lead$,
expanding the logairthm readily yiels $\curr[\bbound] = \bigoh(1/\run)$ and $\curr[\sbound] = \bigoh(1)$.
Our claim then follows as above.

\itempar{\cref{alg:EXP3}: \Acl{EXP3}}

Since $\curr[\purequery]$ is sampled according to $\curr[\query]$, we readily get $\exof{\signal_{\play,\run} \given \curr[\filter]} = \vecfield_{\play}(\curr[\query])$, so $\curr[\bbound] = \bigoh(\norm{\curr[\query] - \curr}) = \bigoh(\curr[\mix]) = \bigoh(1/\run^{\mixexp})$.
Moreover, since $\query_{\play\pure_{\play},\run} \geq \curr[\mix]/\nPures_{\play}$, it follows that $\dnorm{\curr[\signal]} = \bigoh(1/\curr[\mix]) = \bigoh(\run^{\mixexp})$, and our proof is complete.
\end{proof}

\subsection{Energy function derivations}

Our aim in this last appendix is to prove the energy properties of the Fenchel coupling as stated in \cref{lem:energy-VS,lem:energy-cvx}.
For concision, we will prove both as a special case of the following general result:

\begin{proposition}
\label{prop:energy-general}
Let $\set$ be a nonempty compact convex subset of $\points$, and assume that there exists a neighborhood $\nhd$ of $\set$ such that
\begin{equation}
\label{eq:VS-set}
\braket{\payv(\point)}{\point - \base}
	\leq 0
	\quad
	\text{for all $\point\in\nhd$, $\base\in\set$},
\end{equation}
with equality if and only if $\point\in\set$.
If \eqref{eq:rec} holds and $\gauge$ is defined as in \eqref{eq:gauge}, the function $\energy\from\dspace\to\R$ given by
\begin{equation}
\label{eq:energy-global}
\energy(\dpoint)
	= \gauge(\fench_{\set}(\dpoint))
	\quad
	\text{for all $\dpoint\in\dspace$}
\end{equation}
is a local energy function for $\set$ under \eqref{eq:MD}.
In addition, if $\nhd = \points$, $\energy$ is a global energy function for $\set$.
\end{proposition}

\begin{proof}
We will verify the requirements of \cref{def:energy} in order.
\begin{enumerate}
\item
For the first (Lipschitz continuity and smoothness),
note that $\nabla\energy(\dpoint) = \gauge'(\fench_{\set}(\dpoint)) \nabla\fench_{\set}(\dpoint)$, so, lettting  $\point = \mirror(\dpoint)$ and $\point_{\set} = \mirror_{\set}(\dpoint)$, \cref{lem:Fench-set} yields
\begin{equation}
\label{eq:energy-grad}
\nabla\energy(\dpoint)
	= (\point - \point_{\set})
		\cdot \begin{dcases*}
			1
				&\quad
				if $\fench_{\set}(\dpoint) \leq 1$,
			\\
			1/\sqrt{\fench_{\set}(\dpoint)}
				&\quad
				otherwise.
		\end{dcases*}
\end{equation}
Furthermore, by \cref{lem:Fench}, we also have $\fench(\base,\dpoint) \geq (\hstr/2) \norm{\point - \base}^{2}$ for all $\base\in\set$, so, by minimizing over $\base\in\set$, we get
\begin{equation}
\label{eq:Fench-lower}
\fench_{\set}(\dpoint)
	\geq (\hstr/2) \dist(\point,\set)^{2}
	= (\hstr/2) \norm{\point - \proj_{\set}(\point)}^{2}
\end{equation}
where $\proj_{\set}(\point) \defeq \argmin_{\base\in\set} \norm{\point - \base}$.
In turn, this gives $\norm{\point - \proj_{\set}(\point)} \leq \sqrt{2/\hstr}$ whenever $\fench_{\set}(\dpoint) \leq 1$, so we get
\begin{equation}
\label{eq:energy-grad-bound1}
\norm{\nabla\energy(\dpoint)}
	= \norm{\point - \point_{\set}}
	\leq \norm{\point - \proj_{\set}(\point)}
		+ \norm{\proj_{\set}(\point) - \point_{\set}}
	\leq \sqrt{2/\hstr}
		+ \diam(\set)
\end{equation}
whenever $\fench_{\set}(\dpoint) \leq 1$.
On the other hand, if $\fench_{\set}(\dpoint) \geq 1$, we have
\begin{align}
\norm{\nabla\energy(\dpoint)}
	= \frac{\norm{\nabla\fench_{\set}(\dpoint)}}{\sqrt{\fench_{\set}(\dpoint)}}
	&\leq \sqrt{\frac{2}{\hstr}} \frac{\norm{\point - \point_{\set}}}{\norm{\point - \proj_{\set}(\point)}}
	\explain{by \eqref{eq:energy-grad} and \eqref{eq:Fench-lower}}
	\\
	&\leq \sqrt{\frac{2}{\hstr}}
		\parens*{1 + \frac{\norm{\proj_{\set}(\point) - \point_{\set}}}{\norm{\point - \proj_{\set}(\point)}}}
	\explain{by the triangle inequality}
	\\
	&\leq \sqrt{\frac{2}{\hstr}}
		\parens*{1 + \frac{\diam(\set)}{\norm{\point - \proj_{\set}(\point)}}}
\label{eq:energy-grad-bound2}
\end{align}
By the reciprocity condition \eqref{eq:rec}, it follows that the set $\setdef{\point = \mirror(\dpoint)}{\fench_{\set}(\dpoint) \geq 1}$ is well-separated from $\set$, so $\norm{\point - \proj_{\set}(\point)}$ is bounded away from zero if $\fench_{\set}(\dpoint) \geq 1$.
Thus, by combining \cref{eq:energy-grad-bound1,eq:energy-grad-bound2}, we conclude that $\norm{\nabla\energy(\dpoint)}$ is bounded.
Finally, again by \cref{lem:Fench}, $\fench(\dpoint)$ is $(1/\hstr)$-Lipschitz smooth, so
\cref{eq:dsmooth} \textendash\ which is equivalent to the Lipschitz smoothness of $\energy$ \textendash\ follows immediately from the concavity of $\gauge$.

\item
For the positive-definiteness requirement of \cref{def:energy}, note that \cref{lem:Fench-set} and the reciprocity condition \eqref{eq:rec} yield $\mirror(\dpoint) \to \set$ if and only if $\fench_{\set}(\dpoint) \to 0$.
Thus, given that $\gauge(z) = z$ for small $z$, the same will hold for $\energy = \gauge\circ\fench_{\set}$, and our claim follows.

\item
Finally, for the Lyapunov properties of $\energy$ under \eqref{eq:MD}, recall that \cref{lem:Fench-MD} gives $\dot\fench_{\set}(\dpoint) = \braket{\vecfield(\point)}{\point - \point_{\set}}$, so
\begin{equation}
\dot\energy(\dpoint)
	= \braket{\dot\dpoint}{\nabla\energy(\dpoint)}
	= \gauge'(\fench(\dpoint))
		\braket{\vecfield(\point)}{\point - \point_{\set}}
	< 0
	\quad
	\text{whenever $\point \in \nhd \setminus \set$}
\end{equation}
where we used the defining property \eqref{eq:VS-set} of $\set$ (recall that $\point_{\set} \in \set$ by construction).
Moving forward, by \cref{lem:Fench-set}, there exists some $\emax>\ground$ such that the sublevel set $\dbasin = \setdef{\dpoint\in\dspace}{\fench_{\set}(\dpoint) \leq \emax}$ is mapped to $\nhd$ under $\mirror$, \ie $\mirror(\dpoint) \in \nhd$ whenever $\fench_{\set}(\dpoint) \leq \emax$.
Thus, putting everything together, we conclude that $\dot\energy(\dpoint)\to0$ if and only if $\fench_{\set}(\dpoint)\to0$, which implies that $\sup\setdef{\dot\energy(\dpoint)}{\loen < \energy(\dpoint) < \hien} < 0$ for all $\loen \in (\ground,\hien)$, and our proof is complete.
\qedhere
\end{enumerate}
\end{proof}

%% file: Thanks.tex
%
%
P.~Mertikopoulos is grateful for financial support by
the French National Research Agency (ANR) in the framework of
the ``Investissements d'avenir'' program (ANR-15-IDEX-02),
the LabEx PERSYVAL (ANR-11-LABX-0025-01),
MIAI@Grenoble Alpes (ANR-19-P3IA-0003),
and the bilateral ANR-NRF grant ALIAS (ANR-19-CE48-0018-01).
This project has also received funding from the European Research Council (ERC) under the European Union's Horizon 2020 research and innovation programme (grant agreement \textnumero\ 725594-TIME-DATA) and from the Swiss National Science Foundation (SNSF) under grant number 200021-205011.

%% file: Main.bbl
\begin{thebibliography}{71}
\providecommand{\natexlab}[1]{#1}
\providecommand{\url}[1]{\texttt{#1}}
\expandafter\ifx\csname urlstyle\endcsname\relax
  \providecommand{\doi}[1]{doi: #1}\else
  \providecommand{\doi}{doi: \begingroup \urlstyle{rm}\Url}\fi

\bibitem[Arrow et~al.(1958)Arrow, Hurwicz, and Uzawa]{AHU58}
Arrow, K.~J., Hurwicz, L., and Uzawa, H.
\newblock \emph{Studies in linear and non-linear programming}.
\newblock Stanford University Press, 1958.

\bibitem[Auer et~al.(1995)Auer, Cesa-Bianchi, Freund, and Schapire]{ACBFS95}
Auer, P., Cesa-Bianchi, N., Freund, Y., and Schapire, R.~E.
\newblock Gambling in a rigged casino: The adversarial multi-armed bandit
  problem.
\newblock In \emph{Proceedings of the 36th Annual Symposium on Foundations of
  Computer Science}, 1995.

\bibitem[Azizian et~al.(2022)Azizian, Iutzeler, Malick, and
  Mertikopoulos]{AIMM22}
Azizian, W., Iutzeler, F., Malick, J., and Mertikopoulos, P.
\newblock On the rate of convergence of {Bregman} proximal methods in
  constrained variational inequalities.
\newblock \url{http://arxiv.org/abs/2211.08043}, 2022.

\bibitem[Bauschke \& Combettes(2017)Bauschke and Combettes]{BC17}
Bauschke, H.~H. and Combettes, P.~L.
\newblock \emph{Convex Analysis and Monotone Operator Theory in Hilbert
  Spaces}.
\newblock Springer, New York, NY, USA, 2 edition, 2017.

\bibitem[Bena{\"\i}m(1997)]{Ben97}
Bena{\"\i}m, M.
\newblock Vertex reinforced random walks and a conjecture of {Pemantle}.
\newblock \emph{Annals of Probability}, 25:\penalty0 361--392, 1997.

\bibitem[Bena{\"\i}m(1999)]{Ben99}
Bena{\"\i}m, M.
\newblock Dynamics of stochastic approximation algorithms.
\newblock In Az{\'e}ma, J., {\'E}mery, M., Ledoux, M., and Yor, M. (eds.),
  \emph{S{\'e}minaire de Probabilit{\'e}s XXXIII}, volume 1709 of \emph{Lecture
  Notes in Mathematics}, pp.\  1--68. Springer Berlin Heidelberg, 1999.

\bibitem[Bena{\"\i}m \& Hirsch(1996)Bena{\"\i}m and Hirsch]{BH96}
Bena{\"\i}m, M. and Hirsch, M.~W.
\newblock Asymptotic pseudotrajectories and chain recurrent flows, with
  applications.
\newblock \emph{Journal of Dynamics and Differential Equations}, 8\penalty0
  (1):\penalty0 141--176, 1996.

\bibitem[Bervoets et~al.(2020)Bervoets, Bravo, and Faure]{BBF20}
Bervoets, S., Bravo, M., and Faure, M.
\newblock Learning with minimal information in continuous games.
\newblock \emph{Theoretical Economics}, 15:\penalty0 1471--1508, 2020.

\bibitem[Bravo et~al.(2018)Bravo, Leslie, and Mertikopoulos]{BLM18}
Bravo, M., Leslie, D.~S., and Mertikopoulos, P.
\newblock Bandit learning in concave ${N}$-person games.
\newblock In \emph{NeurIPS '18: Proceedings of the 32nd International
  Conference of Neural Information Processing Systems}, 2018.

\bibitem[Brown(1951)]{Bro51}
Brown, G.~W.
\newblock Iterative solutions of games by fictitious play.
\newblock In Coopmans, T.~C. (ed.), \emph{Activity Analysis of Productions and
  Allocation}, 374-376. Wiley, 1951.

\bibitem[Cesa-Bianchi \& Lugosi(2006)Cesa-Bianchi and Lugosi]{CBL06}
Cesa-Bianchi, N. and Lugosi, G.
\newblock \emph{Prediction, Learning, and Games}.
\newblock Cambridge University Press, 2006.

\bibitem[Coucheney et~al.(2015)Coucheney, Gaujal, and Mertikopoulos]{CGM15}
Coucheney, P., Gaujal, B., and Mertikopoulos, P.
\newblock Penalty-regulated dynamics and robust learning procedures in games.
\newblock \emph{Mathematics of Operations Research}, 40\penalty0 (3):\penalty0
  611--633, August 2015.

\bibitem[Daskalakis \& Panageas(2019)Daskalakis and Panageas]{DP19}
Daskalakis, C. and Panageas, I.
\newblock Last-iterate convergence: {Zero}-sum games and constrained min-max
  optimization.
\newblock In \emph{ITCS '19: Proceedings of the 10th Conference on Innovations
  in Theoretical Computer Science}, 2019.

\bibitem[Daskalakis et~al.(2018)Daskalakis, Ilyas, Syrgkanis, and Zeng]{DISZ18}
Daskalakis, C., Ilyas, A., Syrgkanis, V., and Zeng, H.
\newblock Training {GANs} with optimism.
\newblock In \emph{ICLR '18: Proceedings of the 2018 International Conference
  on Learning Representations}, 2018.

\bibitem[Debreu(1952)]{Deb52}
Debreu, G.
\newblock A social equilibrium existence theorem.
\newblock \emph{Proceedings of the National Academy of Sciences of the USA},
  38\penalty0 (10):\penalty0 886--893, October 1952.

\bibitem[Duflo(1997)]{Duf97}
Duflo, M.
\newblock Cibles atteignables avec une probabilit{\'e} positive d'apr{\`e}s {M.
  Bena{\"\i}m}.
\newblock mimeo, 1997.

\bibitem[Duvocelle et~al.(2023)Duvocelle, Mertikopoulos, Staudigl, and
  Vermeulen]{DMSV23}
Duvocelle, B., Mertikopoulos, P., Staudigl, M., and Vermeulen, D.
\newblock Multi-agent online learning in time-varying games.
\newblock \emph{Mathematics of Operations Research}, 48\penalty0 (2):\penalty0
  914--941, May 2023.

\bibitem[Even-dar et~al.(2009)Even-dar, Mansour, and Nadav]{EMN09}
Even-dar, E., Mansour, Y., and Nadav, U.
\newblock On the convergence of regret minimization dynamics in concave games.
\newblock In \emph{{STOC} '09: Proceedings of the 41st annual ACM symposium on
  the Theory of Computing}, pp.\  523--532, New York, NY, 2009. ACM.

\bibitem[Flokas et~al.(2020)Flokas, Vlatakis-Gkaragkounis, Lianeas,
  Mertikopoulos, and Piliouras]{FVGL+20}
Flokas, L., Vlatakis-Gkaragkounis, E.~V., Lianeas, T., Mertikopoulos, P., and
  Piliouras, G.
\newblock No-regret learning and mixed {Nash} equilibria: {They} do not mix.
\newblock In \emph{NeurIPS '20: Proceedings of the 34th International
  Conference on Neural Information Processing Systems}, 2020.

\bibitem[Giannou et~al.(2021)Giannou, Vlatakis-Gkaragkounis, and
  Mertikopoulos]{GVM21b}
Giannou, A., Vlatakis-Gkaragkounis, E.~V., and Mertikopoulos, P.
\newblock The convergence rate of regularized learning in games: {From} bandits
  and uncertainty to optimism and beyond.
\newblock In \emph{NeurIPS '21: Proceedings of the 35th International
  Conference on Neural Information Processing Systems}, 2021.

\bibitem[Giannou et~al.(2022)Giannou, Lotidis, Mertikopoulos, and
  Vlatakis-Gkaragkounis]{GLMV22}
Giannou, A., Lotidis, K., Mertikopoulos, P., and Vlatakis-Gkaragkounis, E.~V.
\newblock On the convergence of policy gradient methods to {Nash} equilibria in
  general stochastic games.
\newblock In \emph{NeurIPS '22: Proceedings of the 36th International
  Conference on Neural Information Processing Systems}, 2022.

\bibitem[Gidel et~al.(2019)Gidel, Berard, Vignoud, Vincent, and
  Lacoste-Julien]{GBVV+19}
Gidel, G., Berard, H., Vignoud, G., Vincent, P., and Lacoste-Julien, S.
\newblock A variational inequality perspective on generative adversarial
  networks.
\newblock In \emph{ICLR '19: Proceedings of the 2019 International Conference
  on Learning Representations}, 2019.

\bibitem[Hall \& Heyde(1980)Hall and Heyde]{HH80}
Hall, P. and Heyde, C.~C.
\newblock \emph{Martingale Limit Theory and Its Application}.
\newblock Probability and Mathematical Statistics. Academic Press, New York,
  1980.

\bibitem[Hart \& Mas-Colell(2003)Hart and Mas-Colell]{HMC03}
Hart, S. and Mas-Colell, A.
\newblock Uncoupled dynamics do not lead to {Nash} equilibrium.
\newblock \emph{American Economic Review}, 93\penalty0 (5):\penalty0
  1830--1836, 2003.

\bibitem[Hart \& Mas-Colell(2006)Hart and Mas-Colell]{HMC06}
Hart, S. and Mas-Colell, A.
\newblock Stochastic uncoupled dynamics and {Nash} equilibrium.
\newblock \emph{Games and Economic Behavior}, 57:\penalty0 286--303, 2006.

\bibitem[H{\'e}liou et~al.(2017)H{\'e}liou, Cohen, and Mertikopoulos]{HCM17}
H{\'e}liou, A., Cohen, J., and Mertikopoulos, P.
\newblock Learning with bandit feedback in potential games.
\newblock In \emph{NIPS '17: Proceedings of the 31st International Conference
  on Neural Information Processing Systems}, 2017.

\bibitem[H{\'e}liou et~al.(2020)H{\'e}liou, Mertikopoulos, and Zhou]{HMZ20}
H{\'e}liou, A., Mertikopoulos, P., and Zhou, Z.
\newblock Gradient-free online learning in continuous games with delayed
  rewards.
\newblock In \emph{ICML '20: Proceedings of the 37th International Conference
  on Machine Learning}, 2020.

\bibitem[Hiriart-Urruty \& Lemar{\'e}chal(2001)Hiriart-Urruty and
  Lemar{\'e}chal]{HUL01}
Hiriart-Urruty, J.-B. and Lemar{\'e}chal, C.
\newblock \emph{Fundamentals of Convex Analysis}.
\newblock Springer, Berlin, 2001.

\bibitem[Hofbauer \& Sandholm(2002)Hofbauer and Sandholm]{HS02}
Hofbauer, J. and Sandholm, W.~H.
\newblock On the global convergence of stochastic fictitious play.
\newblock \emph{Econometrica}, 70\penalty0 (6):\penalty0 2265--2294, November
  2002.

\bibitem[Hofbauer \& Sigmund(2003)Hofbauer and Sigmund]{HS03}
Hofbauer, J. and Sigmund, K.
\newblock Evolutionary game dynamics.
\newblock \emph{Bulletin of the American Mathematical Society}, 40\penalty0
  (4):\penalty0 479--519, July 2003.

\bibitem[Hsieh et~al.(2019)Hsieh, Iutzeler, Malick, and Mertikopoulos]{HIMM19}
Hsieh, Y.-G., Iutzeler, F., Malick, J., and Mertikopoulos, P.
\newblock On the convergence of single-call stochastic extra-gradient methods.
\newblock In \emph{NeurIPS '19: Proceedings of the 33rd International
  Conference on Neural Information Processing Systems}, pp.\  6936--6946, 2019.

\bibitem[Hsieh et~al.(2020)Hsieh, Iutzeler, Malick, and Mertikopoulos]{HIMM20}
Hsieh, Y.-G., Iutzeler, F., Malick, J., and Mertikopoulos, P.
\newblock Explore aggressively, update conservatively: {Stochastic}
  extragradient methods with variable stepsize scaling.
\newblock In \emph{NeurIPS '20: Proceedings of the 34th International
  Conference on Neural Information Processing Systems}, 2020.

\bibitem[Hsieh et~al.(2021)Hsieh, Mertikopoulos, and Cevher]{HMC21}
Hsieh, Y.-P., Mertikopoulos, P., and Cevher, V.
\newblock The limits of min-max optimization algorithms: {Convergence} to
  spurious non-critical sets.
\newblock In \emph{ICML '21: Proceedings of the 38th International Conference
  on Machine Learning}, 2021.

\bibitem[Juditsky et~al.(2011)Juditsky, Nemirovski, and Tauvel]{JNT11}
Juditsky, A., Nemirovski, A.~S., and Tauvel, C.
\newblock Solving variational inequalities with stochastic mirror-prox
  algorithm.
\newblock \emph{Stochastic Systems}, 1\penalty0 (1):\penalty0 17--58, 2011.

\bibitem[Kelly et~al.(1998)Kelly, Maulloo, and Tan]{KMT98}
Kelly, F.~P., Maulloo, A.~K., and Tan, D. K.~H.
\newblock Rate control for communication networks: shadow prices, proportional
  fairness and stability.
\newblock \emph{Journal of the Operational Research Society}, 49\penalty0
  (3):\penalty0 237--252, March 1998.

\bibitem[Korpelevich(1976)]{Kor76}
Korpelevich, G.~M.
\newblock The extragradient method for finding saddle points and other
  problems.
\newblock \emph{{\`E}konom. i Mat. Metody}, 12:\penalty0 747--756, 1976.

\bibitem[Kushner \& Yin(1997)Kushner and Yin]{KY97}
Kushner, H.~J. and Yin, G.~G.
\newblock \emph{Stochastic approximation algorithms and applications}.
\newblock Springer-Verlag, New York, NY, 1997.

\bibitem[Lattimore \& Szepesv{\'a}ri(2020)Lattimore and Szepesv{\'a}ri]{LS20}
Lattimore, T. and Szepesv{\'a}ri, C.
\newblock \emph{Bandit Algorithms}.
\newblock Cambridge University Press, Cambridge, UK, 2020.

\bibitem[Leslie \& Collins(2005)Leslie and Collins]{LC05}
Leslie, D.~S. and Collins, E.~J.
\newblock Individual {$Q$}-learning in normal form games.
\newblock \emph{SIAM Journal on Control and Optimization}, 44\penalty0
  (2):\penalty0 495--514, 2005.

\bibitem[Leslie \& Collins(2006)Leslie and Collins]{LC06}
Leslie, D.~S. and Collins, E.~J.
\newblock Generalised weakened fictitious play.
\newblock \emph{Games and Economic Behavior}, 56\penalty0 (2):\penalty0
  285--298, August 2006.

\bibitem[Littlestone \& Warmuth(1994)Littlestone and Warmuth]{LW94}
Littlestone, N. and Warmuth, M.~K.
\newblock The weighted majority algorithm.
\newblock \emph{Information and Computation}, 108\penalty0 (2):\penalty0
  212--261, 1994.

\bibitem[{Maynard Smith} \& Price(1973){Maynard Smith} and Price]{MP73}
{Maynard Smith}, J. and Price, G.~R.
\newblock The logic of animal conflict.
\newblock \emph{Nature}, 246:\penalty0 15--18, November 1973.

\bibitem[Mertikopoulos \& Sandholm(2016)Mertikopoulos and Sandholm]{MS16}
Mertikopoulos, P. and Sandholm, W.~H.
\newblock Learning in games via reinforcement and regularization.
\newblock \emph{Mathematics of Operations Research}, 41\penalty0 (4):\penalty0
  1297--1324, November 2016.

\bibitem[Mertikopoulos \& Zhou(2019)Mertikopoulos and Zhou]{MZ19}
Mertikopoulos, P. and Zhou, Z.
\newblock Learning in games with continuous action sets and unknown payoff
  functions.
\newblock \emph{Mathematical Programming}, 173\penalty0 (1-2):\penalty0
  465--507, January 2019.

\bibitem[Mertikopoulos et~al.(2018)Mertikopoulos, Papadimitriou, and
  Piliouras]{MPP18}
Mertikopoulos, P., Papadimitriou, C.~H., and Piliouras, G.
\newblock Cycles in adversarial regularized learning.
\newblock In \emph{SODA '18: Proceedings of the 29th annual ACM-SIAM Symposium
  on Discrete Algorithms}, 2018.

\bibitem[Mertikopoulos et~al.(2019)Mertikopoulos, Lecouat, Zenati, Foo,
  Chandrasekhar, and Piliouras]{MLZF+19}
Mertikopoulos, P., Lecouat, B., Zenati, H., Foo, C.-S., Chandrasekhar, V., and
  Piliouras, G.
\newblock Optimistic mirror descent in saddle-point problems: {Going} the extra
  (gradient) mile.
\newblock In \emph{ICLR '19: Proceedings of the 2019 International Conference
  on Learning Representations}, 2019.

\bibitem[Mertikopoulos et~al.(2020)Mertikopoulos, Hallak, Kavis, and
  Cevher]{MHKC20}
Mertikopoulos, P., Hallak, N., Kavis, A., and Cevher, V.
\newblock On the almost sure convergence of stochastic gradient descent in
  non-convex problems.
\newblock In \emph{NeurIPS '20: Proceedings of the 34th International
  Conference on Neural Information Processing Systems}, 2020.

\bibitem[Monderer \& Shapley(1996)Monderer and Shapley]{MS96}
Monderer, D. and Shapley, L.~S.
\newblock Potential games.
\newblock \emph{Games and Economic Behavior}, 14\penalty0 (1):\penalty0 124 --
  143, 1996.

\bibitem[Nemirovski \& Yudin(1983)Nemirovski and Yudin]{NY83}
Nemirovski, A.~S. and Yudin, D.~B.
\newblock \emph{Problem Complexity and Method Efficiency in Optimization}.
\newblock Wiley, New York, NY, 1983.

\bibitem[Nesterov(2009)]{Nes09}
Nesterov, Y.
\newblock Primal-dual subgradient methods for convex problems.
\newblock \emph{Mathematical Programming}, 120\penalty0 (1):\penalty0 221--259,
  2009.

\bibitem[Nevel'son \& Khasminskii(1976)Nevel'son and Khasminskii]{NK76}
Nevel'son, M.~B. and Khasminskii, R.~Z.
\newblock \emph{Stochastic Approximation and Recursive Estimation}.
\newblock American Mathematical Society, Providence, RI, 1976.

\bibitem[Oliveira et~al.(2019)Oliveira, Rodrigues, Kristi{\'c}, and Ba{\c
  s}ar]{ORKB19}
Oliveira, T.~R., Rodrigues, V. H.~P., Kristi{\'c}, M., and Ba{\c s}ar, T.
\newblock Nash equilibrium seeking with arbitrarily delayed player actions.
\newblock In \emph{CDC '20: Proceedings of the 59th IEEE Annual Conference on
  Decision and Control}, 2019.

\bibitem[Polyak(1987)]{Pol87}
Polyak, B.~T.
\newblock \emph{Introduction to Optimization}.
\newblock Optimization Software, New York, NY, USA, 1987.

\bibitem[Popov(1980)]{Pop80}
Popov, L.~D.
\newblock A modification of the {Arrow}\textendash{Hurwicz} method for search
  of saddle points.
\newblock \emph{Mathematical Notes of the Academy of Sciences of the USSR},
  28\penalty0 (5):\penalty0 845--848, 1980.

\bibitem[Rakhlin \& Sridharan(2013)Rakhlin and Sridharan]{RS13-NIPS}
Rakhlin, A. and Sridharan, K.
\newblock Optimization, learning, and games with predictable sequences.
\newblock In \emph{NIPS '13: Proceedings of the 27th International Conference
  on Neural Information Processing Systems}, 2013.

\bibitem[Ratliff et~al.(2016)Ratliff, Burden, and Sastry]{RBS16}
Ratliff, L.~J., Burden, S.~A., and Sastry, S.~S.
\newblock On the characterization of local {Nash} equilibria in continuous
  games.
\newblock \emph{{IEEE} Trans. Autom. Control}, 61\penalty0 (8):\penalty0
  2301--2307, August 2016.

\bibitem[Robbins \& Monro(1951)Robbins and Monro]{RM51}
Robbins, H. and Monro, S.
\newblock A stochastic approximation method.
\newblock \emph{Annals of Mathematical Statistics}, 22:\penalty0 400--407,
  1951.

\bibitem[Robinson(1951)]{Rob51}
Robinson, J.
\newblock An iterative method for solving a game.
\newblock \emph{Annals of Mathematics}, 54:\penalty0 296--301, 1951.

\bibitem[Rosen(1965)]{Ros65}
Rosen, J.~B.
\newblock Existence and uniqueness of equilibrium points for concave
  ${N}$-person games.
\newblock \emph{Econometrica}, 33\penalty0 (3):\penalty0 520--534, 1965.

\bibitem[Rosenthal(1973)]{Ros73}
Rosenthal, R.~W.
\newblock A class of games possessing pure-strategy {Nash} equilibria.
\newblock \emph{International Journal of Game Theory}, 2:\penalty0 65--67,
  1973.

\bibitem[Samuelson \& Zhang(1992)Samuelson and Zhang]{SZ92}
Samuelson, L. and Zhang, J.
\newblock Evolutionary stability in asymmetric games.
\newblock \emph{Journal of Economic Theory}, 57:\penalty0 363--391, 1992.

\bibitem[Scutari et~al.(2010)Scutari, Facchinei, Palomar, and Pang]{SFPP10}
Scutari, G., Facchinei, F., Palomar, D.~P., and Pang, J.-S.
\newblock Convex optimization, game theory, and variational inequality theory
  in multiuser communication systems.
\newblock \emph{{IEEE} Signal Process. Mag.}, 27\penalty0 (3):\penalty0 35--49,
  May 2010.

\bibitem[Shalev-Shwartz \& Singer(2006)Shalev-Shwartz and Singer]{SSS06}
Shalev-Shwartz, S. and Singer, Y.
\newblock Convex repeated games and {Fenchel} duality.
\newblock In \emph{NIPS' 06: Proceedings of the 19th Annual Conference on
  Neural Information Processing Systems}, pp.\  1265--1272. MIT Press, 2006.

\bibitem[Spall(1992)]{Spa92}
Spall, J.~C.
\newblock Multivariate stochastic approximation using a simultaneous
  perturbation gradient approximation.
\newblock \emph{{IEEE} Trans. Autom. Control}, 37\penalty0 (3):\penalty0
  332--341, March 1992.

\bibitem[Syrgkanis et~al.(2015)Syrgkanis, Agarwal, Luo, and Schapire]{SALS15}
Syrgkanis, V., Agarwal, A., Luo, H., and Schapire, R.~E.
\newblock Fast convergence of regularized learning in games.
\newblock In \emph{NIPS '15: Proceedings of the 29th International Conference
  on Neural Information Processing Systems}, pp.\  2989--2997, 2015.

\bibitem[Tatarenko \& Kamgarpour(2019{\natexlab{a}})Tatarenko and
  Kamgarpour]{TK19}
Tatarenko, T. and Kamgarpour, M.
\newblock Learning generalized {Nash} equilibria in a class of convex games.
\newblock \emph{{IEEE} Trans. Autom. Control}, 64\penalty0 (4):\penalty0
  1426--1439, 2019{\natexlab{a}}.

\bibitem[Tatarenko \& Kamgarpour(2019{\natexlab{b}})Tatarenko and
  Kamgarpour]{TK19-CDC}
Tatarenko, T. and Kamgarpour, M.
\newblock Learning {Nash} equilibria in monotone games.
\newblock In \emph{CDC '19: Proceedings of the 58th IEEE Annual Conference on
  Decision and Control}, 2019{\natexlab{b}}.
\newblock \doi{10.1109/CDC40024.2019.9029659}.

\bibitem[Taylor \& Jonker(1978)Taylor and Jonker]{TJ78}
Taylor, P.~D. and Jonker, L.~B.
\newblock Evolutionary stable strategies and game dynamics.
\newblock \emph{Mathematical Biosciences}, 40\penalty0 (1-2):\penalty0
  145--156, 1978.

\bibitem[Tullock(1980)]{Tul80}
Tullock, G.
\newblock Efficient rent seeking.
\newblock In Tollison, J. M. B. R.~D. and Tullock, G. (eds.), \emph{Toward a
  theory of the rent-seeking society}. Texas A\&M University Press, 1980.

\bibitem[Vovk(1990)]{Vov90}
Vovk, V.~G.
\newblock Aggregating strategies.
\newblock In \emph{COLT '90: Proceedings of the 3rd Workshop on Computational
  Learning Theory}, pp.\  371--383, 1990.

\bibitem[Zhang et~al.(2021)Zhang, Ren, and Li]{ZRL21}
Zhang, R., Ren, Z., and Li, N.
\newblock Gradient play in multi-agent {Markov} stochastic games: {Stationary}
  points and convergence.
\newblock \url{https://arxiv.org/abs/2106.00198}, 2021.

\end{thebibliography}
